\documentclass[12pt]{article} 
\usepackage[sectionbib]{natbib}
\usepackage{array,epsfig,fancyheadings,rotating}
\usepackage[]{hyperref}  

\usepackage{amssymb, amsmath}
\usepackage{ulem}
\usepackage{amsfonts}
\usepackage{graphicx}
\usepackage{setspace} 
\usepackage[utf8]{inputenc}
\usepackage{caption}
\usepackage{subcaption}
\usepackage{amsmath}
\usepackage[]{hyperref}
\usepackage{amsmath}
\usepackage{amssymb}
\usepackage{amsfonts}
\usepackage{multirow}
\usepackage{xcolor}
\usepackage{amsthm}
\usepackage{lscape}
\usepackage{afterpage}
\usepackage{enumitem}
\usepackage{booktabs}
\usepackage{arydshln}
\usepackage{float}
\newenvironment{supplementary}[1][Supplementary Material]{%
    \clearpage
    \section*{Supplementary Materials for \\ ``INSPECTING DISCREPANCY
BETWEEN MULTIVARIATE DISTRIBUTIONS USING
HALF-SPACE DEPTH-BASED INFORMATION CRITERIA''}
    \addcontentsline{toc}{section}{#1} 
    \renewcommand{\thesection}{S\arabic{section}} 
    \setcounter{section}{0} 
    \setcounter{figure}{0} 
    \setcounter{table}{0} 
    \renewcommand{\thefigure}{S\arabic{figure}} 
    \renewcommand{\thetable}{S\arabic{table}}   
}{%
    \clearpage
}

\usepackage{sectsty, secdot}
\sectionfont{\fontsize{12}{14pt plus.8pt minus .6pt}\selectfont}
\renewcommand{\theequation}{\thesection\arabic{equation}}
\subsectionfont{\fontsize{12}{14pt plus.8pt minus .6pt}\selectfont}

\usepackage{amsmath}
\usepackage{amssymb}
\usepackage{amsfonts}
\usepackage{multirow}
\usepackage{amsthm}

\theoremstyle{definition}

\theoremstyle{definition}
\newtheorem{theorem}{Theorem}

\newtheorem{proposition}{Proposition}
\newtheorem{problem}{Problem}

\newtheorem{remark}{\textit{Remark}}

\newcommand{\bA}{\mathbf{A}}

\newcommand{\bI}{\mathbf{I}}

\newcommand{\bX}{\mathbf{X}}
\newcommand{\bY}{\mathbf{Y}}

\newcommand{\bb}{\mathbf{b}}

\newcommand{\bt}{\mathbf{t}}
\newcommand{\bu}{\mathbf{u}}
\newcommand{\bv}{\mathbf{v}}

\newcommand{\bx}{\mathbf{x}}
\newcommand{\by}{\mathbf{y}}
\newcommand{\bz}{\mathbf{z}}
\newcommand{\zero}{\mathbf{0}}

\newcommand{\sA}{\mathcal{A}}

\newcommand{\sC}{\mathcal{C}}
\newcommand{\sD}{\mathcal{D}}
\newcommand{\sE}{\mathcal{E}}
\newcommand{\sF}{\mathcal{F}}
\newcommand{\sG}{\mathcal{G}}
\newcommand{\sH}{\mathcal{H}}

\newcommand{\sJ}{\mathcal{J}}
\newcommand{\sK}{\mathcal{K}}
\newcommand{\sL}{\mathcal{L}}

\newcommand{\sR}{\mathcal{R}}
\newcommand{\sS}{\mathcal{S}}
\newcommand{\sT}{\mathcal{T}}
\newcommand{\sU}{\mathcal{U}}
\newcommand{\sV}{\mathcal{V}}

\newcommand{\sX}{\mathcal{X}}
\newcommand{\sY}{\mathcal{Y}}
\newcommand{\sZ}{\mathcal{Z}}

\newcommand{\bbP}{\mathbb{P}}
\newcommand{\bbQ}{\mathbb{Q}}
\newcommand{\bbR}{\mathbb{R}}

\newcommand{\E}{\mathbb{E}}
\newcommand{\Var}{\text{Var}}
\newcommand{\cov}{\text{cov}}
\newcommand{\tp}{\text{T}}

\newcommand{\balpha}{\boldsymbol{\alpha}}

\newcommand{\bmu}{\boldsymbol{\mu}}

\newcommand{\blambda}{\boldsymbol{\lambda}}

\newcommand{\btau}{\boldsymbol{\tau}}

\newcommand{\bSigma}{\boldsymbol{\Sigma}}

\newcommand{\bDelta}{\boldsymbol{\Delta}}
\newcommand{\bPhi}{\boldsymbol{\Phi}}

\newcommand{\bOmega}{\boldsymbol{\Omega}}

\newcommand{\bPsi}{\boldsymbol{\Psi}}

\newcommand\numberthis{\addtocounter{equation}{1}\tag{\theequation}}

\newcommand{\DDD}{\text{DDD}}
\newcommand{\KS}{\text{KS}}
\newcommand{\CvM}{\text{CvM}}
\newcommand{\IF}{\text{IF}}

\begin{document}
	
	
	\renewcommand{\baselinestretch}{2}
	
	\markright{ \hbox{\footnotesize\rm 
		}\hfill\\[-13pt]
		\hbox{\footnotesize\rm
		}\hfill }
	
	\markboth{\hfill{\footnotesize\rm \MakeUppercase{Pratim Guha Niyogi} AND \MakeUppercase{Subhra Sankar Dhar}} \hfill}
	{\hfill {\footnotesize\rm \MakeUppercase{Discrepancy among multivariate distributions}} \hfill}
	
	\renewcommand{\thefootnote}{}
	$\ $\par
	
	
	\fontsize{12}{14pt plus.8pt minus .6pt}\selectfont \vspace{0.8pc}
	\centerline{\large\bf \MakeUppercase{Inspecting discrepancy} }
	\vspace{2pt} 
	\centerline{\large\bf \MakeUppercase{between multivariate distributions using}}
        \vspace{2pt} 
	\centerline{\large\bf \MakeUppercase{half-space depth-based information criteria}}
	\vspace{.4cm} 
	\centerline{Pratim Guha Niyogi and Subhra Sankar Dhar} 
	\vspace{.4cm} 
	\centerline{\it Johns Hopkins University and Indian Institute of Technology, Kanpur}
	\vspace{.55cm} \fontsize{9}{11.5pt plus.8pt minus.6pt}\selectfont
	
	
	\begin{quotation}
		\noindent {\it Abstract:}
This article inspects whether a multivariate distribution is different from a specified distribution or not, and it also tests the equality of two multivariate distributions. In the course of this study, a graphical tool-kit using well-known half-space depth-based information criteria is proposed, which is a two-dimensional plot, regardless of the dimension of the data, and it is even useful in comparing high-dimensional distributions.
The simple interpretability of the proposed graphical tool-kit motivates us to formulate test statistics to carry out the corresponding testing of hypothesis problems.
It is established that the proposed tests based on the same information criteria are consistent, and moreover, the asymptotic distributions of the test statistics under contiguous/local alternatives are derived, which enables us to compute the asymptotic power of these tests.
Furthermore, it is observed that the computations associated with the proposed tests are unburdensome.
Besides, these tests perform better than many other tests available in the literature when data are generated from various distributions such as heavy-tailed distributions, which indicates that the proposed methodology is robust as well. Finally, the usefulness of the proposed graphical tool-kit and tests is shown on two benchmark real data sets.

		\vspace{9pt}
		\par
	\end{quotation}\par

	\def\thefigure{\arabic{figure}}
	\def\thetable{\arabic{table}}
	
	\renewcommand{\theequation}{\thesection.\arabic{equation}}

	\fontsize{12}{14pt plus.8pt minus .6pt}\selectfont

\section{Introduction}
In many scientific investigations, data are collected for multiple variables, and therefore, some techniques are needed to visualize and compare multivariate observations.
In this spirit, this article addresses the problem of goodness-of-fit to validate the assumption of the underlying multivariate distribution, which is commonly encountered in many statistical data analyses. Moreover, we propose statistical tests to compare two multivariate probability distributions along with appropriate graphical tool-kits. It is needless to mention that such testing of hypothesis problems have a plethora of applications. For example, to understand the dynamics of the human physical activity patterns, the distribution of diurnal activity and the rest period is assumed to be a double Pareto distribution, and in order to validate it, a goodness-of-fit test is needed \citep{paraschiv2013unraveling}. In psychology, the study of maternal depression, and examination of maternal and infant sleep throughout the first year of life is an example of a two-sample testing problem \citep{newland2016goodness}. 
\par
Let us now formulate the problem with notation. Consider the data $\sX = \{\bX_{1}, \cdots, \bX_{n}\}$ associated with unknown distribution $F$, and $F_{0}$ is the specified distribution function on $\bbR^{d} ~ ( d\geq 1)$.
We now want to test that $H_{0}: F = F_{0}$ against $H_{1}: F \neq F_{0}$, and a few more technical assumptions on $F$ and $F_{0}$ will be stated at appropriate places. 
Next, the two-sample problem for comparing multivariate distributions can also be formulated in the following way. 
Suppose that $\sX = \{ \bX_{1}, \cdots, \bX_{n}\}$ and $\sY = \{ \bY_{1}, \cdots, \bY_{m}\}$ are two independent data sets associated with two unknown multivariate distributions $F$ and $G$, respectively. Here, we now want to test $H_{0}: F = G$ against $H_{1} : F \neq G$. For $d=1$, i.e., in the case of univariate data, the aforesaid problems have already been investigated in many articles. For example, the readers may refer to Kolmogorov-Smirnov (KS) \citep{daniel1990kolmogorov}, Cram\'er-von Mises (CvM) \citep{anderson1962distribution},
Anderson-Darling (AD) \citep{anderson1952asymptotic} and a few more tests. In this context, for comparing univariate distributions, a few graphical tool-kits have also been used in literature for a long period of time such as 
probability-probability plot \citep{michael1983stabilized}
quantile-quantile plot \citep{gnanadesikan2011methods, wilk1968probability, chambers2018graphical}, Lorenz curve \citep{lorenz1905methods}. 
\par
For $d >1$, one may look at AD test \citep{paulson1987some}, CvM test \citep{koziol1982class}, 
Doornik-Hansen (DH) test \citep{doornik2008omnibus}, 
Henze-Zirkler (HZ) test \citep{henze1990class}, 
Royston (R) test \citep{royston1992approximating}
and a series of $\chi^{2}$-type tests (for example McCulloch (M), Nikulin-Rao-Robson (NRR), Dzhaparidze-Nikulin (DN) \citep{voinov2016new}). These tests validate whether the data are obtained from a multivariate normal distribution or not. 
For two-sample multivariate problems, similar studies have been investigated by \citet{chen2017new} and see a few references therein. Besides, like univariate data, there have been a few attempts to compare two multivariate distributions in two-dimensional plots (see for example, \citet{friedman1981graphics, 
dhar2014comparison}).
\par
In this article, we also study similar hypotheses problems based on the difference of some functions based on the center-outward ranking \citep{tukey1975mathematics}, which is known as data-depth (for details see \citet{liu1999multivariate}), and the proposed test procedure is called data-depth discrepancy (DDD). The main idea of the the discrepancy measure DDD is that it is large if and only if, the null hypothesized assumption on the parent distribution is likely to be false (in goodness-of-fit test) or the samples are likely to be from different distributions (in two-sample test).
Strictly speaking, in this work, the DDD is defined based on the $L_{2}$ and $L_{\infty}$ distances between two relevant data-depth functions for KS and CvM test statistics, respectively. However, to measure the DDD, one may use other suitable distance functions (such as $L_{p}, p \geq 1$ distance) in principle.
\par
This article proposes DDD based on Tukey's half-space depth \citep{tukey1975mathematics} since under some regularity conditions, it uniquely determines a uniformly absolutely continuous distribution with compact support under minimal assumptions \citep{koshevoy2002tukey, hassairi2008tukey}.
In addition, this can be computed using freely available packages in various statistical software such as R. In order to carry out the test, we replace the distribution function by half space depth in the KS and CvM type test statistics for the goodness-fit and the two-sample tests, respectively. To summarize, the following are the main contributions of this article. First, we define a discrepancy measure based on half-space depth information criteria, which can be used to create an alternative graphical tool-kit to visualize the disparity of two distribution functions. Second, we propose two test statistics based on the proposed discrepancy measure for both goodness-of-fit and two-sample testing procedures. Third, most of the existing test statistics are available for the normality test only; our testing procedure provides a computationally feasible unified solution for any underlying distribution for any dimension. Fourth, we have shown that the proposed tests are consistent. Moreover, the asymptotic distributions of the test statistics under the contiguous alternative are derived, which enables us to compute the asymptotic power of the tests under contiguous alternatives. Fifth, a new graphical tool-kit based on the information criteria is introduced here with proper theoretical justification. The diagram related to our proposed graphical tool-kit lies on the two-dimensional plane irrespective of the dimension of the data. This fact enables us to visualize features of fairly large-dimensional data using the proposed graphical tool-kit.
\par
The rest of the paper is organized as follows. In Section \ref{sec:depth}, we briefly review the different well-known notions of data-depth, the usefulness of Tukey's half-space depth, and a graphical display, which is a well-known depth-depth plot. Section \ref{sec:method} is dedicated to the proposed methods in statistical testing for goodness-of-fit and two-sample situations. First, we provide the heuristic establishment of the graphical tool-kit, and in the later part of Section \ref{sec:method}, we proceed with the formal definition of DDD and conclude with two multivariate tests based on the data-depth. In Section \ref{sec:technical}, we investigate the asymptotic properties of the proposed testing procedures. Section \ref{sec:conclusion} concludes the article with a discussion. The online supplementary material contains the finite sample performance and implementation of the proposed test on two interesting data sets. In addition, technical details and mathematical proofs of the proposed theorems in Section \ref{sec:technical} are discussed in the supplementary material.

\section{Some preliminaries}
\label{sec:depth}
In the pioneering work by John W. Tukey in 1975, he introduced a novel way to describe the data cloud. 
For a given multivariate data cloud $\sX = \{\bX_{1}, \cdots, \bX_{n}\}$, a point $\bx$ in the same Euclidean space becomes the representative of $\sX$ through the function $D_{\sX}(\bx)$, which measures how `close' $\bx$ is to the center of $\sX$. 
Mathematically, a function, viz., depth function, will be bounded and non-negative function of the form $D: \bbR^{d} \times \sF \rightarrow \bbR$,  where $\sF$ is the class of all distributions on $\bbR^{d}$. In most of the versions of the data-depth, the functions are affine invariant in the scenes $D_{F_{\bA\bx+\bb}}(\bA\bx + \bb) = D_{F_{\bx}}(\bx)$ for any $d\times d$ matrix $\bA$ and a vector $\bb \in \bbR^{d}$, and it maximizes somewhere in the center of the data cloud and vanishes to infinity. 
\par
Half-space depth is one of the depth functions which does not impose any moment conditions of the data, and it can characterize a certain family of distributions \citep{koshevoy2003lift, hassairi2008tukey, CUESTAALBERTOS, KongZuo}. The technical definition of the half-space depth is as follows. Let $F$ be a probability distribution on $\bbR^{d}$ for $d \geq 1$, and $\sH$ be the class of closed half-spaces $H$ in $\bbR^{d}$. The Tukey's depth (or half-space depth) of a point $ \bx \in \bbR^{d}$ with respect to $F$ is defined by $D_{F}(\bx) = \inf \left\lbrace F(H) : H \in \sH, \bx \in H \right\rbrace$.
\par
One may also look at $D_{F}(\cdot)$ in the following way. 
Let $\sS^{d-1} = \{\bu \in \bbR^{d}: \|\bu\| = 1\}$ be the unit sphere of $\bbR^{d}$, then for $\bu \in \sS^{d-1}$ and $\bx \in \bbR^{d}$, we consider the closed half-space passing through  $\bx \in \bbR^{d}$ as $H[\bx, \bu] = \{\by \in \bbR^{d}: \bu^{\tp}\by\leq\bu^{\tp}\bx\}$. 
Now, Tukey's half-space depth with respect to the distribution function $F$ can be defined as $D_{F}(\bx) = \inf_{\bu \in \sS^{d-1}}F(H[\bx, \bu])$. The sample version of $D_{F}(\bx)$, which is denoted as $D_{\sX}(\bx)$,  based on the empirical distribution $\widehat{F}_{n}$ of the sample $\sX = \left\lbrace \bX_{1}, \cdots, \bX_{n}\right\rbrace $ and is defined as $D_{\sX}(\bx) = \min_{\bu \in \sS^{d-1}}\frac{1}{n}\sum_{i=1}^{n}\textbf{1}\{ \bu^{\tp}\bX_{i} \leq \bu^{\tp}\bx\}$. Observe that, for $d=1$, $D_{F}(x) = \min\{ F(x), S(x) \}$, where $S(x) = 1 - F(x)$. Since $\widehat{F}_{n}(x)$ is the right-continuous empirical cumulative distribution function, 
we define $\widehat{S}_{n}(x) = 1-\widehat{F}_{n}(x^{-})$, and consequently, the sample version of half-space depth for univariate data becomes $D_{\sX}(x) = \min\{\widehat{F}_{n}(x), \widehat{S}_{n}(x) \}$. Based on this definition, it can easily be observed that the depth function achieves the maximum at the median of the distribution and monotonically decays to zero when $x$ deviates from the median. \citet{masse2004asymptotics} provided some results on asymptotics for the half-space depth process, and \citet{rousseeuw1996algorithm,ruts1996computing,rousseeuw1998computing} pointed out the computational aspects of the half-space depth.
\par
As a graphical tool-kit, the depth-depth (DD) plot \citep{liu1999multivariate} is a well-known tool to compare two multivariate distributions graphically. 
For the sake of completeness, we here briefly describe the methodology of the DD-plot. 
Let $F$ and $G$ be two distributions on $\bbR^{d}$, and let $D(\cdot)$ be an affine-invariant depth (for example, half-space depth). 
Define the population version of DD-plot, $DD(F,G) = \left\lbrace (D_{F}(\bx) , D_{G}(\bx)) \; \text{for all} \; \bx \in \bbR^{d} \right\rbrace$. If the two distributions are identical, then these graphs are segments of the diagonal line from (0,0) to (1,1) on the $\bbR^{2}$ plane. Plots that deviate from the straight line indicate the difference of two underlying distributions. For goodness-of-fit problem, if the underlying distribution $F$ is unknown, for a given sample ${\cal{X}} = \left\lbrace \bX_{1}, \cdots, \bX_{n} \right\rbrace$, we may determine whether $F$ is some specified distribution, say $F_{0}$, by examining the sample version of DD-plot $DD(\sX, F_{0}) = \left\lbrace (D_{F_{n}}(\bx) , D_{F_{0}}(\bx)) \; \text{for all} \; \bx \in  {\cal{X}}\rbrace \right\rbrace$,
where $F_{n}$ is the empirical distribution function based on ${\cal{X}}$. 
For two-sample problem, if $F$ and $G$ are the population distributions for the sample $\sX = \left\lbrace \bX_{1}, \cdots, \bX_{n} \right\rbrace$ and $\sY = \left\lbrace \bY_{1}, \cdots, \bY_{m} \right\rbrace$, respectively, one can use the DD-plot (i.e., the sample version) to determine whether or not the two distributions are identical; $DD(\widehat{F}_{n},\widehat{G}_{m}) = \left\lbrace (D_{\widehat{F}_{n}}(\bx), D_{\widehat{G}_{m}}(\bx)) \; \text{for all} \; \bx \in  \left\lbrace \sX \cup \sY \right\rbrace \right\rbrace,$ where $\hat{F}_{n}$ and $\hat{G}_{m}$ are empirical distribution functions based on ${\cal{X}}$ and ${\cal{Y}}$, respectively. 
It is easy to check that the DD-plot is affine-invariant if the associated depth measures are so be. On the one hand, if $d=1$, the Lebesgue measure of $DD$ vanishes when $F \neq G$, and on the other hand, for $d>1$ and $F$ and $G$ are absolutely continuous, $DD$ is a region with non-zero area.
Moreover, \citet{liu1999multivariate} showed that, for different distributional differences, such as location, scale, skewness, and kurtosis, different patterns are observed in the DD-plot.

\section{Data-depth discrepancy and associated statistical tests}
\label{sec:method}
In this section, we propose a data-depth discrepancy (DDD) describing a graphical tool-kit and associated tests. 
The DDD is defined in terms of the difference of the depth functions, and we begin this discussion in Section \ref{subsec:ddd} for any probability distribution. We conclude this section with the motivation for defining a graphical tool-kit to visualize the dissimilarities of the distributions. 
We introduce new multivariate goodness-of-fit and two-sample tests based on half-space depth in Section \ref{subsec:test}. 

\subsection{Data-depth discrepancy (DDD)}
\label{subsec:ddd}
Our goal is to construct statistical tests that can answer the following questions. 
\begin{problem}(One-sample/Goodness-of-fit multivariate problem) 
\label{problem:one-sample}
For $d \geq 1$, let $\bX$ be a $d$-dimensional random variable with distribution function $F$ on $\bbR^{d}$, where $\bX = (X_{1}, \cdots, X_{d})^{\tp}$. 
Suppose that $F_{0}$ is a pre-specified distribution function. 
Given observations $\sX = \{\bX_{1}, \cdots, \bX_{n}\}$ independent and identically distributed (i.i.d.) from $F$, can we decide whether $F \neq F_{0}$?
\end{problem}
\begin{problem}(Two-sample multivariate problem)
\label{problem:two-sample}
For $d \geq 1$, let $\bX$ and $\bY$ be two random variables with distribution functions $F$ and $G$, respectively, where $\bX = (X_{1}, \cdots, X_{d})^{\tp}$ and $\bY = (Y_{1}, \cdots, Y_{d})^{\tp}$. 
Given the two independent data sets $\sX = \{ \bX_{1}, \cdots, \bX_{n} \}$ and $\sY = \{ \bY_{1}, \cdots, \bY_{m} \}$ independently and identically distributed from $F$ and $G$, respectively, can we decide whether $F \neq G$?
\end{problem}
To answer the above-mentioned questions, we propose the data-depth discrepancy (DDD) between two distributions $F$ and $G$ at $\bx \in \bbR^{d}$ based on the data-depth as follows. 
\begin{equation}
\label{eq:ddd}
\DDD(\bx; F, G) = D_{F}(\bx) - D_{G}(\bx)
\end{equation}
The sample version of these measures can be obtained by replacing $D$ with its empirical data-depth values based on the specific sample sizes. 
Moreover, the discrepancy between the empirical and theoretical distribution is defined as $\DDD(\bx; \sX, F_{0}) = D_{\sX}(\bx) - D_{F_{0}}(\bx)$. 
Note that, $\DDD(\cdot; \cdot, \cdot)$ is valid criterion to resolve Problems \ref{problem:one-sample} and \ref{problem:two-sample} as long as the depth functions (here, half-space depth) 
characterizes the distribution (i.e., $D_{F}({\bf x}) = D_{G}({\bf x})$ for all ${\bf x}\in\mathbb{R}^{d}\Leftrightarrow F = G$). For details on the characterization of the distributions based on half-space depth, see the following propositions.
\begin{proposition}[\citet{koshevoy2002tukey, CUESTAALBERTOS}]
\label{proposition:discrete}
If $F$ and $G$ are two distribution functions defined on $\bbR^{d}$, both with finite support and their half-space depth coincide, then $F=G$. 
Moreover, if $F$ is a discrete distribution with finite or countable support defined on $\bbR^{d}$ and $G$ is a Borel distribution on $\bbR^{d}$ such that the half-space depth coincides, then $F = G$.
\end{proposition}

\begin{proposition}[\citet{hassairi2008tukey}]
\label{proposition:cont}
Let $F$ be a probability distribution on $\bbR^{d}$ with connected support, and $\sH$ be the class of closed half-spaces $H$ in $\bbR^{d}$. Assume that $F$ is absolutely continuous with the connected support and that the corresponding probability density function is continuous with the interior of the support. 
Then the half-space depth characterizes the distribution. 
\end{proposition}
Propositions \ref{proposition:discrete} and \ref{proposition:cont} indicate that the difference of the half-space depth values can provide a valid measure of the discrepancy, which is defined as $\DDD$ in \eqref{eq:ddd}. 
Now, we propose a graphical method where we plot $\DDD$ as a scatter plot across a horizontal axis and measure the deviation from the horizontal axis. Needless to say, the advantage of this graphical tool-kit is that it remains unaffected by the dimension of data. This graphical approach can be used for any distribution and based on any depth function, but to produce a valid graphical tool-kit for addressing the Problems \ref{problem:one-sample} and \ref{problem:two-sample}, we continue our discussion with half-space depth since we can exploit the characterization properties discussed in Propositions \ref{proposition:discrete} and \ref{proposition:cont}. As a result, we can conclude that, if $\DDD(\bx; F, G) = 0$ for all $\bx$, and graphically, if the majority of points are concentrated on the horizontal axis, we can conclude that the two underlying distributions are expected to be identical, where $F$ and $G$ belong to a certain family of distribution functions. The illustrative examples provide the different patterns of deviations from the horizontal line for both goodness-of-fit and two-sample problems when the null hypothesis is not true. 

\subsubsection{Illustration of graphical tool-kit: Goodness-of-fit testing problem}
For a given sample $\sX = \left\lbrace \bX_{1}, \cdots, \bX_{n} \right\rbrace$ from an unknown distribution function $F$, 
we are interested in plotting $\DDD(\bx; \sX, F_{0}) =  D_{\sX}(\bx) - D_{F_{0}}(\bx)$ with respect to indices of $\bx$. Therefore, if $F = F_{0}$, then the value of $\DDD(\bx; \sX, F_{0})$ will be/close to zero with respect to the observed data. 
We now consider four simulated data sets each consisting of 100 i.i.d. observations.
Four sets of samples are generated from (1) standard bivariate normal distribution, (2) standard bivariate Cauchy distribution, (3) bivariate t-distribution with degrees of freedom (df) 3 and (4) standard bivariate Laplace distribution, respectively.
For all of them, we consider the bivariate normal distribution as the specified distribution $F_{0}$ with unknown mean $\bmu$ and dispersion $\bSigma$, that are estimated from the sample using the sample mean vector and the sample variance-covariance matrix, receptively. 
We standardize the sample using these estimates and see the DDD-plots based on standardized data. 
The data-depth discrepancy plots for the four simulated samples are shown in Figure \ref{fig:one-sample}. 
In those plots, the dotted gray curves in the figures indicate the two-sigma limits of $\DDD$. 
It is clearly evident from the graphs that the specified distribution fits the data in Figure \ref{fig:1normal} reasonably well; moreover, the data points in this plot are clustered around the horizontal axis and most of them are in the two-sigma limits. The other three plots \ref{fig:1cauchy}, \ref{fig:1laplace}, \ref{fig:1t} indicate deviations from the horizontal axis, which immediately implies that the data are not from a normal distribution. 

\begin{figure}[t!]
    \centering
    \begin{subfigure}{0.49\textwidth}
         \centering
         \includegraphics[width=\textwidth]{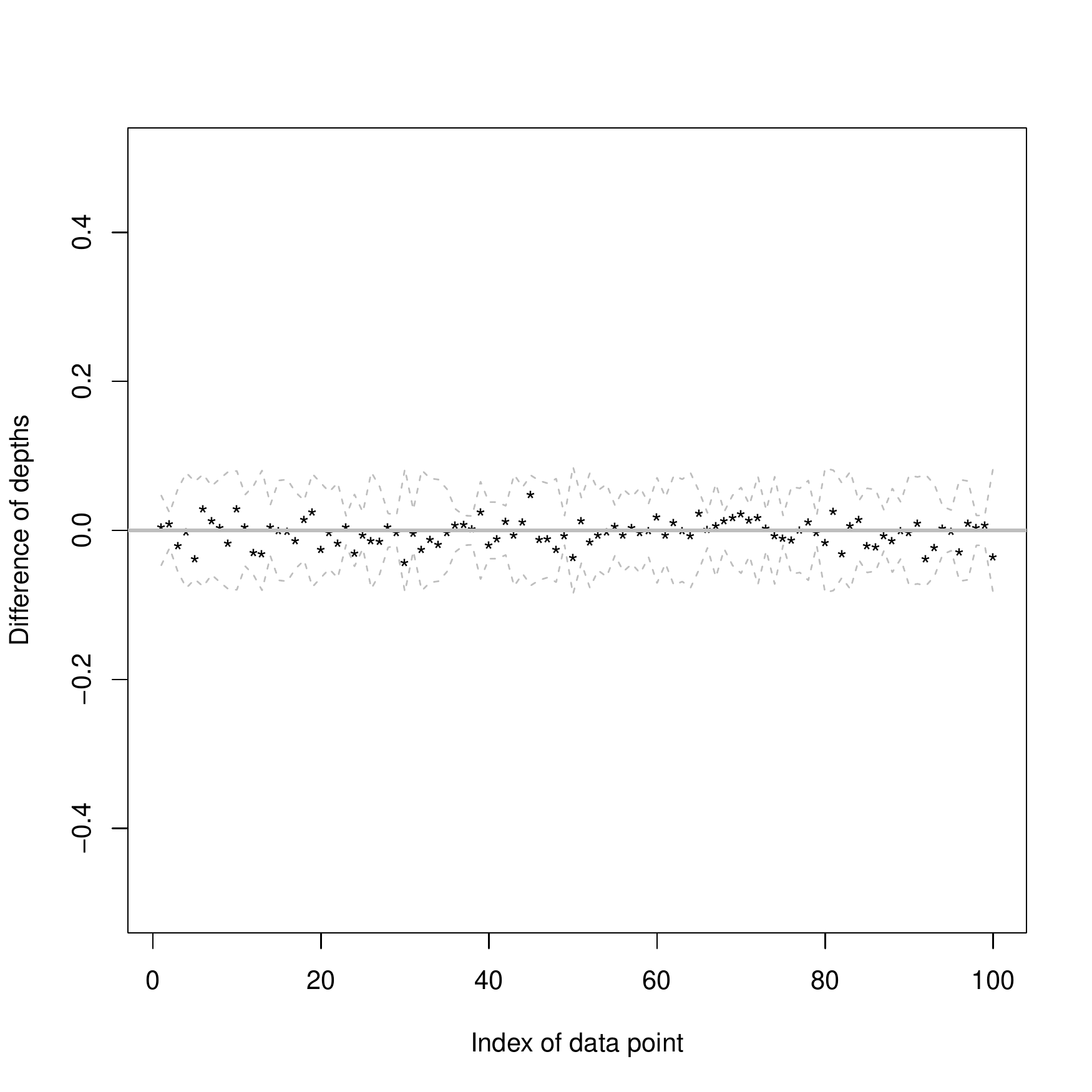}
         \caption{$\bX \sim \text{standard normal}$}
         \label{fig:1normal}
     \end{subfigure}
     \hfill
    \begin{subfigure}{0.49\textwidth}
         \centering
         \includegraphics[width=\textwidth]{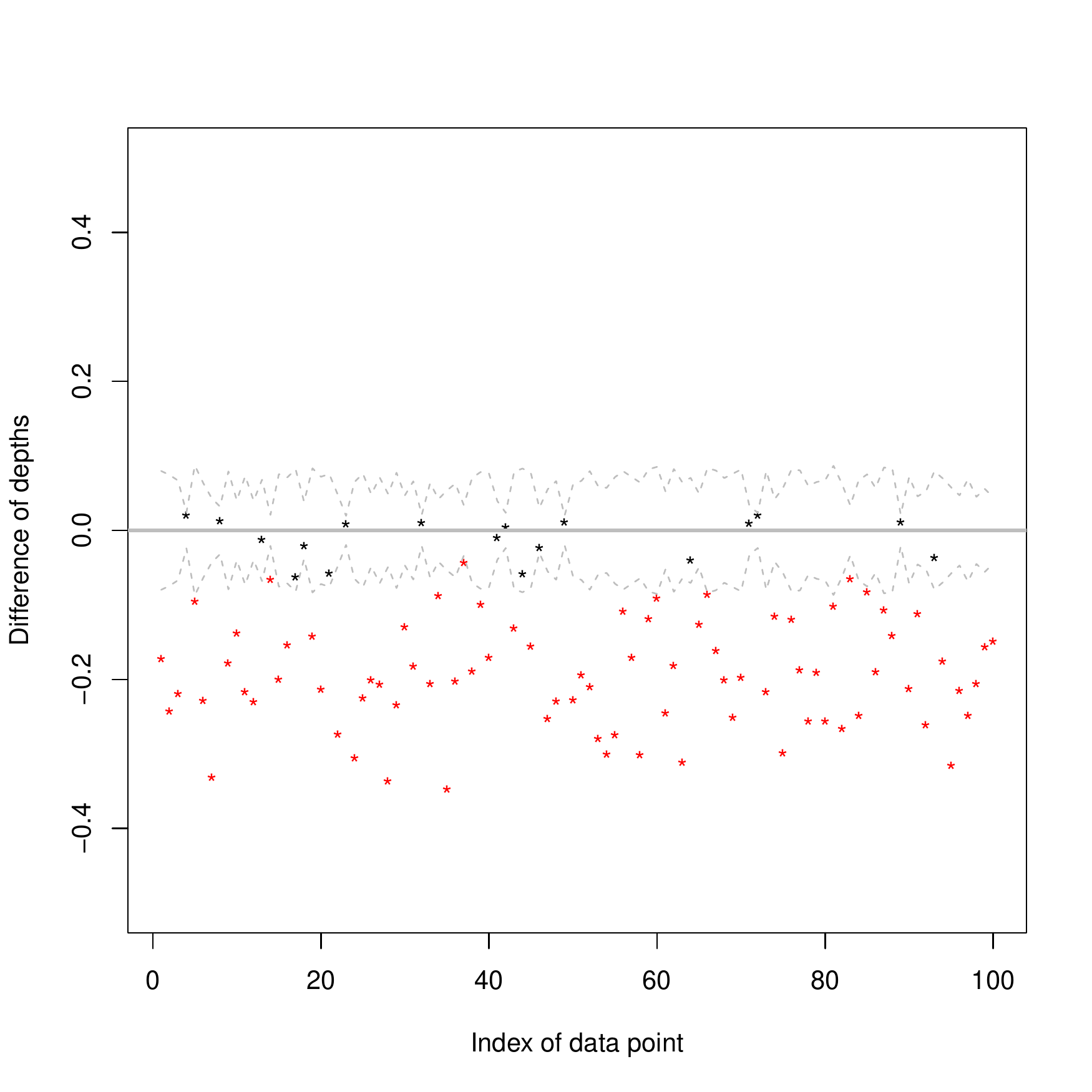}
         \caption{$\bX \sim \text{standard Cauchy}$}
         \label{fig:1cauchy}
     \end{subfigure}
     \hfill 
	\begin{subfigure}{0.49\textwidth}
         \centering
         \includegraphics[width=\textwidth]{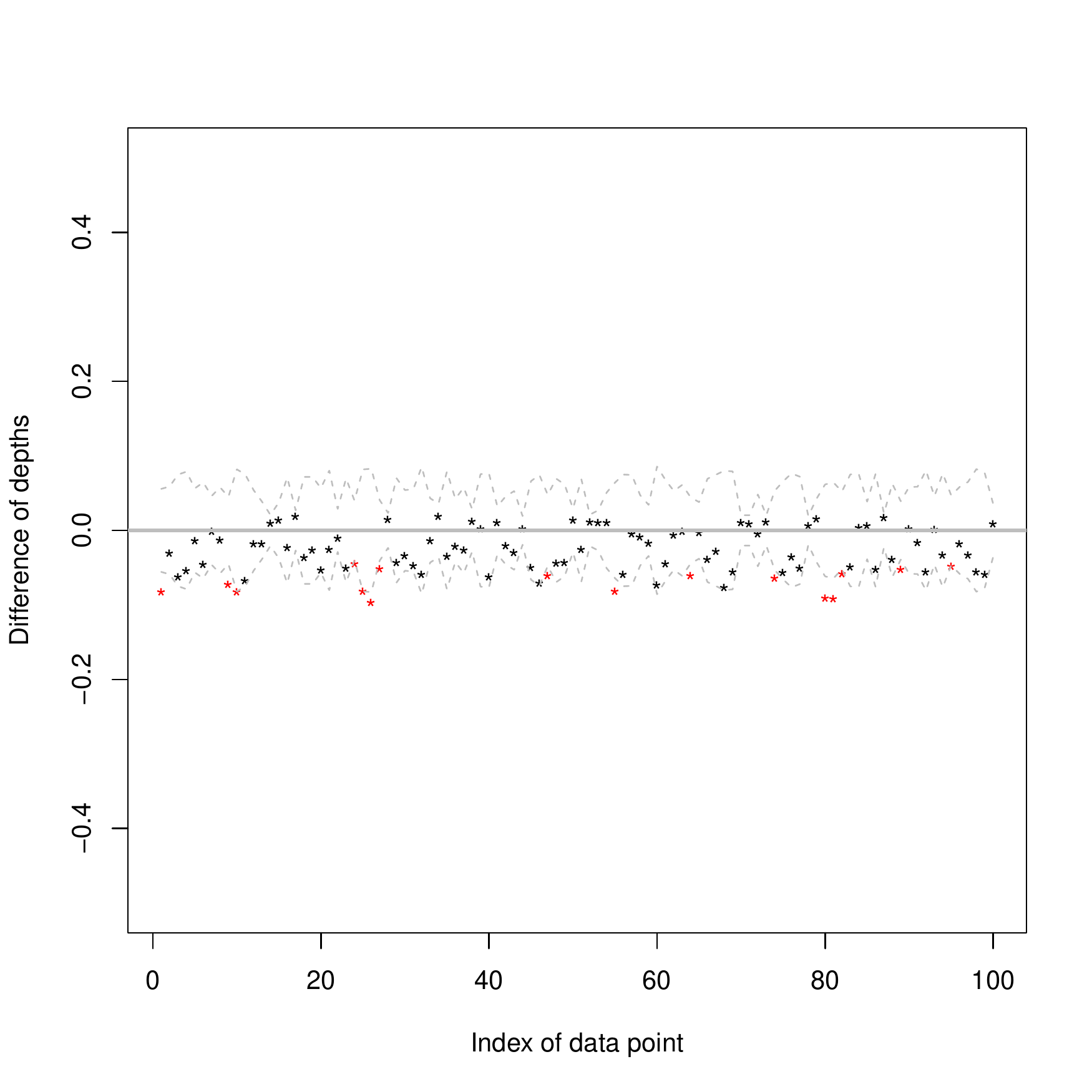}
         \caption{$\bX \sim \text{standard Laplace}$}
         \label{fig:1laplace}
     \end{subfigure}
     \hfill 
     \begin{subfigure}{0.49\textwidth}
         \centering
         \includegraphics[width=\textwidth]{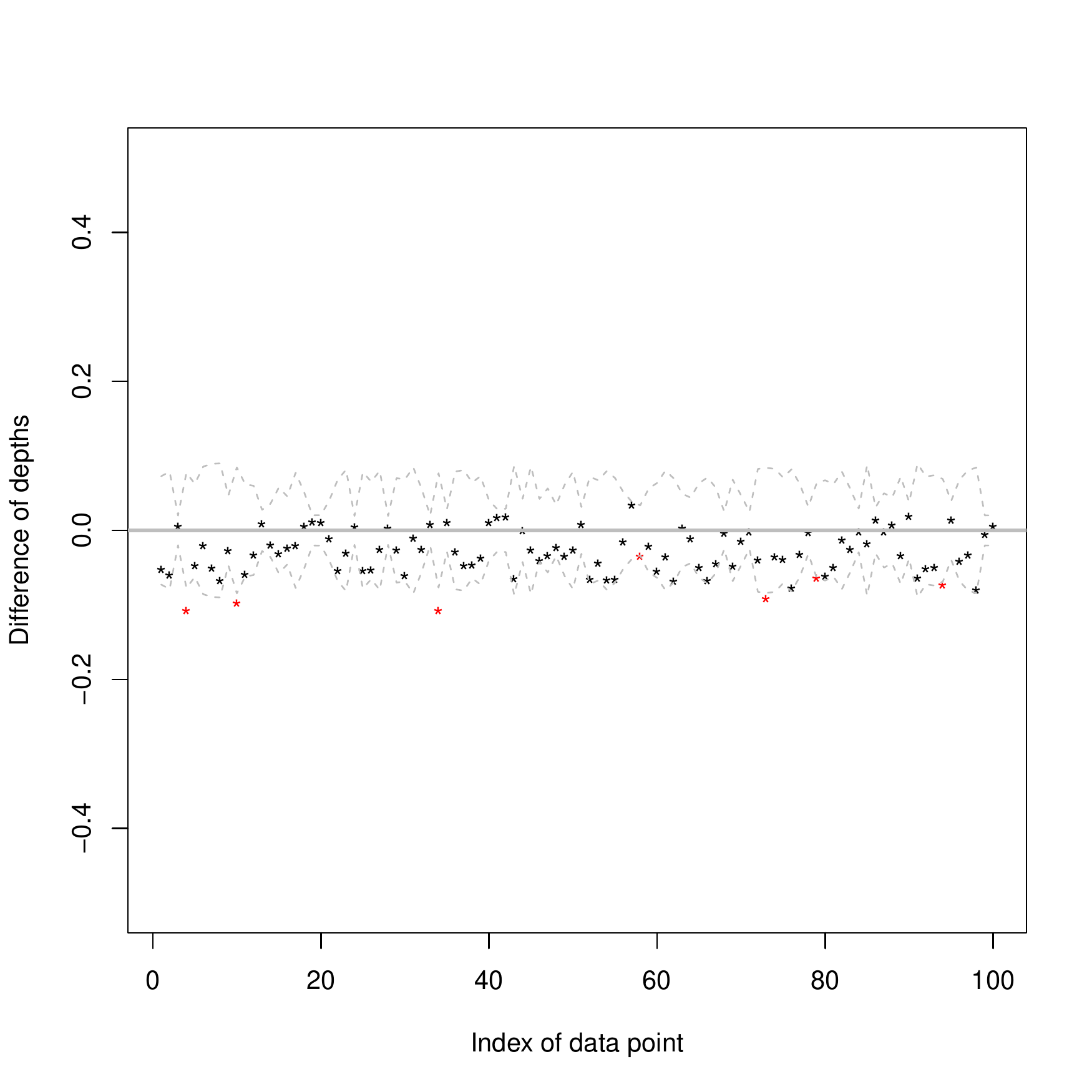}
         \caption{$\bX \sim \text{standard t with df 3}$}
         \label{fig:1t}
     \end{subfigure}
     \caption{The data-depth discrepancy plot for the goodness-of-fit problem where the specified distribution is bivariate normal. The plots of the first and second rows are for the examples, where the distribution of the data are form bivariate standard normal, Cauchy, Laplace, and $t_{3}$ distributions respectively. The dotted gray curves indicate the two-sigma limits of data-depth discrepancy. Points which are outside the two-sigma limits are color-coded by red.}
    \label{fig:one-sample}
\end{figure}

\subsubsection{Illustration of graphical tool-kit: Two-sample testing problem}
For given samples $\sX = \{\bX_{1}, \cdots, \bX_{n}\}$ and $\sY = \{\bY_{1}, \cdots, \bY_{m}\}$, here we plot $\DDD(\bx; \sX, \sY) = D_{\sX}(\bx) - D_{\sY}(\bx)$ for $\bx \in \sX\cup\sY$. If the two distributions are identical, then the value of $\DDD(\bx; \sX, \sY)$ will be zero/close to zero with respect to observed data. Here we consider two simulated data sets to demonstrate the performance of our proposed graphical tool-kit. 
In our first problem, we simulate \textit{sample-1} consisting of 100 i.i.d. observations generated from the trivariate normal distribution having zero mean and scatter matrix $\bSigma = (\sigma_{i,j})_{i, j = 1, 2, 3}$ with $\sigma_{1,2} = 0.9, \sigma_{1,3} = 0.2$ and $\sigma_{2,3} = 0.5$ $(F)$
and \textit{sample-2} consisting of 50 i.i.d. observations from the standard trivariate normal distribution denoted as $(G)$. 
In the next problem, \textit{sample-1} consists of 100 i.i.d. observations from the standard trivariate normal distribution $(F)$, and \textit{sample-2} consists of 50 i.i.d. observations from a trivariate skew-normal distribution $(G)$ \citep{azzalini1996multivariate}. 
The p.d.f. of the trivariate skew-normal distribution is given by $f(\bx) = 2\phi_{3}(\bx; \bOmega)\Phi(\balpha^{\tp}\bx)$, 
where, $\balpha^{\tp} = \frac{\blambda^{\tp} \bPsi^{-1} \bDelta^{-1}}{\sqrt{1+\blambda^{\tp} \bPsi^{-1} \blambda}}$, 
$\Delta = diag(\sqrt{1 - \lambda_{1}^{2}}, \sqrt{1 - \lambda_{2}^{2}}, \sqrt{1 - \lambda_{3}^{2}})$, 
$\blambda = \left(\frac{\lambda_{1}}{\sqrt{1 - \lambda_{1}^{2}}}, \frac{\lambda_{2}}{\sqrt{1 - \lambda_{2}^{2}}}, \frac{\lambda_{3}}{\sqrt{1 - \lambda_{3}^{2}}}\right)^{\tp}$,
and $\bOmega = \bDelta (\bPsi + \blambda \blambda^{\tp}) \bDelta$. 
Here $\phi_{3}(\bx; \bOmega)$ denotes the probability density function of a trivariate normal distribution with standardized marginals and the correlation matrix $\bOmega$, and $\bPsi$ is the distribution function of the standard univariate normal distribution. In this example, we have considered $\lambda_{1} = \lambda_{2} = \lambda_{3} = 0.9$ and $\bPhi = \bI_{3}$, identity matrix of dimension $3\times 3$. 
The data-depth discrepancy plots for these two toy examples are shown in Figure \ref{fig:two-sample}. As before, the gray dotted curves in the figures indicate the two-sigma limits of $\DDD$. 

\afterpage{
\begin{landscape}
\begin{figure}[t!]
    \centering
    \begin{subfigure}{0.7\textwidth}
         \centering
         \includegraphics[width=\textwidth]{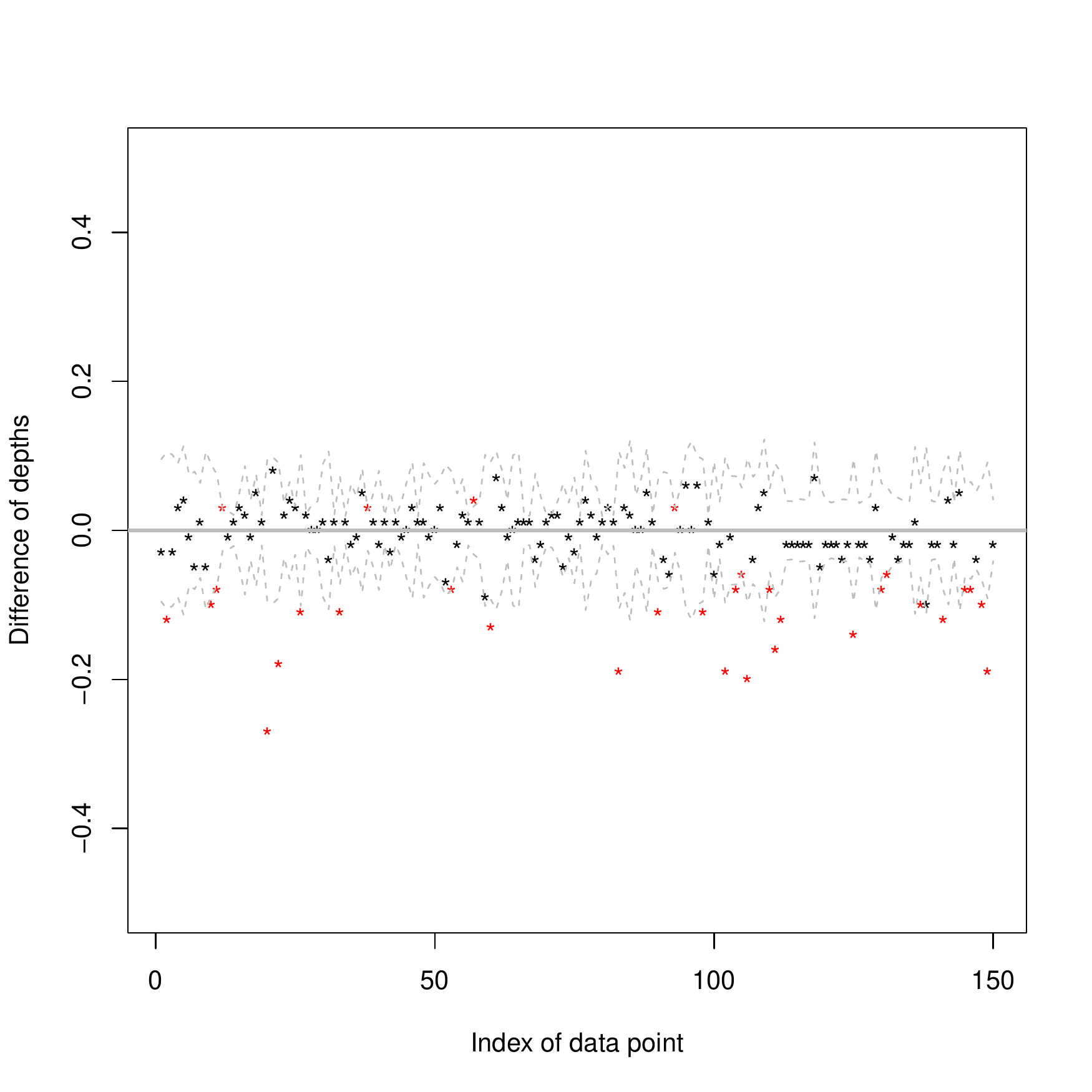}
         \caption{$\bX \sim N(0, \bSigma)$ and $\bY \sim N(0, \bI_{3})$}
         \label{fig:2normal}
     \end{subfigure}
	\begin{subfigure}{0.7\textwidth}
         \centering
         \includegraphics[width=\textwidth]{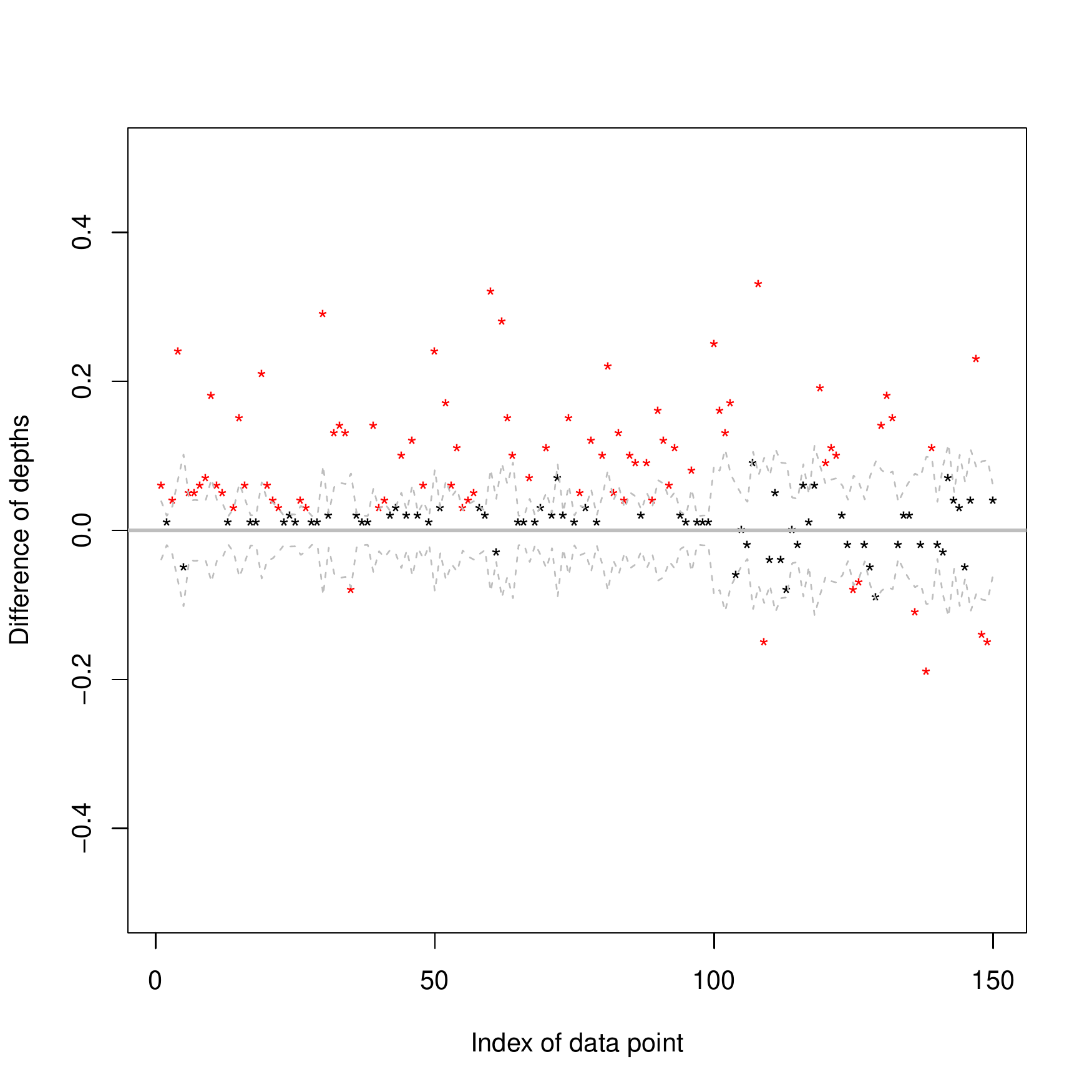}
         \caption{$\bX \sim N(0, \bI_{3})$ and $\bY \sim f(z) $ where $f(z)$ is the p.d.f. of skew-normal distribution}
         \label{fig:2skewnormal}
     \end{subfigure}
	\caption{The data-depth discrepancy plot for the two-sample examples. The plots of the first column for example where the samples are from a trivariate normal distribution with different dispersion matrices, and those in the second column for example, where the samples are generated from different distributions. The dotted gray curves indicate the two-sigma limits of data-depth discrepancy. Points that are outside the two-sigma limits are color-coded in red.}
 \label{fig:two-sample}
\end{figure}
\end{landscape}}

\subsection{Associated statistical tests}
\label{subsec:test}
In Problem \ref{problem:one-sample}, for given sample $\sX$ of size $n$, we are interested to test $ H_{0} : F = F_{0}$  against $ H_{1} : F \neq F_{0} $, where $F_{0}$ is a pre-specified distribution function. Due to Propositions \ref{proposition:discrete} and \ref{proposition:cont}, for half-space depth function $D$, it is equivalent to test $H^{*}_{0} : D_{F}(\bx) = D_{F_{0}}(\bx) \; \text{for all} \; \bx \in \bbR^{d}$ against $H^{*}_{1} : D_{F}(\bx) \neq D_{F_{0}}(\bx)$ for some $\bx \in \bbR^{d}$. 
The most well-known test statistics for comparing two distributions are the Kolmogorov–Smirnov (KS) test statistics and Cram\'er–von Mises (CvM) test statistics. These test statistics are based on certain differences between the empirical distribution function $F_{n}$ and the hypothesized distribution function $F_{0}$. Here, we propose alternative test statistics based on $\DDD$ defined in Section \ref{subsec:ddd} so that we can perform the goodness-of-fit test for any arbitrary dimension of the data. 
\begin{enumerate}
    \item Depth-based KS test statistic:
    \item[] $T^{\KS}_{\sX, F_{0}} = 
    \sqrt{n}\sup_{\bx \in \bbR^{d}}|\DDD(\bx; \sX, F)|
    = \sqrt{n}\sup_{\bx \in \bbR^{d}}|D_{\sX}(\bx) - D_{F_{0}}(\bx)|.$
    
    \item Depth-based CvM test statistic: 
    \item[] $T^{\CvM}_{\sX, F_{0}} = n \int \DDD^{2}(\bx; \sX, F_{0}) dF_{0}(\bx) = n \int (D_{\sX}(\bx) - D_{F_{0}}(\bx))^{2} dF_{0}(\bx).$
\end{enumerate}
For testing $H_{0}$ against $H_{1}$, the null hypothesis will be rejected when the values of the test statistics are very large. The asymptotic behaviors of the proposed test statistics are provided in Section \ref{sec:technical}. 
\begin{remark}
\label{re:1}
    \textit{Note that the true depth value $D_{F_{0}}(\cdot)$ is approximated by its empirical version in computing $T^{KS}_{{\cal{X}}, F_{0}}$ as the
    empirical half-space depth is a uniformly consistent estimator of population half-space depth, which follows from Corollary 2.3 of \citet{masse2004asymptotics}. 
    Consequently, for large enough samples, the approximated test statistic provided in the algorithm (see after Remark \ref{re:2}) computes the actual value of the original test statistic arbitrarily well. This kind of idea is often used in computing the values of one sample of Kolmogorov-Smirnov or Cram\'er-von Mises test statistics.}
\end{remark}
\begin{remark}
\label{re:2}
    \textit{In order to compute supremum involved in $T_{\sX, F_{0}}^{\KS}$, one may approximate supremum over an unbounded set by supremum over a compact set. For instance, we here consider $d$-dimensional unit ball. Next, for computing  $T_{\sX, F_{0}}^{\CvM}$, note that $n \int (D_{\sX}(\bx) - D_{F_{0}}(\bx))^{2} dF_{0}(\bx) = \E_{\bX \sim F_{0}}\{\sqrt{n}(D_{\sX}(\bX) -  D_{F_{0}}(\bX))\}^{2}$. Hence, for a given observations $\bY_{1}, \cdots, \bY_{M}$ from $F_{0}$, using law of large number, $\frac{1}{M}\sum_{i=1}^{M}\{\sqrt{n}(D_{\sX}(\bY_{i}) -  D_{F_{0}}(\bY_{i}))\}^{2}$ can approximate the value of $T_{\sX, F_{0}}^{\CvM}$ arbitrarily well. }
\end{remark}
To obtain the empirical p-value based on the proposed test statistics $T_{\sX, F_{0}}$, where $T_{\sX, F_{0}}$ is either $T_{\sX, F_{0}}^{\KS}$ or $T_{\sX, F_{0}}^{\CvM}$, we consider a bootstrap procedure \citep{efron1982jackknife} that consists of the following steps. 
\begin{enumerate}[label=Step 1.\arabic*., align=left]
    \item Generate $M$ points uniformly from the boundary of the $d$-dimensional unit ball, for $M$ sufficiently large. The set of such points is denoted as $\sU_{1}$. In addition, generate a random sample with size $M$ from $F_{0}$, and the set of such points is denoted as $\sU_{2}$. 
    \item Calculate half-space depth with respect to $\sX$ and $F_{0}$, respectively (follow remark \ref{re:1} for computing Half-space depth with respect to $F_{0}$). Therefore, (a) $\Tilde{T}^{\KS}_{\sX, F_{0}} \approx \sqrt{n}\max_{\bx \in \sU_{1}} |\DDD(\bx; \sX, F_{0})|$, (b) $\Tilde{T}^{\CvM}_{\sX, F_{0}} \approx n \times \frac{1}{M}\sum\limits_{j=1}^{M}\DDD^{2}(\bx_{j}; \sX, F_{0})$ where $\bx_{j} \in \sU_{2}$ defined in Step 1.1.
    
    \item Draw a random sample $\sX^{*} = \{\bX^{*}_{1}, \cdots, \bX^{*}_{n}\}$ from the null distribution. 

    \item Repeat Step 1.2 based on the data $\sX^{*}$. 
    
    \item Repeat Steps 1.3 and 1.4 for $B$ times, where $B$ is sufficiently large. 
    Thus, the empirical distribution of 
    $\{\widetilde{T}^{(b)}_{\sX^{*}, F_{0}}: b = 1, \cdots, B\}$ 
    can be used to approximate the null distribution of ${T}_{\sX, F_{0}}$.
    \item The empirical p-value is computed using the following formula: $\widehat{\bbP}_{H_{0}}\left\{T_{\sX, F_{0}} > \Tilde{T}_{\sX, F_{0}}\right\} = \frac{1}{B}\sum_{b = 1}^{B}\textbf{1}\left\{\widetilde{T}_{\sX^{*}, F_{0}}^{(b)} > \tilde{T}_{\sX, F_{0}}\right\}$.
\end{enumerate}
\par
Next, we consider the two-sample $d$-dimensional problem for two independent samples $\sX = \left\lbrace \bX_{1}, \cdots, \bX_{n} \right\rbrace$, and $\sY = \left\lbrace \bY_{1},  \cdots, \bY_{m} \right\rbrace$ where $\bX_{i}$'s are independently distributed with distribution function $F$ and $\bY_{i}$'s are independently distributed with distribution function $G$. 
In Problem \ref{problem:two-sample}, we want to test $H_{0} : F = G$  against $ H_{1} : F \neq G $, which is equivalent to test $H^{*}_{0} : D_{F}(\bx) = D_{G}(\bx) \; \text{for all} \; \bx \in \bbR^{d}$ against $H^{*}_{1} : D_{F}(\bx) \neq D_{G}(\bx)$ for some $\bx \in \bbR^{d}$, where $D$ is the half-space depth. 
Similar to the goodness-of-fit test problem, the above equivalence is meaningful due to the characterization property of half-space depth. 
We construct the KS and CvM type test statistics based on $\DDD(\bx; \sX, \sY)$ and obtain the following test statistics. 
\begin{enumerate}
    \item Depth-based KS test statistic: 
    \item[] $T^{\KS}_{\sX, \sY} = 
        \sqrt{n+m}\sup_{\bx \in \bbR^{d}}|\DDD(\bx; \sX, \sY)|
        = \sqrt{n+m}\sup_{\bx \in \bbR^{d}}|D_{\sX}(\bx) - D_{\sY}(\bx)|$
    
    \item Depth-based CvM test statistic: 
    \item[] $T^{\CvM}_{\sX, \sY} = (n+m)\int \DDD^{2}(\bx; \sX, \sY)dH_{n,m}(\bx)
        = (n+m)\int (D_{\sX}(\bx) - D_{\sY}(\bx))^{2}dH_{n, m}(\bx),$ where $H_{n,m}(\cdot)$ is the empirical distribution function based on the combined sample $\sX \cup \sY$. 
\end{enumerate}
Here, we also reject the null hypothesis for the large value of the test statistics. The asymptotic results for the proposed test statistics are studied in Section \ref{sec:technical}. 
\begin{remark}
\label{re:two}
    \textit{Similar to the one-sample problem discussed in Remark \ref{re:2}, in order to compute $T_{\sX, \sY}^{KS}$, one may approximate supremum over an unbounded set by supremum over a compact set. For instance, we here consider $d$-dimensional unit ball. Next, note that, $T_{\sX, \sY}^{\CvM}=(n + m) \int (D_{\sX}(\bx) - D_{\sY}(\bx))^{2} dH_{n, m}(\bx) = (n+m)\times \frac{1}{(n+m)}\sum_{j=1}^{(n+m)}(D_{\sX}(\by_{j}) - D_{\sY}(\by_{j}))^{2}$, where $\by_{j} \in \sX \cup \sY$ for $j = 1, \cdots, (n+m)$ as $H_{n, m}(\cdot)$ is the empirical distribution function based on the combined sample $\sX \cup \sY$. This fact will enable us to compute $T_{\sX, \sY}^{\CvM}$.}
\end{remark}
The following algorithm provides an empirical p-value \citep{tibshirani1993introduction} for Problem \ref{problem:two-sample} based on any of the proposed test statistics $T_{\sX, \sY}$ where $T_{\sX, \sY}$ is either $T_{\sX, \sY}^{\KS}$ or $T_{\sX, \sY}^{\CvM}$. 
\begin{enumerate}[label=Step 2.\arabic*., align=left]
    \item Generate $M$ points uniformly from the boundary of $d$-dimensional unit ball, for $M$ sufficiently large. The set of such points is denoted as $\sU_{1}$. 
    \item Calculate half-space depth with respect to samples $\sX$ and $\sY$ respectively, and therefore compute (a) $\widetilde{T}^{\KS}_{\sX, \sY} = \sqrt{n+m}\max_{\bx \in \sU_{1}} |\DDD(\bx; \sX, \sY)|$, (b) $\widetilde{T}^{\CvM}_{\sX, \sY} = \sum\limits_{j=1}^{n+m}\DDD^{2}(\bx_{j}; \sX, \sY)$ for $\bx_{j} \in \sX \cup \sY$.
    
    \item Combine the samples $\sX$ and $\sY$, and the pooled sample is denoted as 
    $\sZ = \{\bX_{1}, \cdots, \bX_{n}\} \cup \{ \bY_{1}, \cdots, \bY_{m}\}$. 
    \item For the $b$-th bootstrap replicate, take a sample of size $(n+m)$ from $\sZ$ with replacement and treat the first $n$ as the $\sX_{b}^{*}$ sample and the last $m$ are $\sY_{b}^{*}$  for $b = 1, \cdots, B$, where $B$ is sufficiently large. 
    \item Compute half-space depths based on the samples $\sX_{b}^{*}$ and $\sY_{b}^{*}$ separately for each of the replicates. 
    Therefore, compute the approximate test statistic $\widetilde{T}_{\sX^{*}, \sY^{*}}^{(b)}$ as described in Step 2.2 by replacing $\sX$ with $\sX_{b}^{*}$ and $\sY$ with $\sY_{b}^{*}$. Thus, the empirical distribution of $\widetilde{T}^{(b)}_{\sX^{*}, \sY^{*}}$s can be used to approximate the null distribution of ${T}_{\sX, \sY}$.
    \item The empirical p-value is computed using the following formula: $\widehat{\bbP}_{H_{0}}\left\{T_{\sX, \sY} > \widetilde{T}_{\sX, \sY}\right\} = \frac{1}{B}\sum_{b = 1}^{B}\textbf{1}\left\{\widetilde{T}_{\sX^{*}, \sY^{*}}^{(b)} > \widetilde{T}_{\sX, \sY}\right\}$.
\end{enumerate}
\begin{remark}
\textit{Using the Algorithms described above, it is observed that the computation of the p-values is simple and efficient enough even when the data dimension is fairly large. }
\end{remark}

\section{Main results}
\label{sec:technical}
In the earlier section, we studied the algorithms to compute the p-value of the tests. However, one needs to know the distribution of the test statistics to estimate the size and power of the tests. 
Observe that the derivation of the exact distribution of the test statistics is intractable because of its complex expression, and to overcome this issue, this section investigates
the asymptotic consistency and the asymptotic power under contiguous alternatives of the proposed test using the asymptotic distribution of the test statistics. 

\subsection{Large sample statistical properties}
\label{subsec:consistency}
To investigate the asymptotic properties of the proposed graphical tool-kit and tests, one first needs to assume a few technical conditions. 
\par
Suppose that $\bX_{1}, \cdots, \bX_{n}$ are the multivariate observations from an unknown distribution $F$, and the corresponding empirical distribution function is denoted by $\widehat{F}_{n}$. We define $\sL_{n}(\bx) = \sqrt{n}(D_{\sX}(\bx) - D_{F}(\bx))$ for $\bx \in \bbR^{d}$. Therefore, $\{\sL_{n}:\bx\in \bbR^{d}\}$ is a stochastic process with bounded sample path which is a map into $l^{\infty}(\bbR^{d})$. Here $l^{\infty}(\bbR^{d})$ is a space of bounded real valued function on $\bbR^{d}$ equipped with the uniform norm, viz., $\|a\|_{\infty} = \sup\limits_{\bt\in \bbR^{d}}|a(\bt)|$.
Moreover, define $\sV_{n} (H) = \sqrt{n}(\widehat{F}_{n} -F)(H)$, and that can be viewed as a map into space $l^{\infty}(\sH)$, where $\sH$ is the class of closed half-spaces $H$ in $\bbR^{d}$. 
The technical assumptions: 
\begin{enumerate}[label=(C\arabic*)]
    \item\label{cond:smooth}  The distribution function $F$ satisfies $F(\partial H) = 0$ for all $H \in \sH$, where  $\partial H$ is the topological boundary of $H$. 
    \item\label{cond:set} $\sF = \{F: F~\mbox{is a proper distribution function}\}$ is a totally bounded, and permissible subclass of $L_{2}(F)$ on $\sH$, where $L_{2}(F) = \{f : \int_{{\cal{H}}} f^{2}(x) dF(x)<\infty\}. $  
    Moreover, for each $\eta >0$ and $\epsilon >0$, there exists a $\delta > 0$ such that,  $\lim\sup\bbP\left\{\sup_{\sD(\delta)}|\sV_{n}(f-g)| > \eta\right\} < \epsilon$, where $\sD(\delta) = \{(f, g): f, g \in \sF, \rho_{F}(f-g) < \delta\}$ for the semi-norm $\rho_{F}$ on $\sH$. 
    \item\label{cond:LR} A minimal closed half-space at $\bx$, viz. $H[\bx]$ is uniquely defined if $D_{F}(\bx) > 0$.
    \item\label{cond:dist} The distribution functions are either with finite support or absolutely continuous with continuous probability density function.
\end{enumerate}
\begin{remark}
\ref{cond:smooth} indicates the behavior of the underlying distribution function $F$ on the boundary of the half spaces. This is useful to obtain the stronger properties of Tukey's half-space depth in different applications. For example, it is easy to see that if $F$ is absolutely continuous, the condition \ref{cond:smooth} is satisfied \citep{masse2004asymptotics}. 
Moreover, if $F$ satisfies \ref{cond:smooth}, $D_{F}(\bx)$ can be expressed as the probability of some closed half-space whose boundary goes through $\bx$. 
In addition to that, for the closed half-space $H[\bx, \bu] = \{\by\in \bbR^{d}: \bu^{\tp}\by \leq \bu^{\tp}\bx\}$,  the function $(\bx, \bu) \mapsto F(H[\bx, \bu])$ is continuous on $\bbR^{d}\times \sS^{d-1}$ and $\bx \rightarrow D_{F}(\bx)$ is continuous. 
Condition \ref{cond:set} is useful for empirical central limit theorems (see \citet{pollard2012convergence}. 
Moreover, if $D(\bx) = F(H[\bx, \bu])$, then $H[\bx, \bu]$ is a minimal half-space at $\bx$, and $\bu$ is a minimal direction at the same point. 
\ref{cond:LR} satisfies the condition related to the multiplicity of the minimal direction at $\bx$. 
The conditions \ref{cond:smooth} and \ref{cond:LR} are referred to as the ``local regularity conditions'' in the literature. 
The details about the condition \ref{cond:dist} is discussed in Propositions \ref{proposition:discrete} and \ref{proposition:cont}.
\end{remark}
\par
Proposition \ref{proposition:weakVn} states the weak convergence of $\sV_{n}$ (defined at the beginning of Section 4.1), which is the key component of ${\cal{L}}_{n}({\bf x})$. 
\begin{proposition}(Theorem 21 in section VII.5 of  \citet{pollard2012convergence})
\label{proposition:weakVn}
Under condition \ref{cond:set}, $\sV_{n} =\sqrt{n}(\widehat{F}_{n}-F)$ converges weakly to $\sV_{F}$, where $\sV_{F}$ is tight and $F$-Brownian bridge with mean zero and covariance function $\Sigma_{F} = F(H_{1}\cap H_{2}) - F(H_{1})F(H_{2})$ for $H_{1}, H_{2}\in \sH$. Moreover, $\sV_{F}$ can be chosen such that each sample path is continuous with respect to $\rho_{F}$. 
\end{proposition}
Let us now define $\sJ(\sV_{F})(\bx) = \inf\limits_{\bv\in V(\bx)}\sV_{F}H[\bx, \bv]$ for $\bx \in \bbR^{d}$, where $V(\bx)$ is the set of all minimal directions passing through that point. Proposition \ref{proposition:depth} states the asymptotic distribution of Tukey's half-space depth. 
\begin{proposition}\citep{masse2004asymptotics}
\label{proposition:depth}
Under the conditions \ref{cond:smooth}-\ref{cond:LR}, 
$\{\sL_{n}\}$ (same as the defined in the beginning of Section \ref{subsec:consistency}) converges weakly to $\sJ(\sV_{F})$, which is tight measurable map into $l^{\infty}(A)$, where $A \subset \bbR^{d}$. Strictly speaking, $\sJ(\sV_{F})$ is a random element associated with a centered Gaussian process with covariance function $\cov\{\sJ(\sV_{F})(\bx_{1}),  \sJ(\sV_{F})(\bx_{2})\} = F(H[\bx_{1}]\cap H[\bx_{2}]) - F(H[\bx_{1}])F(H[\bx_{2}])$ for $\bx_{1}, \bx_{2} \in A$. 
\end{proposition}
\par
Using the assertions in Propositions \ref{proposition:weakVn} and \ref{proposition:depth}, we now state the results, which justifies the usefulness of the proposed graphical tool-kits for sufficiently large sample sizes. 
See Theorems \ref{thm:graph1} and \ref{thm:graph2}:
\begin{theorem}
\label{thm:graph1}
For every $\epsilon > 0$, define 
$\sC_{\epsilon}(F, F_{0}) = \{ ({\bf x}, y): \bx \in \bbR^{d},  
|y - (D_{F}(\bx) - D_{F_{0}}(\bx))| < \epsilon\}$ and for a given sample $\sX = \{ \bX_{1}, \cdots, \bX_{n}\}$, define 
$\widehat{\sC}(\sX, F_{0}) = \{({\bf x}, D_{\sX}(\bx) - D_{F_{0}}(\bx)): \bx \in \sX\}$.
Then, under the conditions \ref{cond:smooth}-\ref{cond:dist}, for every $\epsilon > 0$, we have $\lim\limits_{n \rightarrow \infty}\bbP\left\{\widehat{\sC}(\sX, F_{0}) \subset \sC_{\epsilon}(F, F_{0}) \right\} = 1.$

\end{theorem}
\begin{theorem}
\label{thm:graph2}
For every $\epsilon > 0$, define $\sC_{\epsilon}(F, G) = \{ ({\bf x}, y): \bx \in \bbR^{d},  |y - (D_{F}(\bx) - D_{G}(\bx))| < \epsilon\}$ and for given independent samples $\sX = \{ \bX_{1}, \cdots, \bX_{n}\}$ and $\sY = \{ \bY_{1}, \cdots, \bY_{m}\}$, define  
$\widehat{\sC}(\sX,\sY) = \{({\bf x}, D_{\sX}(\bx) - D_{\sY}(\bx)): \bx \in \sX \cup \sY\}$.
Then, under the conditions \ref{cond:smooth}-\ref{cond:dist}, for positive finite number $\lambda = \lim\limits_{\min(n, m) \rightarrow \infty} n/(n+m)$ and for every $\epsilon >0$, we have $\lim\limits_{\min(n, m) \rightarrow \infty}\bbP\left\{\widehat{\sC}(\sX, \sY) \subset \sC_{\epsilon}(F, G) \right\} = 1.$
\end{theorem}
\begin{remark}
{\it Theorems \ref{thm:graph1} and \ref{thm:graph2} show that for large sample size, the points cluster around the horizontal axis if and only if the null hypothesis is true for both the goodness-of-fit and two-sample testing problem under some mild conditions. }
\end{remark}
Next, Theorems \ref{thm:const1-KM} and \ref{thm:const1-CVM} assert the point-wise asymptotic properties of the goodness-of-fit tests based on $T_{\sX, F_{0}}^{\KS}$ and $T_{\sX, F_{0}}^{\CvM}$.
\begin{theorem}
\label{thm:const1-KM}
Under the conditions \ref{cond:smooth}-\ref{cond:dist}, the test based on test statistic $T_{\sX, F_{0}}^{\KS}$ for testing $H_{0}: F = F_{0}$ against $H_{1}: F \neq F_{0}$ is point-wise consistent, i.e., $\bbP_{F}\{T_{\sX, F_{0}}^{KM} > s_{1-\alpha}^{(1)}\} \rightarrow 1$ as $n \rightarrow \infty$ under $H_{1}$, where $s_{1-\alpha}^{(1)}$ is such that $\lim\limits_{n\rightarrow\infty}\bbP_{H_{0}}\left\{T_{\sX, F_{0}}^{\KS} >s_{1-\alpha}^{(1)}\right\} = \alpha\in (0, 1)$.
\end{theorem}
\begin{theorem}
\label{thm:const1-CVM}
Under the conditions \ref{cond:smooth}-\ref{cond:dist}, 
the test based on test statistic $T_{\sX, F_{0}}^{\CvM}$ for testing $H_{0}: F = F_{0}$ against $H_{1}: F \neq F_{0}$ is point-wise consistent, i.e., $\bbP\left\{T_{\sX, F_{0}}^{\CvM} > s_{1-\alpha}^{(2)}\right\} \rightarrow 1$ as $n \rightarrow \infty$,
where $s_{1-\alpha}^{(2)}$ is such that $\lim\limits_{n\rightarrow\infty}\bbP_{H_{0}}\left\{T_{\sX, F_{0}}^{\CvM} > s_{1-\alpha}^{(2)}\right\} = \alpha\in (0, 1)$.
\end{theorem}

\begin{remark}
{\it Theorems \ref{thm:const1-KM} and \ref{thm:const1-CVM} indicate that the power of the proposed test for goodness of fit problem will converge to the highest possible value, i.e., one when the sample size is sufficiently large}.
\end{remark}
Now, Theorems \ref{thm:unif-const1-KM} and \ref{thm:unif-const1-CVM} assert the asymptotically uniform power properties of the goodness-of-fit test based on $T_{\sX, F_{0}}^{\KS}$ and $T_{\sX, F_{0}}^{\CvM}$. 
\begin{theorem}
\label{thm:unif-const1-KM}
Suppose that $\sX = \{ \bX_{1}, \cdots, \bX_{n}\}$
is a collection of i.i.d. random variables with distribution function $F$. For testing $F=F_{0}$ against $H_{1}: F \neq F_{0}$, under the conditions \ref{cond:smooth}-\ref{cond:dist}, the power of the test based on $T_{\sX, F_{0}}^{\KS}$ tends to 1 uniformly over a sequence of alternatives $F_{n}$ satisfying $\sqrt{n}L_{\infty}(D_{F_{n}}, D_{F_{0}}) \geq \Delta_{n}$, where $\Delta_{n} \rightarrow\infty$ as $n \rightarrow\infty$. 
\end{theorem}

\begin{theorem}
\label{thm:unif-const1-CVM}
Suppose that $\sX = \{ \bX_{1}, \cdots, \bX_{n}\}$
is a collection of i.i.d. random variables with distribution function $F$. For testing $H_{0}: F= F_{0}$ against $H_{1}: F \neq F_{0}$, under the conditions \ref{cond:smooth}-\ref{cond:dist}, 
the power of test based on $T_{\sX, F_{0}}^{\CvM}$ tends to 1 uniformly over a sequence of alternatives $F_{n}$ satisfying $\sqrt{n}L_{\infty}(D_{F_{n}}, D_{F_{0}}) \geq \Delta_{n}$, where $\Delta_{n} \rightarrow\infty$ as $n \rightarrow\infty$. 
\end{theorem}

\begin{remark}
{\it Not only for the fixed alternative, Theorems \ref{thm:unif-const1-KM} and \ref{thm:unif-const1-CVM} indicate that the power of the proposed test for goodness of fit problem will converge to the highest possible value, i.e., one when the sample size is sufficiently large as long as the sequence of alternatives are satisfying a certain condition.}
\end{remark}
\par
We now describe the results related to pointwise and uniform consistency of the proposed two-sample testing procedure in the similar spirit of Theorems \ref{thm:const1-KM}, \ref{thm:const1-CVM}, \ref{thm:unif-const1-KM} and \ref{thm:unif-const1-CVM}. 
\begin{theorem}
\label{thm:two-sample}
Let us denote $T^{(1)}_{\sX, \sY} := T_{\sX, \sY}^{\KS}$ and $T^{(2)}_{\sX, \sY} := T_{\sX, \sY}^{\CvM}$, and $n$ and $m$ are such that $\lim\limits_{\min(n,m)\rightarrow \infty}\frac{n}{n+m} = \lambda \in (0,1)$.
For $i=1$ and $2$,
let $t_{1-\alpha}^{(i)}$ be such that 
$\lim\limits_{\min(n, m) \rightarrow \infty}\bbP_{H_{0}}\left\{ T_{\sX, \sY}^{(i)} > t_{1-\alpha}^{(i)}\right\} = \alpha\in (0, 1)$.
Further, suppose that $H(\bx) = \lambda F(\bx) + (1-\lambda)G(\bx)$. Then, under the conditions \ref{cond:smooth}-\ref{cond:dist}, 
$\bbP_{H_{1}}\left\{ T_{\sX, \sY}^{(i)} > t_{1-\alpha}^{(i)} \right\} \rightarrow 1$, as $\min(n, m)\rightarrow\infty$. 
Moreover, for $i=1$ and $2$, the power of the test based on $T^{(i)}_{\sX, \sY}$ tends to one, uniformly over a sequence of alternatives $F_{n, m}$ satisfying $\sqrt{n + m}L_{\infty}(D_{F_{n, m}}, D_{G_{m}}) \geq \Delta_{m, n}$, where $\Delta_{m,n} \rightarrow \infty$ as $\min(m, n) \rightarrow \infty$, where $G_{m}$ is the empirical distribution based on ${\cal{Y}}$.
\end{theorem}

\subsection{Asymptotic local power study}
In Section \ref{subsec:consistency}, we have established that data-depth based KS and CvM tests are all asymptotically consistent, 
and therefore, a natural question is how is the asymptotic power of the tests under local/contiguous alternatives (see \citet{sidak1999theory} or \citep{8766871}). 
Let $\bbP_{n}$ and $\bbQ_{n}$ be the sequences of the probability measures defined on the sequence of probability spaces $(\Omega_{n}, \sA_{n})$.
Then, $\bbQ_{n}$ is said to be contiguous with respect to $\bbP_{n}$ when $\bbP_{n}\{A_{n}\} \rightarrow 0$ implies that $\bbQ_{n}\{A_{n}\} \rightarrow 0$ for every sequence of measurable sets $A_{n}$. 
It is important to note that the sequence of set $A_{n}$ changes with $n$ along with the $\sigma-$field $\sA_{n}$, and hence, it does not directly follow from the definition of contiguity that any distribution function $\bbQ_{n}$ is contiguous with respect to $\bbP_{n}$.
In order to characterize the contiguity, 
Le Cam proposed some results which are popularly known as ``Le Cam's Lemma''. 
A consequence of Le Cam's first lemma is that the sequence $\bbQ_{n}$ will be contiguous with respect to the sequence $\bbP_{n}$ if $\log(\bbQ_{n}/\bbP_{n})$ is an asymptotically normal random variable with mean $-\sigma^{2}/2$ and variance $\sigma^{2}$, where $\sigma$ is a positive constant. 
Moreover, the consequence of Le Cam's third lemma is that for $\bX_{n} \in \bbR^{d}$, the joint distribution of $\bX_{n}$ and $\log(\bbQ_{n}/\bbP_{n})$ is distributed as multivariate normal with mean vector $(\bmu, -\sigma^{2}/2)^{\tp}$ and covariance $\begin{pmatrix}
\bSigma & \btau\\
\btau^{\tp} & \sigma^{2}
\end{pmatrix}$ under $\bbP_{n}$, then under the alternative distribution $\bbQ_{n}$, the asymptotic distribution of $\bX$ is also a normal with mean $\bmu + \btau$ and covariance $\bSigma$. 
\par
In order to test 
$H_{0}: F = F_{0}$, we consider the sequence of alternatives 
\begin{equation}
\label{eq:alternate}
    H_{n}: F_{n} = \left(1 -\frac{\gamma}{\sqrt{n}}\right)F_{0} + \frac{\gamma}{\sqrt{n}}H
\end{equation}
for a fixed $\gamma > 0$ and $n = 1, 2, \cdots$. 
In terms of depth function, this testing of the hypothesis problem can equivalently be written as 
$H_{0}^{*}: D_{F} = D_{F_{0}}$, and the sequence of alternatives $H_{n}^{*}: D_{F_{n}} = (1 - \frac{\gamma}{\sqrt{n}})D_{F_{0}} + \frac{\gamma}{\sqrt{n}}D_{H}$. 
It follows from the Propositions \ref{proposition:discrete} and \ref{proposition:cont} that the aforesaid hypothesis statement is valid for half-space depth function.
Theorems \ref{thm:contiguous-one} and \ref{thm:contiguous-two} state the asymptotic power properties of the proposed test. 

\begin{theorem}
\label{thm:contiguous-one}
Assume that $F_{0}$ and $H$ (see \eqref{eq:alternate}) have continuous and positive densities $f_{0}$ and $h$, respectively on $\bbR^{d} \, (d \geq 2)$ such that $\E_{H_{0}}\left\{ \frac{h(\bx)}{f_{0}(\bx)} - 1\right\}^{4} < \infty$, and 
suppose that the optimal half-space depth associated to $\bx$ is unique. 
In addition, conditions \ref{cond:smooth}-\ref{cond:dist} hold.
Then the sequence of alternatives is a contiguous sequence. 
Further, assume that $\sG_{1}(\bx)$ is a random element associated with a Gaussian process with mean function zero and covariance kernel  $F_{0}(H[\bx_{1}]\cap H[\bx_{2}])  - F_{0}(H[\bx_{1}])F_{0}(H[\bx_{2}])$, and 
$\sG_{1}'(\bx)$ is a random element associated with a Gaussian process with mean function $-\gamma\E_{\bx \sim h}\{D_{F_{0}}(\bx)\}$ and covariance kernel $F_{0}(H[\bx_{1}]\cap H[\bx_{2}])  - F_{0}(H[\bx_{1}])F_{0}(H[\bx_{2}])$. 
Under the alternatives described in \eqref{eq:alternate}, the asymptotic power of the test based on $T_{\sX, F_{0}}^{\KS}$ is
$\bbP_{\gamma}\left\{\sup\limits_{\bx}|\sG_{1}'(\bx)| > s_{1-\alpha}^{(1)}\right\}$, where 
$s_{1-\alpha}^{(1)}$ is 
such that $\bbP_{\gamma = 0}\left\{\sup\limits_{\bx}|\sG_{1}(\bx)| > s_{1-\alpha}^{(1)}\right\} = \alpha$.
Moreover, under the alternative hypothesis described in \eqref{eq:alternate}, the asymptotic power of the test based on $T_{\sX, F_{0}}^{\CvM}$ is
$\bbP_{\gamma}\left\{\int\limits|\sG_{1}'(\bx)|^{2}dF_{0}(\bx) > s_{1-\alpha}^{(2)}\right\}$, where 
$s_{1-\alpha}^{(2)}$ is 
such that $\bbP_{\gamma = 0}\left\{\int\limits|\sG_{1}(\bx)|^{2}dF_{0}(\bx) > s_{1-\alpha}^{(2)}\right\} = \alpha$.
\end{theorem}
Now consider the scenario of a two-sample problem. The null hypothesis is given by $H_{0}: F = G$ against the sequences of alternatives 
\begin{equation}
\label{eq:alternate-two}
 H_{n,m}: G = \left(1 - \frac{\gamma}{\sqrt{n + m}}\right)F + \frac{\gamma}{\sqrt{n + m}}H   
\end{equation}
for a fixed $\gamma > 0$ and $n, m = 1, 2, \cdots$, which is equivalent to test $H_{0}^{*}: D_{F} = D_{G}$ against the sequence of alternatives $H_{n,m}^{*}: D_{G} = \left(1 - \frac{\gamma}{\sqrt{n + m}}\right)D_{F} + \frac{\gamma}{\sqrt{n + m}}D_{H}$.
\begin{theorem}
\label{thm:contiguous-two}
Assume $F$ and $H$ (see \eqref{eq:alternate-two}) have continuous and positive densities $f$ and $h$, respectively on $\bbR^{d} \, (d \geq 2)$ such that $\E_{F}\left\{ \frac{h(\bx)}{f(\bx)} - 1\right\}^{4} < \infty$. Suppose that the optimal half-space depth associated to $\bx$ is unique and $\lim\limits_{\min(n, m) \rightarrow \infty} \frac{n}{m + n} = \lambda \, \in (0,1)$, and  
in addition, conditions \ref{cond:smooth}-\ref{cond:dist} hold. Then the sequence of alternatives is a contiguous sequence.
Furthermore, assume that $\sG_{2}(\bx)$ is a random element associated with a Gaussian process with mean function zero and covariance kernel $\left\{F(H[\bx_{1}]\cap H[\bx_{2}])  - F(H[\bx_{1}])F(H[\bx_{2}])\right\}/\lambda(1-\lambda)$ and $\sG_{2}'(\bx)$ is a random element associated with a Gaussian process with mean function $\gamma\sqrt{\lambda/(1-\lambda)}\E_{\bu \sim h}\left\{D_{F}(\bx)\right\}$ and covariance kernel $\left\{F(H[\bx_{1}]\cap H[\bx_{2}])  - F(H[\bx_{1}])F(H[\bx_{2}])\right\}/\lambda(1-\lambda)$. 
Under the alternatives described in \eqref{eq:alternate-two}, the asymptotic power of the test based on KS is $\bbP_{\gamma}\left\{\sup_{\bx}|\sG_{2}'(\bx) > t_{1-\alpha}\right\}$ where $t_{1-\alpha}^{(1)}$ is such that $\bbP_{\gamma=0}\{\sup_{\bx}|\sG_{2}(\bx)| > t_{1-\alpha}^{(1)}\} = \alpha$. Moreover, under the alternative hypothesis described in \eqref{eq:alternate-two}, the asymptotic power of the test based on $\CvM$ is 
$\bbP_{\gamma}\left\{\int |\sG_{2}'(\bx)|^{2}d\bx > t^{(2)}_{1-\alpha}\right\}$ where $\bbP_{\gamma = 0}\left\{\int |\sG_{2}'(\bx)|^{2}d\bx > t^{(2)}_{1-\alpha}\right\} = \alpha$. 
\end{theorem}
\par
The results in Theorems \ref{thm:contiguous-one} and \ref{thm:contiguous-two} provide us the asymptotic power of the proposed tests under the local alternatives described in \eqref{eq:alternate} and \eqref{eq:alternate-two}, respectively.
Using those results, we compute the asymptotic powers of the proposed tests for various choices of $\gamma$. In this study, we consider that $F_{0}$ is the standard bivariate normal distribution, and $H$ is the standard bivariate Laplace distribution. Figure \ref{fig:conitiguous} illustrates the summarised results, and it is clearly indicated by those diagrams that our proposed tests perform well for this example. 
For the sake of concise presentation, we are not reporting here the results for other choices of $F_{0}$ and $H$, and some other choices of dimension. Nevertheless, our preliminary investigation suggests that the tests perform well for alternative choices.

\begin{figure}[t!]
    \centering
    \begin{subfigure}{0.49\textwidth}
         \centering
         \includegraphics[width=\textwidth]{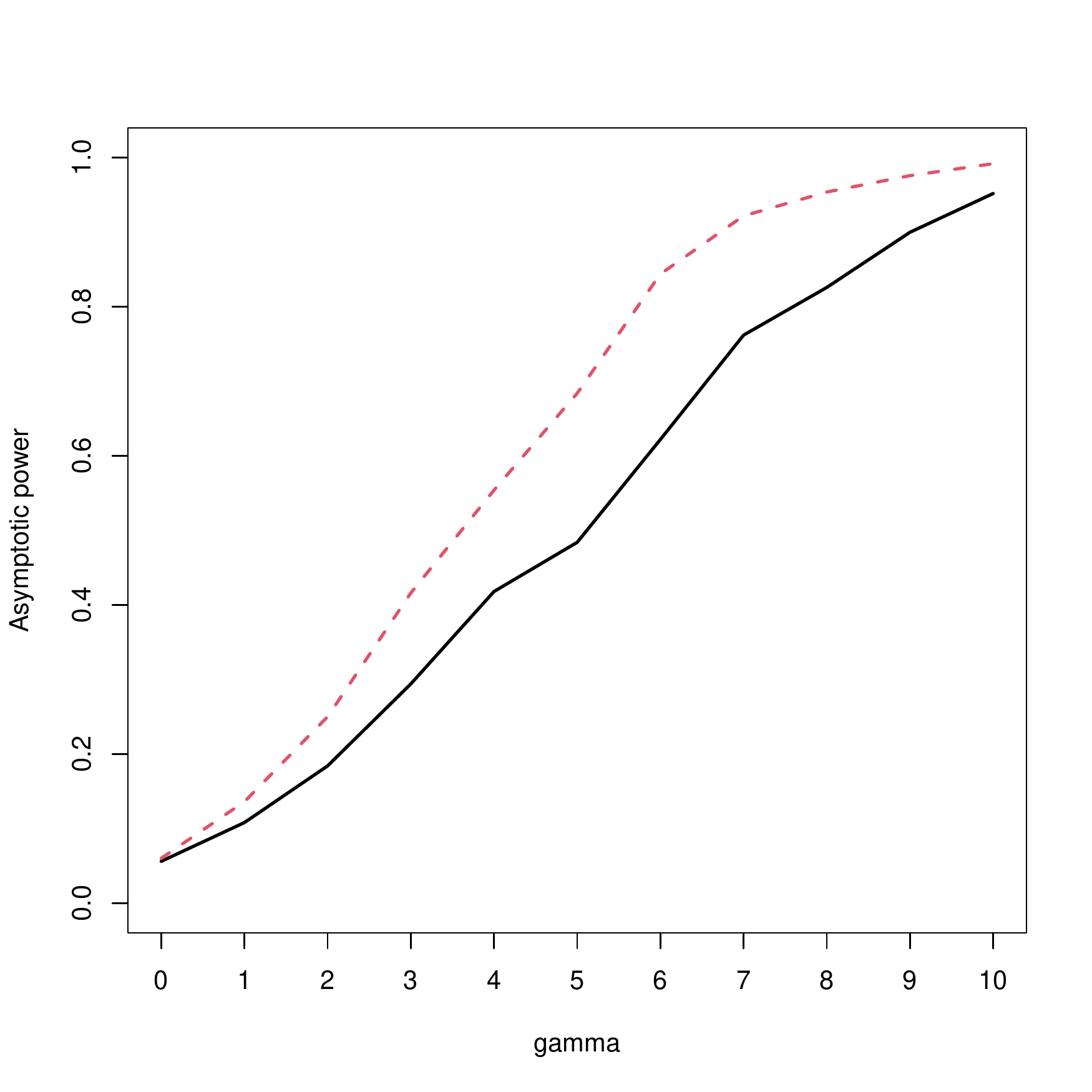}
         \caption{Goodness-of-fit test}
         \label{fig:1conitiguous}
     \end{subfigure}
     \hfill
     \begin{subfigure}{0.49\textwidth}
         \centering
         \includegraphics[width=\textwidth]{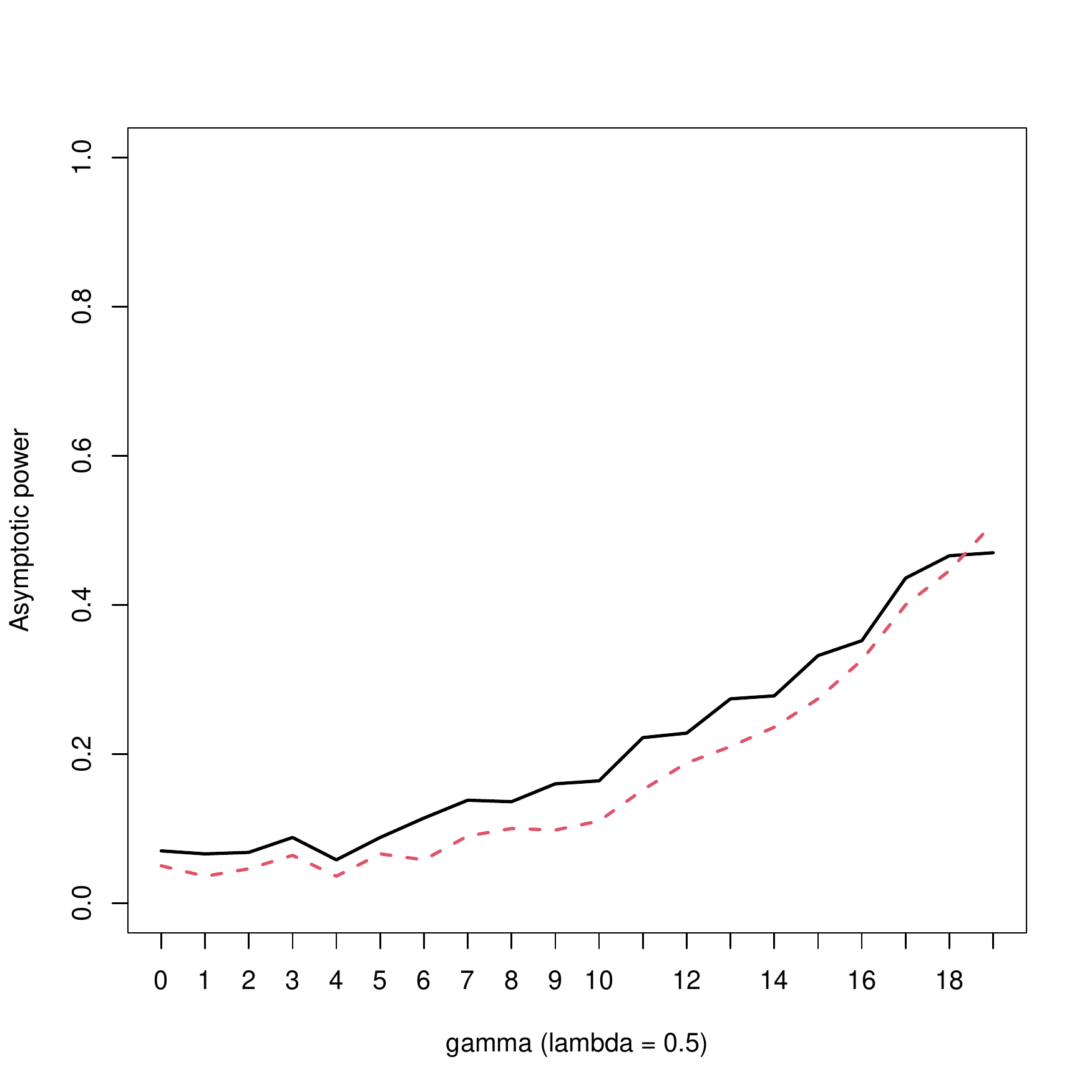}
         \caption{Two-sample test}
         \label{fig:2conitiguous}
     \end{subfigure}
     \caption{The asymptotic power for (a) goodness-of-fit test and (b) two-sample test with $\lambda=0.5$ under contiguous alternative described in \eqref{eq:alternate} and \eqref{eq:alternate-two}, respectively, where $F_{0}$ is standard bivariate normal distribution and $H$ is the standard Laplace distribution. The black solid line indicates the results based on proposed depth-based KS test statistics and the red dotted line indicates the results based on proposed depth-based CvM test statistics.}
    \label{fig:conitiguous}
\end{figure}

\section{Conclusion}
\label{sec:conclusion}
In this paper, we have proposed a data-depth discrepancy and a graphical tool-kit based on Tukey's half-space depth. This device is used to test the equality of two distribution functions/goodness of fit problem and is applicable to any dimension of the distributions. 
The motivations behind the graphical device are shown through simulation examples. 
Influenced by the graphical device, we have proposed test statistics based on the Kolmogorov-Smirnov and Cram\'er-von Mises tests. 
For goodness-of-fit and two-sample testing problems, the proposed test statistics can test the equality of two distribution functions, which are common in practice. We have shown that the test statistics are pointwise and also uniformly constant. Moreover, we have studied the asymptotic power under contiguous alternatives. 
The applicability of the proposed method is illustrated by simulation studies and real data analysis.
The graphical tool and test statistics developed in this article can be used regardless of the dimension of the data and therefore, significantly contribute to statistics and machine learning literature.

	\section*{Supplementary Materials}
	The online supplementary material contains an illustration of the proposed graphical tool-kit for high-dimension data in Section \ref{sec:high-dim}; the finite sample performance is presented in Section \ref{sec:simulation}; in Section \ref{sec:real-data}, we implement the proposed test on two interesting data sets. Moreover, technical details and mathematical proofs of the proposed theorems are discussed in Section \ref{sec:proofs}.
	\par
	\section*{Acknowledgements}
	
	Subhra Sankar Dhar acknowledges the research grant CRG/2022/001489, Government of India.
	\par
	

\bibhang=1.7pc
\bibsep=2pt
\fontsize{9}{14pt plus.8pt minus .6pt}\selectfont
\renewcommand\bibname{\large \bf References}
\expandafter\ifx\csname
natexlab\endcsname\relax\def\natexlab#1{#1}\fi
\expandafter\ifx\csname url\endcsname\relax
  \def\url#1{\texttt{#1}}\fi
\expandafter\ifx\csname urlprefix\endcsname\relax\def\urlprefix{URL}\fi

\bibliographystyle{chicago}      
\bibliography{ref}   

\vskip .65cm
\noindent
Pratim Guha Niyogi\\
Department of Biostatistics, Johns Hopkins University\\
\noindent
E-mail: (pnyogi1@jhmi.edu)
\vskip 2pt

\noindent
Subhra Sankar Dhar\\
Department of Mathematics and Statistics, Indian Institute of Technology\\
\noindent
E-mail: (subhra@iitk.ac.in)

\begin{supplementary}
\section{Illustration of graphical tool-kit: High-dimension data}
\label{sec:high-dim}
In the last decades or so, there have been plenty of applications involving high-dimensional data, and we here illustrate the usefulness of proposed graphical tool-kits to meet the challenges of modern data applications.
By varying the dimensions of the data sets, $d = 15, 20, 30, 40, 50$ and $60$, we generate \textit{sample-1} from $d$-variate normal distribution and \textit{sample-2} from $d$-variate Cauchy distribution with sample size $10$ each for a two-sample testing problem. Note that here, the dimension is larger than the sample size. The data-depth discrepancy plots for different dimensions are shown in Figure \ref{fig:high-dim} which illustrates the fact that differences between the depth values of two samples are not clustered around the horizontal axis. This indicates that the proposed graphical tool-kit is useful for high-dimensional data. 

\begin{figure}[t!]
    \centering
    \begin{subfigure}{0.32\textwidth}
         \centering
         \includegraphics[width=\textwidth]{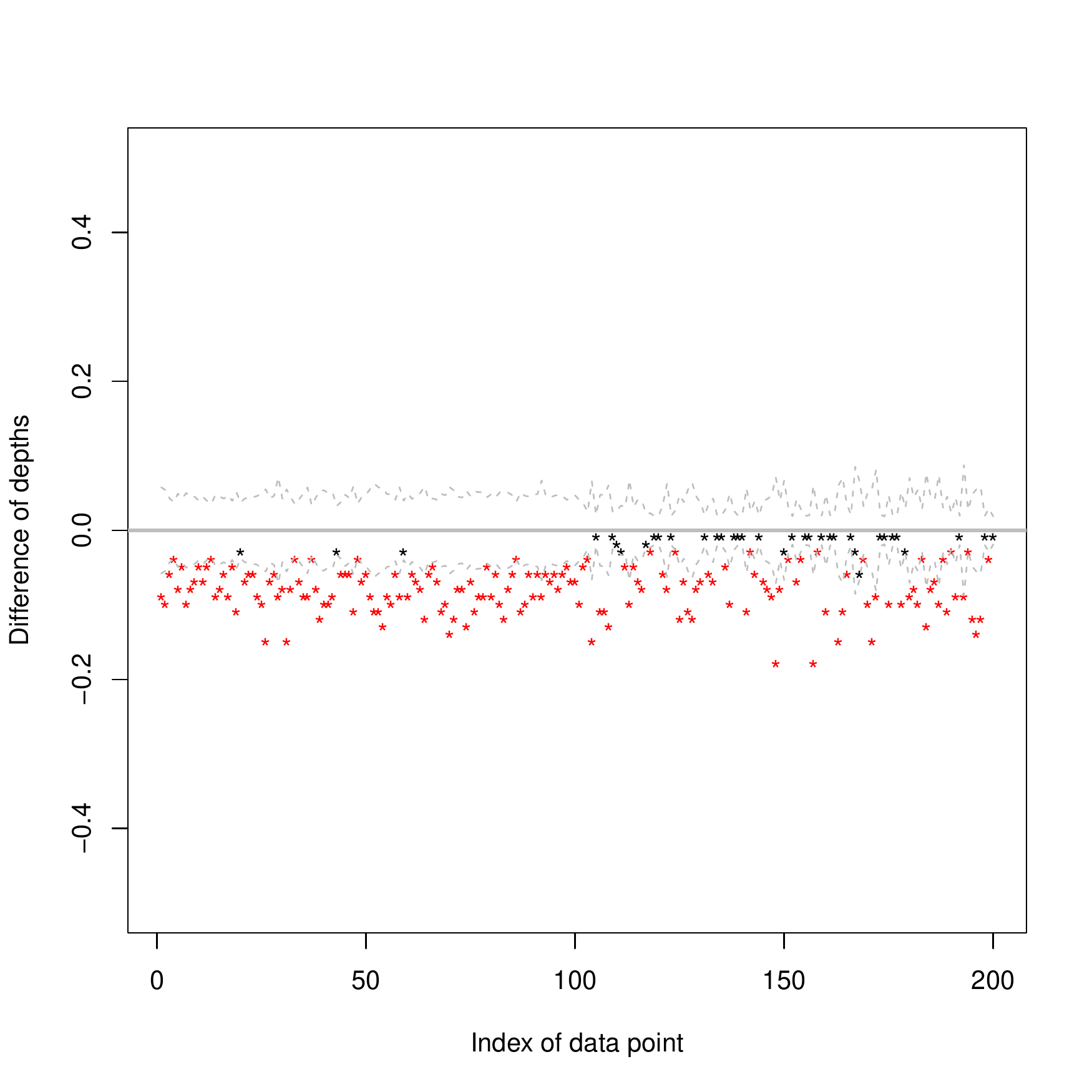}
         \caption{$d = 15$}
         \label{fig:normal-cauchy-10}
     \end{subfigure}
     \begin{subfigure}{0.32\textwidth}
         \centering
         \includegraphics[width=\textwidth]{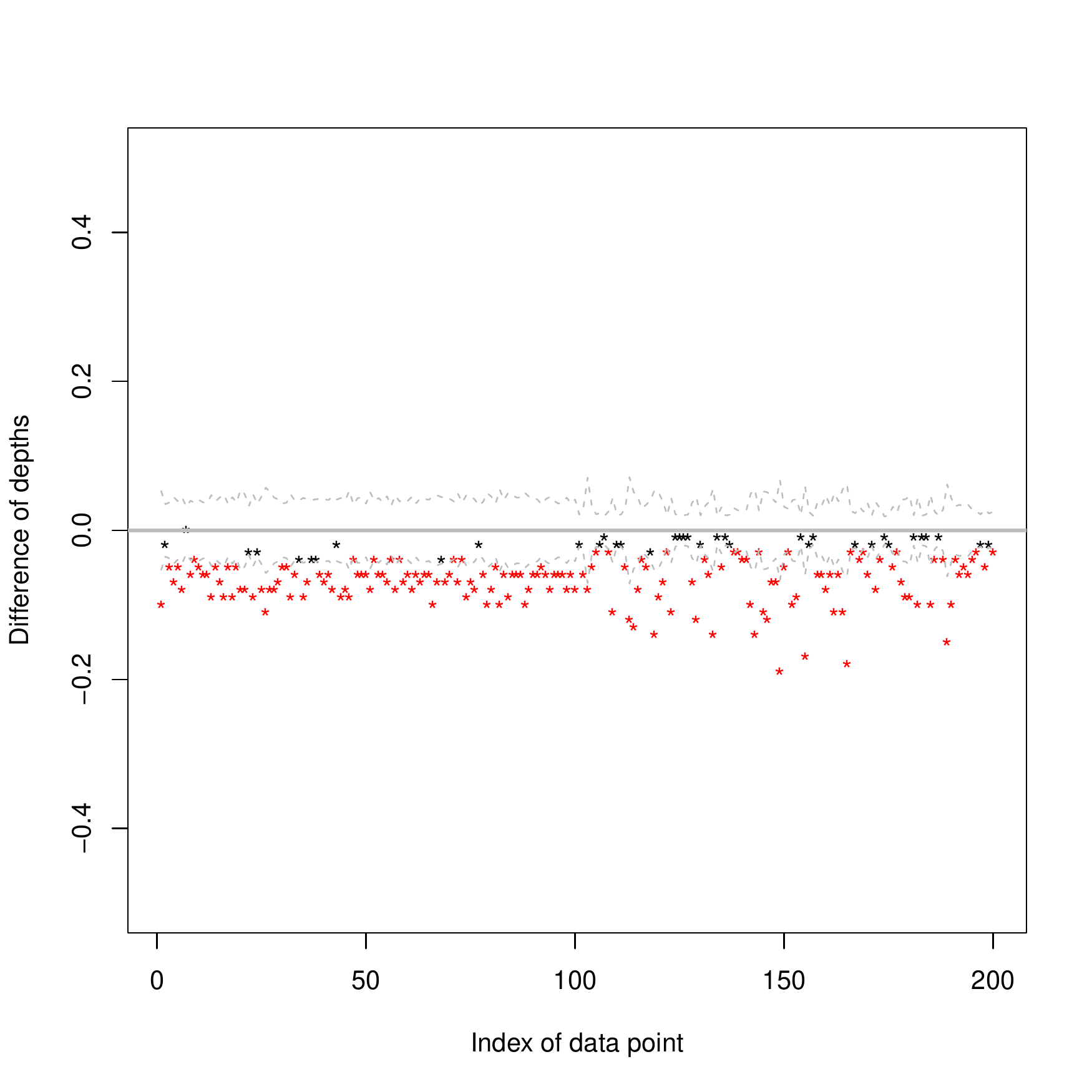}
         \caption{$d = 20$}
         \label{fig:normal-cauchy-20}
     \end{subfigure}
     \begin{subfigure}{0.32\textwidth}
         \centering
         \includegraphics[width=\textwidth]{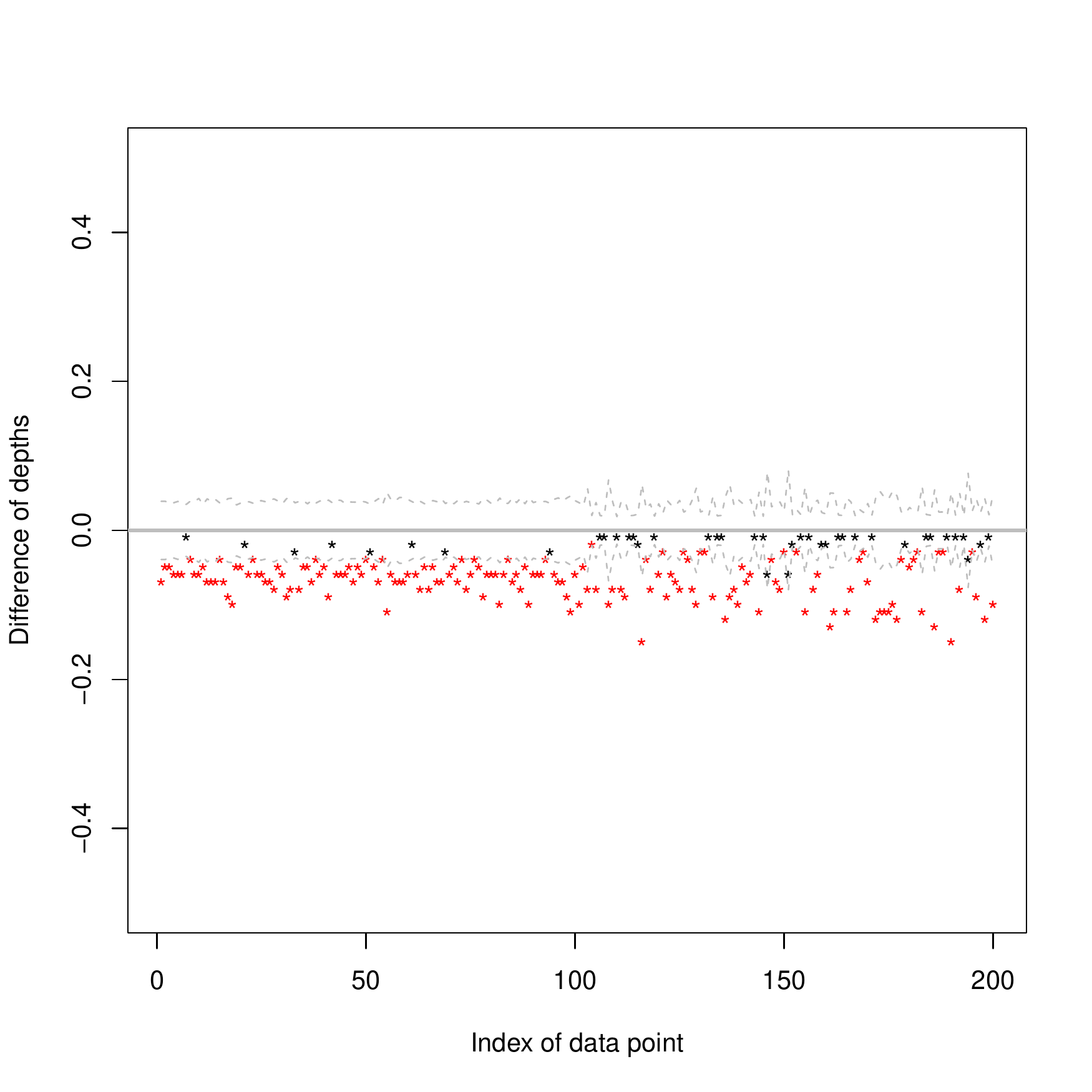}
         \caption{$d = 30$}
         \label{fig:normal-cauchy-30}
     \end{subfigure}
     \begin{subfigure}{0.32\textwidth}
         \centering
         \includegraphics[width=\textwidth]{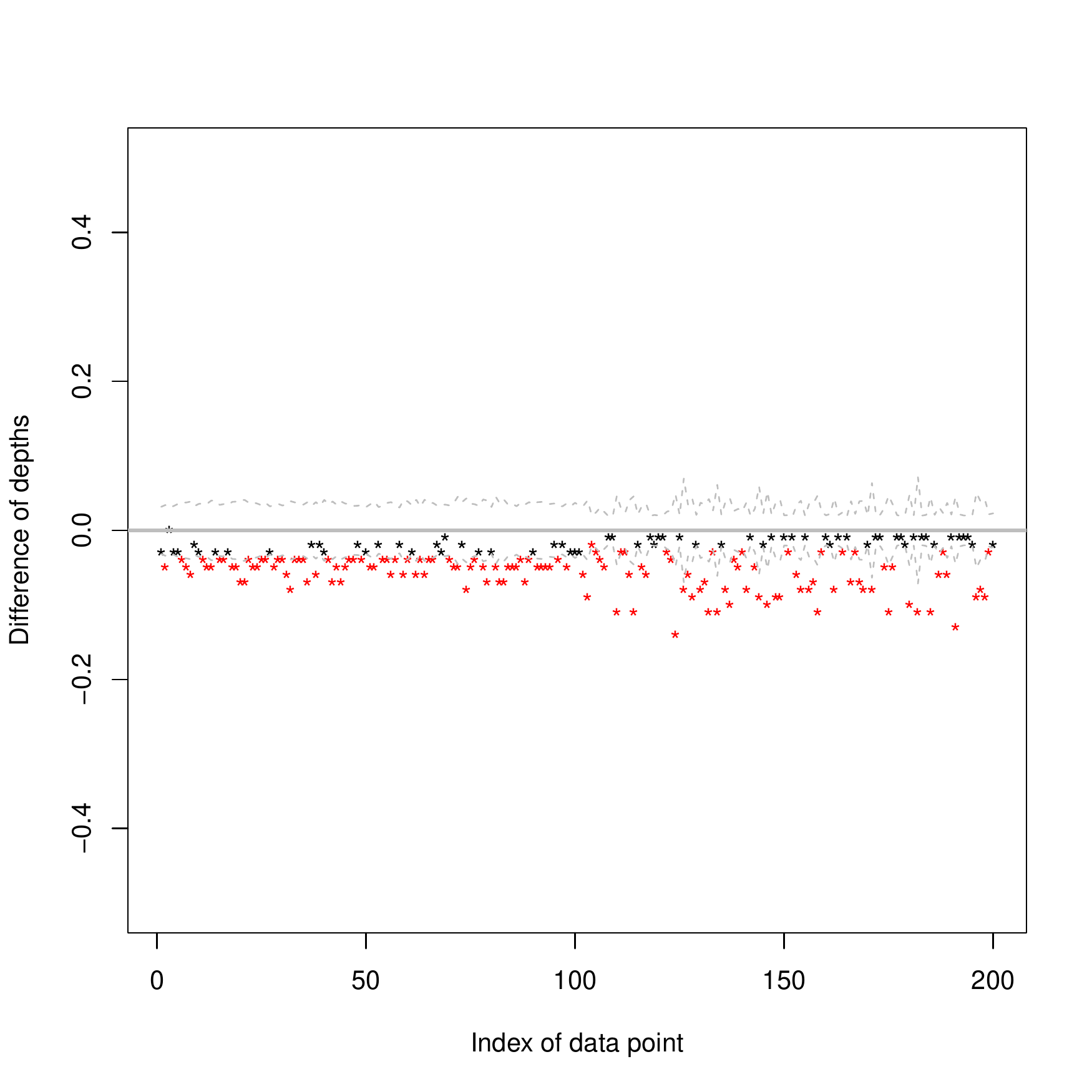}
         \caption{$d = 40$}
         \label{fig:normal-cauchy-40}
     \end{subfigure}
     \begin{subfigure}{0.32\textwidth}
         \centering
         \includegraphics[width=\textwidth]{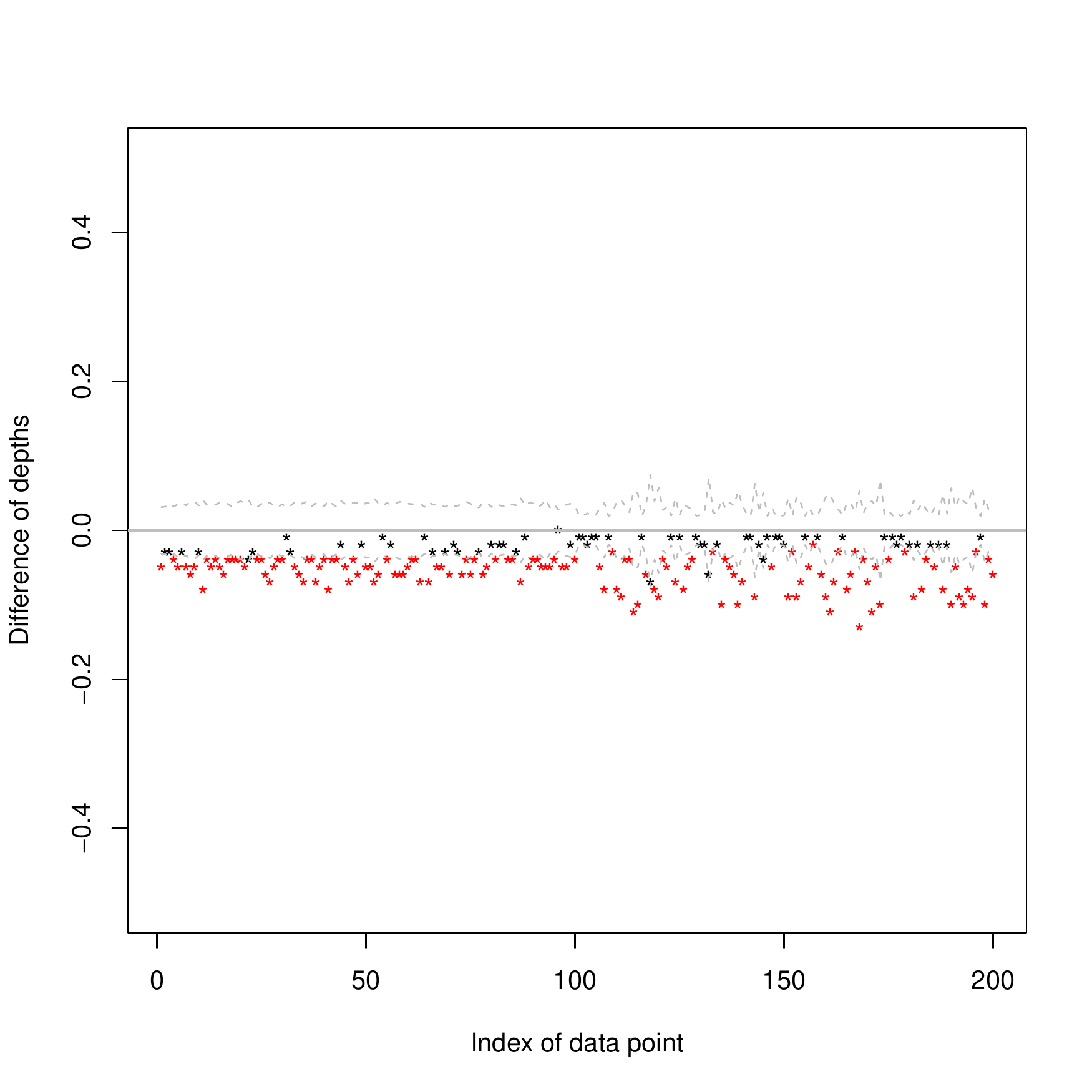}
         \caption{$d = 50$}
         \label{fig:normal-cauchy-50}
     \end{subfigure}
     \begin{subfigure}{0.32\textwidth}
         \centering
         \includegraphics[width=\textwidth]{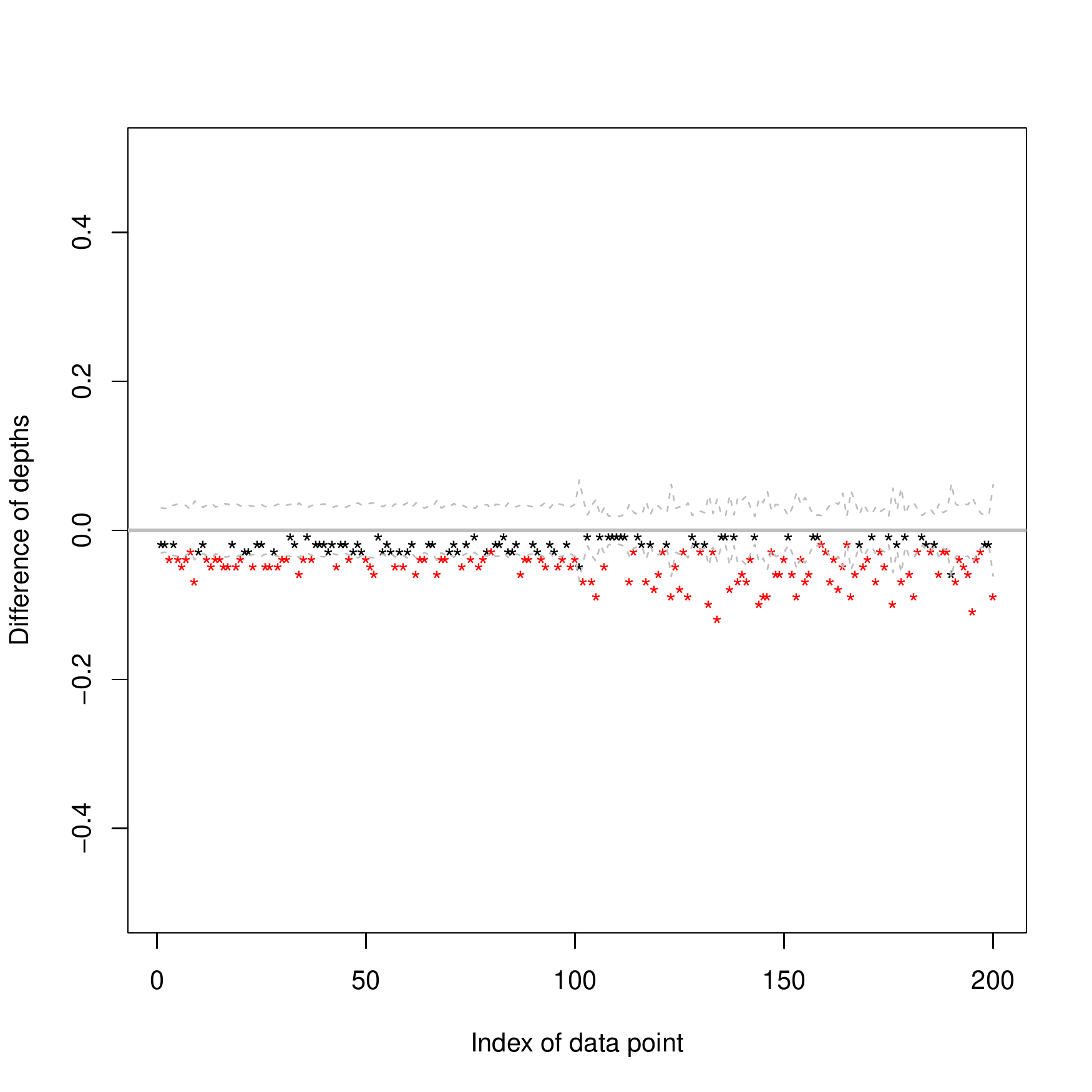}
         \caption{$d = 60$}
         \label{fig:normal-cauchy-60}
     \end{subfigure}
    \caption{The data-depth discrepancy plot for the two-sample examples in high-dimensional situations for different dimensions. Here samples are from multivariate normal and Cauchy distributions. The dotted gray curves indicate the two-sigma limits of data-depth discrepancy. Points that are outside the two-sigma limits are color-coded in red.}
    \label{fig:high-dim}
\end{figure}

\section{Finite sample level and power studies}
\label{sec:simulation}
In this section, we demonstrate the finite sample performance by reporting the size and power of the proposed tests obtained from multiple studies compared to other existing methods. Precisely speaking, we use \texttt{mvtnorm} and \texttt{LaplacesDemon} to generate the data from multivariate normal, t, and Laplace distributions, respectively.
Next, in order to compute the half-space depth, the package, viz., \texttt{ddalpha} is used. Next, in order to implement the other tests, \texttt{mvnTest} is used to run goodness-of-fit tests. All the simulation results in the following examples are based on 500 replications, and the critical value of the test is estimated by 500 bootstrap samples in each of the simulation runs. We also study the simulation results for cases with different dimensions, choosing $d \in \{ 2, 6, 10\}$. 
Furthermore, we denote $\bmu\in \bbR^{d}$ as a location parameter, and $\bSigma$ is a $d\times d$ positive definite matrix as a scale parameter in those examples. 
\par
In the goodness-of-fit problem (see Problem \ref{problem:one-sample}), 
we assume that $F_{0}$ is a standard $d$-dimensional normal distribution. 
The choices of unknown distributions described in Problem \ref{problem:one-sample} are listed below. 
\begin{enumerate}[label=Model A.\arabic*., align=left]
    \item Standard $d$-variate normal distribution $N_{d}(\zero_{d}, \bI_{d})$, where 
    $\zero_{d} = (0, \cdots, 0)^{\tp}$ is the $d$-dimensional null vector, and 
    $\bI_{d}$ refers to the $d\times d$ identity matrix. 
    \item Mixture of the $d$-variate normal distribution with mixing probability $0.8$, where the samples are taken from $0.8 N_{d}(\zero_{d},\bI_{d}) + 0.2 N_{d}(5\times\textbf{1}_{d}, \bI_{d})$. 
    \item Mixture of the $d$-variate normal distribution with mixing probability $0.8$, where the samples are taken as a from $0.8N_{d}(\zero_{d}, \bI_{d}) + 0.2 N_{d}(\textbf{0}_{d}, \bSigma)$. Here, $\bSigma$ be a compound symmetric matrix with off-diagonal elements as $0.5$ and diagonal elements as $1$.
    \item $d$-variate t-distribution with degrees of freedom 3, $t(\zero_{d}, \bI_{d}, 3)$. 
    \item Standard $d$-variate Cauchy distribution $C(\zero_{d}, \bI_{d})$ or $t(\zero_{d}, \bI_{d}, 1)$ with p.d.f. $f(\bx) = [\Gamma((d+1)/2)/(\sqrt{\pi}\Gamma(d/2))]\times (1+\|\bx\|^{2})^{-(d+1)/2}$, where ${\bf x}\in\mathbb{R}^{d}$.
    \item Standard $d$-variate Laplace distribution $L(\zero_{d}, \bI_{d})$ with p.d.f. $f(\bx) = [\Gamma(d/2)/(2\Gamma(d)\pi^{d/2})]\times \exp(-\|\bx\|)$, where ${\bf x}\in\mathbb{R}^{d}$. 
\end{enumerate}
We fix the sample size $n \in \{25, 50, 100, 300, 500\}$ and the nominal significance level $\alpha$ is assumed to be $0.05$. We compare the proposed method denoted as KS.depth and CvM.depth respectively with 
Anderson-Darling (AD), 
Cram\'er-von Mises (CvM),
Doornik-Hansen (DH), 
Henze-Zirkler (HZ), 
Royston (R) tests
and a series of $\chi^{2}$-type tests such as McCulloch (M), Nikulin-Rao-Robson (NRR), Dzhaparidze-Nikulin (DN) tests. 

\begin{table}
\footnotesize
\centering
\caption{Goodness-of-fit test: Observed relative frequency for rejecting the null, $d=2$.}
\label{tab:one-sample2}
\begin{tabular}{rrrrrrrrrrr}
\toprule
$n$ & KS.depth & CvM.depth & AD & CM & DH & HZ & R & McCulloch & NRR & DN\\
\midrule
\multicolumn{11}{l}{Model A.1. $H_{0}: \bX \sim N_{2}(\zero_{2}, \bI_{2})$ against $H_{1}: \bX \sim N_{2}(\zero_{2}, \bI_{2})$}\\
\midrule
  25 & 0.036 & 0.056 & 0.040 & 0.024 & 0.054 & 0.042 & 0.056 & 0.036 & 0.036 & 0.038 \\ 
  50 & 0.050 & 0.064 & 0.052 & 0.048 & 0.054 & 0.060 & 0.050 & 0.030 & 0.056 & 0.058 \\ 
  100 & 0.056 & 0.064 & 0.048 & 0.052 & 0.060 & 0.052 & 0.062 & 0.054 & 0.050 & 0.042 \\ 
  300 & 0.062 & 0.066 & 0.050 & 0.042 & 0.072 & 0.044 & 0.052 & 0.030 & 0.046 & 0.052 \\ 
  500 & 0.074 & 0.056 & 0.036 & 0.036 & 0.062 & 0.042 & 0.056 & 0.052 & 0.044 & 0.060 \\ 
\midrule
\multicolumn{11}{l}{Model A.2. $H_{0}: \bX \sim N_{2}(\zero_{2}, \bI_{2})$ against $H_{1}: \bX \sim 0.8N_{2}(\zero_{2}, \bI_{2}) + 0.2N_{2}(5\textbf{1}_{2}, \bI_{2})$}\\
\midrule
  25 & 0.712 & 0.938 & 0.202 & 0.180 & 0.140 & 0.970 & 0.996 & 0.118 & 0.134 & 0.090 \\ 
  50 & 0.918 & 0.994 & 0.252 & 0.250 & 0.204 & 1.000 & 1.000 & 0.100 & 0.154 & 0.124 \\ 
  100 & 0.986 & 1.000 & 0.288 & 0.326 & 0.488 & 1.000 & 1.000 & 0.138 & 0.260 & 0.236 \\ 
  300 & 1.000 & 1.000 & 0.632 & 0.636 & 0.986 & 1.000 & 1.000 & 0.148 & 0.658 & 0.626 \\ 
  500 & 1.000 & 1.000 & 0.808 & 0.810 & 1.000 & 1.000 & 1.000 & 0.160 & 0.834 & 0.824 \\ 
\midrule
\multicolumn{11}{l}{Model A.3. $H_{0}: \bX \sim N_{2}(\zero_{2}, \bI_{2})$ against $H_{1}: \bX \sim 0.8N_{2}(\zero_{2}, \bI_{2}) + 0.2N_{2}(\zero_{2}, \bSigma)$}\\
\midrule
  25 & 0.470 & 0.670 & 0.062 & 0.054 & 0.052 & 0.038 & 0.060 & 0.032 & 0.044 & 0.056 \\ 
  50 & 0.564 & 0.792 & 0.072 & 0.072 & 0.054 & 0.050 & 0.058 & 0.052 & 0.058 & 0.052 \\ 
  100 & 0.550 & 0.786 & 0.048 & 0.042 & 0.054 & 0.036 & 0.036 & 0.054 & 0.042 & 0.036 \\ 
  300 & 0.610 & 0.852 & 0.068 & 0.042 & 0.032 & 0.058 & 0.028 & 0.050 & 0.034 & 0.048 \\ 
  500 & 0.612 & 0.874 & 0.074 & 0.064 & 0.050 & 0.056 & 0.048 & 0.068 & 0.056 & 0.056 \\  
\midrule
\multicolumn{11}{l}{Model A.4. $H_{0}: \bX \sim N_{2}(\zero_{2}, \bI_{2})$ against $H_{1}: \bX \sim t(\zero_{2}, \bI_{2}, 3)$}\\
\midrule
  25 & 0.514 & 0.698 & 0.530 & 0.488 & 0.620 & 0.588 & 0.542 & 0.472 & 0.396 & 0.202 \\ 
  50 & 0.770 & 0.924 & 0.866 & 0.852 & 0.874 & 0.832 & 0.806 & 0.812 & 0.774 & 0.446 \\ 
  100 & 0.940 & 0.994 & 0.986 & 0.986 & 0.984 & 0.964 & 0.954 & 0.970 & 0.968 & 0.772 \\ 
  300 & 1.000 & 1.000 & 1.000 & 1.000 & 1.000 & 1.000 & 1.000 & 1.000 & 1.000 & 1.000 \\ 
  500 & 1.000 & 1.000 & 1.000 & 1.000 & 1.000 & 1.000 & 1.000 & 1.000 & 1.000 & 1.000 \\ 
\midrule
\multicolumn{11}{l}{Model A.5. $H_{0}: \bX \sim N_{2}(\zero_{2}, \bI_{2})$ against $H_{1}: \bX \sim C(\zero_{2}, \bI_{2})$}\\
\midrule
  25 & 0.970 & 0.996 & 0.988 & 0.988 & 0.972 & 0.990 & 0.968 & 0.944 & 0.964 & 0.888 \\ 
  50 & 1.000 & 1.000 & 1.000 & 1.000 & 1.000 & 1.000 & 1.000 & 1.000 & 1.000 & 0.998 \\ 
  100 & 1.000 & 1.000 & 1.000 & 1.000 & 1.000 & 1.000 & 1.000 & 1.000 & 1.000 & 1.000 \\ 
  300 & 1.000 & 1.000 & 1.000 & 1.000 & 1.000 & 1.000 & 1.000 & 1.000 & 1.000 & 1.000 \\ 
  500 & 1.000 & 1.000 & 1.000 & 1.000 & 1.000 & 1.000 & 1.000 & 1.000 & 1.000 & 1.000 \\
\midrule
\multicolumn{11}{l}{Model A.6. $H_{0}: \bX \sim N_{2}(\zero_{2}, \bI_{2})$ against $H_{1}: \bX \sim L(\zero_{2}, \bI_{2})$}\\
\midrule
  25 & 0.404 & 0.652 & 0.562 & 0.492 & 0.484 & 0.530 & 0.402 & 0.254 & 0.324 & 0.224 \\ 
  50 & 0.708 & 0.914 & 0.884 & 0.852 & 0.740 & 0.818 & 0.702 & 0.610 & 0.722 & 0.536 \\ 
  100 & 0.954 & 0.994 & 0.994 & 0.990 & 0.942 & 0.980 & 0.936 & 0.896 & 0.972 & 0.876 \\ 
  300 & 1.000 & 1.000 & 1.000 & 1.000 & 1.000 & 1.000 & 1.000 & 1.000 & 1.000 & 1.000 \\ 
  500 & 1.000 & 1.000 & 1.000 & 1.000 & 1.000 & 1.000 & 1.000 & 1.000 & 1.000 & 1.000 \\ 
\bottomrule
\end{tabular}
\end{table}
\par
Tables \ref{tab:one-sample2}, \ref{tab:one-sample6} and \ref{tab:one-sample10} describe Monte Carlo results for goodness-of-fit tests that represents the observed relative frequency for rejecting $H_0$ for $d=2,6$ and $10$, respectively.
Here, the observed relative frequency for rejecting the null for Model A.1 indicates the size of the test and for Model A.2-A.6
represent the power of the test for different alternatives.
If we increase the sample size $n$, the power of the test increases for all the methods and the size of the test are more or less stabilized at the nominal significance level.
If we increase the dimension of the data, the aforementioned statement holds true. The power of the proposed tests are consistently better for all the choices of alternative distributions. Moreover, for mixture distributions like Model A.2 and A.3, the competing testing procedures are worse, whereas the proposed methods KS.depth and CvM.depth perform satisfactorily. The Figure \ref{fig:auc1} represents the receiver operating characteristics (ROC) curve that illustrates the performance of the testing procedure. The ROC curve is created by plotting the power as a function of the Type I error of the testing procedure. For brevity, here we represent the ROC curve for the goodness-of-fit test when the alternative hypotheses are (a) Cauchy and (d) Laplace distributions for small sample size.
Here, this figure tells us that our proposed testing procedures (especially the data-depth based CvM method) outperform the existing methods.
\par
In the two-sample test scenario (see Problem \ref{problem:two-sample}), we use sample sizes $(n, m)$ such that $\lambda = n/(n+m) \in \{ 0.3, 0.5, 0.8\}$ and fix $n \in \{50, 100, 300\}$. We generate samples from the following distribution family.
\begin{enumerate}[align=left]
    \item[Model B.] The first samples are taken from the standard multivariate distribution $N_{d}(\zero_{d}, \bI_{d})$ and the second sample is taken from the $d$ variate normal $N_{d}(\mu\times\textbf{1}_{d}, \bI_{d})$ where we chose $\mu \in \{0, 0.5,0 1\}$. 
\end{enumerate}
Table \ref{tab:two} represents the observed relative frequency for rejecting $H_0$ based on the above model. 
The values on the table corresponding to $\mu=0$ is the size of the test and that to $\mu=\mu_{0} \in \{0.5,1\}$ represent the power of the test at $\mu_0$. If $\mu$ creases, so do the powers of the tests.
Even if the dimension of the data is increased, the performance of the proposed test is not affected.
\begin{table}
\footnotesize
\centering
\caption{Two-sample test: Observed relative frequency for rejecting the null based on Model B.}
\label{tab:two}
\begin{tabular}{rrrrrrrrr}
\toprule
&&& \multicolumn{2}{c}{$\mu = 0$} & \multicolumn{2}{c}{$\mu = 0.5$} & \multicolumn{2}{c}{$\mu = 1$}\\
\cmidrule(lr){4-5}\cmidrule(lr){6-7}\cmidrule(lr){8-9}
 & n & m & KS.depth & CvM.depth & KS.depth & CvM.depth & KS.depth & CvM.depth \\ 
\midrule
\multicolumn{9}{l}{$d = 2$}\\
\midrule
\multirow{3}{*}{$\rho = 0.3$}& 
  50 & 117 & 0.056 & 0.056 & 0.914 & 0.970 & 1.000 & 1.000 \\ 
  &100 & 234 & 0.038 & 0.048 & 1.000 & 1.000 & 1.000 & 1.000 \\ 
  &300 & 700 & 0.034 & 0.044 & 1.000 & 1.000 & 1.000 & 1.000 \\ 
\hdashline
\multirow{3}{*}{$\rho = 0.5$}& 
  50 & 50 & 0.086 & 0.052 & 0.802 & 0.844 & 1.000 & 1.000 \\ 
  &100 & 100 & 0.044 & 0.046 & 0.988 & 0.996 & 1.000 & 1.000 \\ 
  &300 & 300 & 0.040 & 0.052 & 1.000 & 1.000 & 1.000 & 1.000 \\ 
\hdashline
\multirow{3}{*}{$\rho = 0.8$}& 
  50 & 13 & 0.086 & 0.036 & 0.392 & 0.344 & 0.888 & 0.920 \\ 
  &100 & 25 & 0.076 & 0.056 & 0.676 & 0.758 & 0.996 & 0.998 \\ 
  &300 & 75 & 0.072 & 0.056 & 0.996 & 0.998 & 1.000 & 1.000 \\ 
\midrule
\multicolumn{9}{l}{$d = 6$}\\
\midrule
\multirow{3}{*}{$\rho = 0.3$}& 
   50 & 117 & 0.056 & 0.014 & 0.976 & 0.958 & 1.000 & 1.000 \\ 
  &100 & 234 & 0.062 & 0.028 & 1.000 & 1.000 & 1.000 & 1.000 \\ 
  &300 & 700 & 0.060 & 0.052 & 1.000 & 1.000 & 1.000 & 1.000 \\ 
\hdashline
\multirow{3}{*}{$\rho = 0.5$}& 
  50 & 50 & 0.120 & 0.000 & 0.990 & 0.384 & 1.000 & 0.854 \\ 
  &100 & 100 & 0.068 & 0.010 & 1.000 & 1.000 & 1.000 & 1.000 \\ 
  &300 & 300 & 0.056 & 0.042 & 1.000 & 1.000 & 1.000 & 1.000 \\ 
\hdashline
\multirow{3}{*}{$\rho = 0.8$}& 
  50 & 13 & 0.328 & 0.000 & 0.826 & 0.000 & 0.900 & 0.000 \\ 
  &100 & 25 & 0.130 & 0.004 & 0.954 & 0.238 & 1.000 & 0.650 \\ 
  &300 & 75 & 0.072 & 0.042 & 1.000 & 1.000 & 1.000 & 1.000 \\ 
 \midrule
\multicolumn{9}{l}{$d = 10$}\\
\midrule
\multirow{3}{*}{$\rho = 0.3$}& 
  50 & 117 & 0.102 & 0.000 & 0.998 & 0.038 & 1.000 & 0.152 \\ 
  &100 & 234 & 0.054 & 0.002 & 1.000 & 0.998 & 1.000 & 1.000 \\ 
  &300 & 700 & 0.044 & 0.010 & 1.000 & 1.000 & 1.000 & 1.000 \\ 
\hdashline
\multirow{3}{*}{$\rho = 0.5$}& 
  50 & 50 & 0.128 & 0.000 & 1.000 & 0.000 & 1.000 & 0.000 \\ 
  &100 & 100 & 0.070 & 0.000 & 1.000 & 0.498 & 1.000 & 0.878 \\ 
  &300 & 300 & 0.042 & 0.004 & 1.000 & 1.000 & 1.000 & 1.000 \\ 
\hdashline
\multirow{3}{*}{$\rho = 0.8$}& 
  50 & 13 & 0.546 & 0.000 & 0.898 & 0.000 & 0.896 & 0.000 \\ 
  &100 & 25 & 0.134 & 0.000 & 0.994 & 0.000 & 1.000 & 0.000 \\ 
  &300 & 75 & 0.068 & 0.004 & 1.000 & 0.986 & 1.000 & 1.000 \\ 
\bottomrule
\end{tabular}
\end{table}

\begin{table}
\footnotesize
\centering
\caption{Goodness-of-fit test: Observed relative frequency for rejecting the null, $d=6$.}
\label{tab:one-sample6}
\begin{tabular}{rrrrrrrrrrr}
\toprule
$n$ & KS.depth & CvM.depth & AD & CM & DH & HZ & R & McCulloch & NRR & DN \\ 
\midrule
\multicolumn{11}{l}{Model A.1. $H_{0}: \bX \sim N_{6}(\zero_{6}, \bI_{6})$ against $H_{1}: \bX \sim N_{6}(\zero_{6}, \bI_{6})$}\\
\midrule
  25 & 0.064 & 0.032 & 0.048 & 0.048 & 0.060 & 0.040 & 0.050 & 0.044 & 0.052 & 0.070 \\ 
  50 & 0.044 & 0.054 & 0.054 & 0.050 & 0.044 & 0.064 & 0.048 & 0.048 & 0.070 & 0.074 \\ 
  100 & 0.034 & 0.022 & 0.052 & 0.054 & 0.040 & 0.036 & 0.020 & 0.058 & 0.060 & 0.062 \\ 
  300 & 0.030 & 0.010 & 0.050 & 0.048 & 0.048 & 0.050 & 0.036 & 0.046 & 0.058 & 0.062 \\ 
  500 & 0.044 & 0.040 & 0.036 & 0.034 & 0.042 & 0.048 & 0.054 & 0.036 & 0.040 & 0.046 \\
\midrule
\multicolumn{11}{l}{Model A.2. $H_{0}: \bX \sim N_{6}(\zero_{6}, \bI_{6})$ against $H_{1}: \bX \sim 0.8N_{6}(\zero_{6}, \bI_{6}) + 0.2N_{6}(5\textbf{1}_{6}, \bI_{6})$}\\
\midrule
  25 & 0.968 & 0.980 & 0.039 & 0.039 & 0.039 & 0.324 & 1.000 & 0.039 & 0.073 & 0.078 \\ 
  50 & 0.994 & 1.000 & 0.090 & 0.080 & 0.068 & 0.932 & 1.000 & 0.046 & 0.066 & 0.070 \\ 
  100 & 0.998 & 1.000 & 0.084 & 0.072 & 0.066 & 1.000 & 1.000 & 0.064 & 0.058 & 0.066 \\ 
  300 & 1.000 & 1.000 & 0.146 & 0.138 & 0.108 & 1.000 & 1.000 & 0.056 & 0.112 & 0.104 \\ 
  500 & 1.000 & 1.000 & 0.186 & 0.176 & 0.172 & 1.000 & 1.000 & 0.096 & 0.168 & 0.160 \\ 
\midrule
\multicolumn{11}{l}{Model A.3. $H_{0}: \bX \sim N_{6}(\zero_{6}, \bI_{6})$ against $H_{1}: \bX \sim 0.8N_{6}(\zero_{6}, \bI_{6}) + 0.2N_{6}(\zero_{6}, \bSigma)$}\\
\midrule
  25 & 0.954 & 0.956 & 0.022 & 0.022 & 0.046 & 0.049 & 0.037 & 0.046 & 0.051 & 0.066 \\ 
  50 & 0.954 & 0.984 & 0.032 & 0.024 & 0.040 & 0.060 & 0.046 & 0.044 & 0.044 & 0.056 \\ 
  100 & 0.960 & 0.996 & 0.058 & 0.038 & 0.050 & 0.050 & 0.026 & 0.056 & 0.048 & 0.048 \\ 
  300 & 0.968 & 0.998 & 0.172 & 0.162 & 0.052 & 0.112 & 0.026 & 0.104 & 0.116 & 0.098 \\ 
  500 & 0.970 & 1.000 & 0.270 & 0.240 & 0.050 & 0.136 & 0.042 & 0.126 & 0.168 & 0.122 \\ 
\midrule
\multicolumn{11}{l}{Model A.4. $H_{0}: \bX \sim N_{6}(\zero_{6}, \bI_{6})$ against $H_{1}: \bX \sim t(\zero_{6}, \bI_{6}, 3)$}\\
\midrule  
  25 & 0.730 & 0.790 & 0.668 & 0.544 & 0.586 & 0.766 & 0.480 & 0.270 & 0.408 & 0.292 \\ 
  50 & 0.956 & 0.982 & 0.994 & 0.988 & 0.956 & 0.976 & 0.902 & 0.918 & 0.976 & 0.872 \\ 
  100 & 0.998 & 0.998 & 1.000 & 1.000 & 1.000 & 1.000 & 0.998 & 0.996 & 1.000 & 1.000 \\ 
  300 & 1.000 & 1.000 & 1.000 & 1.000 & 1.000 & 1.000 & 1.000 & 1.000 & 1.000 & 1.000 \\ 
  500 & 1.000 & 1.000 & 1.000 & 1.000 & 1.000 & 1.000 & 1.000 & 1.000 & 1.000 & 1.000 \\
\midrule
\multicolumn{11}{l}{Model A.5. $H_{0}: \bX \sim N_{6}(\zero_{6}, \bI_{6})$ against $H_{1}: \bX \sim C(\zero_{6}, \bI_{6})$}\\
\midrule
  25 & 0.994 & 0.998 & 0.998 & 0.988 & 0.984 & 0.998 & 0.968 & 0.876 & 0.974 & 0.944 \\ 
  50 & 1.000 & 1.000 & 1.000 & 1.000 & 1.000 & 1.000 & 1.000 & 1.000 & 1.000 & 1.000 \\ 
  100 & 1.000 & 1.000 & 1.000 & 1.000 & 1.000 & 1.000 & 1.000 & 1.000 & 1.000 & 1.000 \\ 
  300 & 1.000 & 1.000 & 1.000 & 1.000 & 1.000 & 1.000 & 1.000 & 1.000 & 1.000 & 1.000 \\ 
  500 & 1.000 & 1.000 & 1.000 & 1.000 & 1.000 & 1.000 & 1.000 & 1.000 & 1.000 & 1.000 \\  
\midrule
\multicolumn{11}{l}{Model A.6. $H_{0}: \bX \sim N_{6}(\zero_{6}, \bI_{6})$ against $H_{1}: \bX \sim L(\zero_{6}, \bI_{6})$}\\
\midrule
  25 & 0.656 & 0.912 & 0.828 & 0.590 & 0.568 & 0.896 & 0.438 & 0.080 & 0.400 & 0.418 \\ 
  50 & 0.910 & 0.996 & 1.000 & 0.990 & 0.942 & 1.000 & 0.854 & 0.682 & 0.966 & 0.930 \\ 
  100 & 1.000 & 1.000 & 1.000 & 1.000 & 1.000 & 1.000 & 1.000 & 0.982 & 1.000 & 1.000 \\ 
  300 & 1.000 & 1.000 & 1.000 & 1.000 & 1.000 & 1.000 & 1.000 & 1.000 & 1.000 & 1.000 \\ 
  500 & 1.000 & 1.000 & 1.000 & 1.000 & 1.000 & 1.000 & 1.000 & 1.000 & 1.000 & 1.000 \\
\bottomrule
\end{tabular}
\end{table}

\begin{table}
\footnotesize
\centering
\caption{Goodness-of-fit test: Observed relative frequency for rejecting the null, $d=10$.}
\label{tab:one-sample10}
\begin{tabular}{rrrrrrrrrrr}
 \toprule
$n$ & KS.depth & CvM.depth & AD & CM & DH & HZ & R & McCulloch & NRR & DN \\ 
\midrule
\multicolumn{11}{l}{Model A.1. $H_{0}: \bX \sim N_{10}(\zero_{10}, \bI_{10})$ against $H_{1}: \bX \sim N_{10}(\zero_{10}, \bI_{10})$}\\
\midrule
  25 & 0.040 & 0.040 & 0.056 & 0.062 & 0.052 & 0.042 & 0.072 & 0.048 & 0.152 & 0.190 \\ 
  50 & 0.040 & 0.018 & 0.050 & 0.048 & 0.048 & 0.054 & 0.056 & 0.052 & 0.076 & 0.100 \\ 
  100 & 0.024 & 0.004 & 0.044 & 0.034 & 0.040 & 0.040 & 0.042 & 0.046 & 0.058 & 0.056 \\ 
  300 & 0.026 & 0.016 & 0.052 & 0.048 & 0.068 & 0.050 & 0.066 & 0.068 & 0.052 & 0.036 \\ 
  500 & 0.022 & 0.014 & 0.058 & 0.070 & 0.042 & 0.064 & 0.044 & 0.050 & 0.056 & 0.066 \\ 
\midrule
\multicolumn{11}{l}{Model A.2. $H_{0}: \bX \sim N_{10}(\zero_{10}, \bI_{10})$ against $H_{1}: \bX \sim 0.8N_{10}(\zero_{10}, \bI_{10}) + 0.2N_{10}(5\textbf{1}_{10}, \bI_{10})$}\\
\midrule
 25 & 1.000 & 0.954 & 0.031 & 0.023 & 0.046 & 0.084 & 1.000 & 0.046 & 0.160 & 0.168 \\ 
  50 & 0.992 & 1.000 & 0.044 & 0.044 & 0.044 & 0.330 & 1.000 & 0.064 & 0.076 & 0.074 \\ 
  100 & 1.000 & 1.000 & 0.040 & 0.048 & 0.050 & 0.800 & 1.000 & 0.064 & 0.052 & 0.058 \\ 
  300 & 1.000 & 1.000 & 0.060 & 0.058 & 0.070 & 1.000 & 1.000 & 0.070 & 0.072 & 0.074 \\ 
  500 & 1.000 & 1.000 & 0.076 & 0.068 & 0.090 & 1.000 & 1.000 & 0.078 & 0.054 & 0.052 \\ 
\midrule
\multicolumn{11}{l}{Model A.3. $H_{0}: \bX \sim N_{10}(\zero_{10}, \bI_{10})$ against $H_{1}: \bX \sim 0.8N_{10}(\zero_{10}, \bI_{10}) + 0.2N_{10}(\zero_{10}, \bSigma)$}\\
\midrule
  25 & 1.000 & 0.992 & 0.023 & 0.023 & 0.031 & 0.053 & 0.069 & 0.069 & 0.076 & 0.107 \\ 
  50 & 0.982 & 0.980 & 0.012 & 0.014 & 0.050 & 0.074 & 0.054 & 0.042 & 0.046 & 0.046 \\ 
  100 & 0.974 & 0.956 & 0.040 & 0.024 & 0.052 & 0.078 & 0.044 & 0.040 & 0.042 & 0.038 \\ 
  300 & 0.870 & 0.678 & 0.312 & 0.224 & 0.092 & 0.182 & 0.076 & 0.090 & 0.160 & 0.142 \\ 
  500 & 0.698 & 0.392 & 0.578 & 0.446 & 0.082 & 0.308 & 0.070 & 0.114 & 0.274 & 0.250 \\ 
\midrule
\multicolumn{11}{l}{Model A.4. $H_{0}: \bX \sim N_{10}(\zero_{10}, \bI_{10})$ against $H_{1}: \bX \sim t(\zero_{10}, \bI_{10}, 3)$}\\
\midrule
   25 & 0.406 & 0.152 & 0.292 & 0.102 & 0.378 & 0.684 & 0.282 & 0.032 & 0.180 & 0.210 \\ 
  50 & 0.930 & 0.442 & 1.000 & 0.990 & 0.964 & 1.000 & 0.912 & 0.894 & 0.972 & 0.934 \\ 
  100 & 0.992 & 0.734 & 1.000 & 1.000 & 1.000 & 1.000 & 1.000 & 0.998 & 1.000 & 1.000 \\ 
  300 & 1.000 & 1.000 & 1.000 & 1.000 & 1.000 & 1.000 & 1.000 & 1.000 & 1.000 & 1.000 \\ 
  500 & 1.000 & 1.000 & 1.000 & 1.000 & 1.000 & 1.000 & 1.000 & 1.000 & 1.000 & 1.000 \\ 
\midrule
\multicolumn{11}{l}{Model A.5. $H_{0}: \bX \sim N_{10}(\zero_{10}, \bI_{10})$ against $H_{1}: \bX \sim C(\zero_{10}, \bI_{10})$}\\
\midrule
  25 & 0.946 & 0.390 & 0.988 & 0.920 & 0.934 & 1.000 & 0.864 & 0.102 & 0.850 & 0.862 \\ 
  50 & 1.000 & 0.962 & 1.000 & 1.000 & 1.000 & 1.000 & 1.000 & 1.000 & 1.000 & 1.000 \\ 
  100 & 1.000 & 1.000 & 1.000 & 1.000 & 1.000 & 1.000 & 1.000 & 1.000 & 1.000 & 1.000 \\ 
  300 & 1.000 & 1.000 & 1.000 & 1.000 & 1.000 & 1.000 & 1.000 & 1.000 & 1.000 & 1.000 \\ 
  500 & 1.000 & 1.000 & 1.000 & 1.000 & 1.000 & 1.000 & 1.000 & 1.000 & 1.000 & 1.000 \\ 
 \midrule
\multicolumn{11}{l}{Model A.6. $H_{0}: \bX \sim N_{10}(\zero_{10}, \bI_{10})$ against $H_{1}: \bX \sim L(\zero_{10}, \bI_{10})$}\\
\midrule
 25 & 0.442 & 0.084 & 0.738 & 0.284 & 0.456 & 0.936 & 0.358 & 0.224 & 0.418 & 0.362 \\ 
  50 & 0.878 & 0.114 & 1.000 & 1.000 & 0.958 & 1.000 & 0.926 & 0.410 & 0.972 & 0.972 \\ 
  100 & 0.996 & 0.186 & 1.000 & 1.000 & 1.000 & 1.000 & 1.000 & 0.956 & 1.000 & 1.000 \\ 
  300 & 1.000 & 0.994 & 1.000 & 1.000 & 1.000 & 1.000 & 1.000 & 1.000 & 1.000 & 1.000 \\ 
  500 & 1.000 & 1.000 & 1.000 & 1.000 & 1.000 & 1.000 & 1.000 & 1.000 & 1.000 & 1.000 \\ 
\bottomrule
\end{tabular}
\end{table}

\begin{figure}[h]
    \centering
    \begin{subfigure}{\textwidth}
         \centering
         \includegraphics[width=0.45\textwidth]{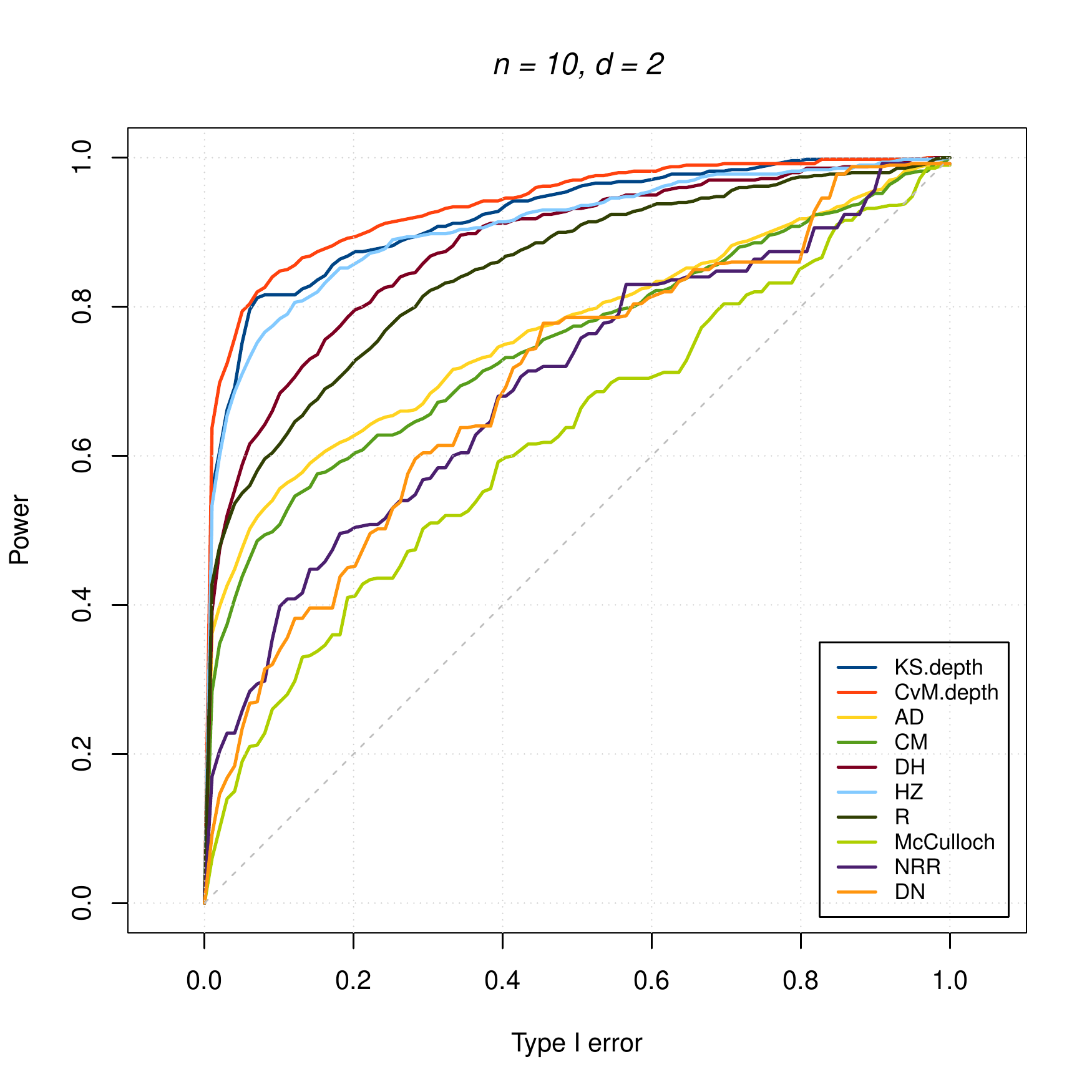}
         \includegraphics[width=0.45\textwidth]{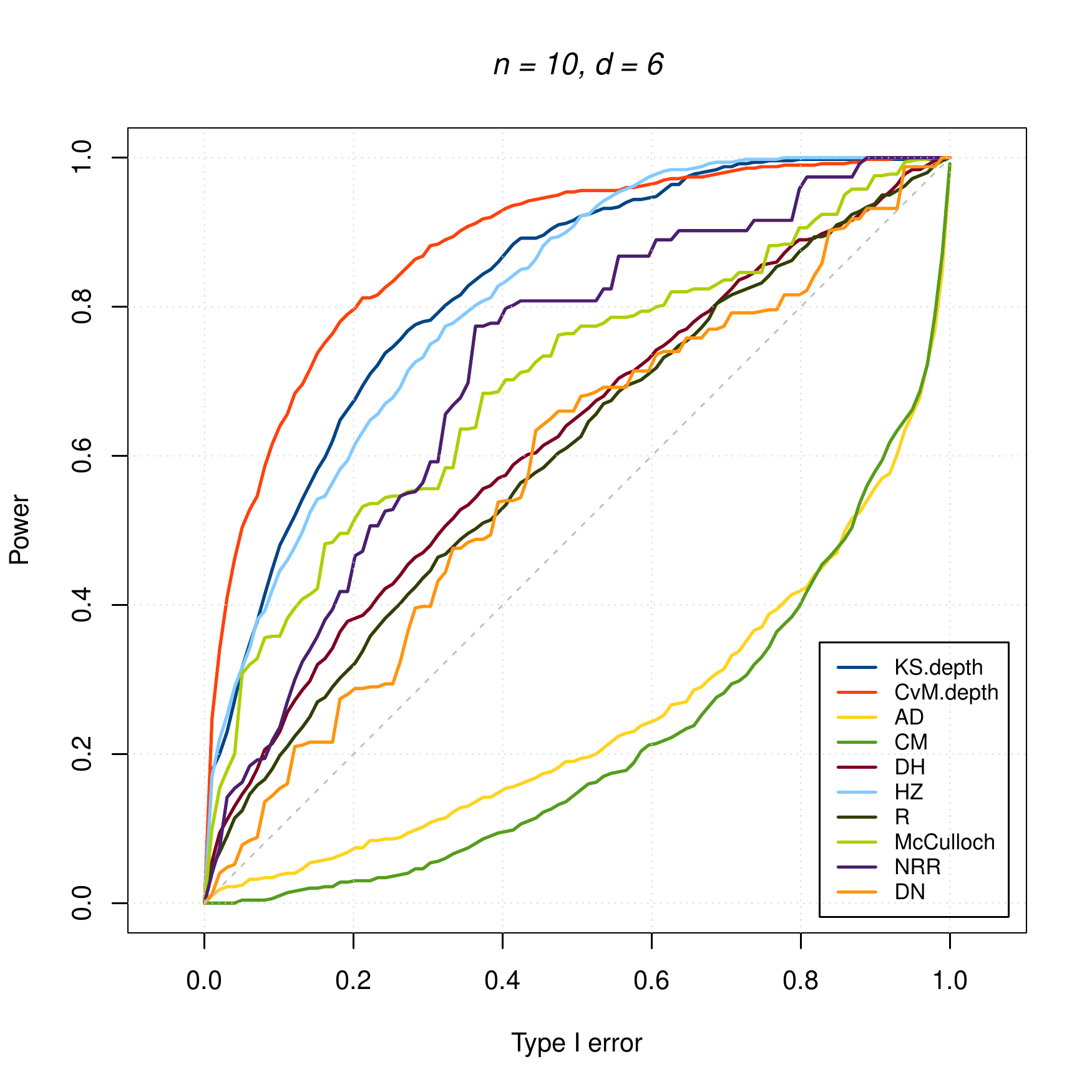}
         \caption{$H_{1}: \bX \sim C(\zero_{d}, \bI_{d})$}
         \label{fig:auc-cauchy}
     \end{subfigure}
     \hfill
     \begin{subfigure}{\textwidth}
         \centering
         \includegraphics[width=0.45\textwidth]{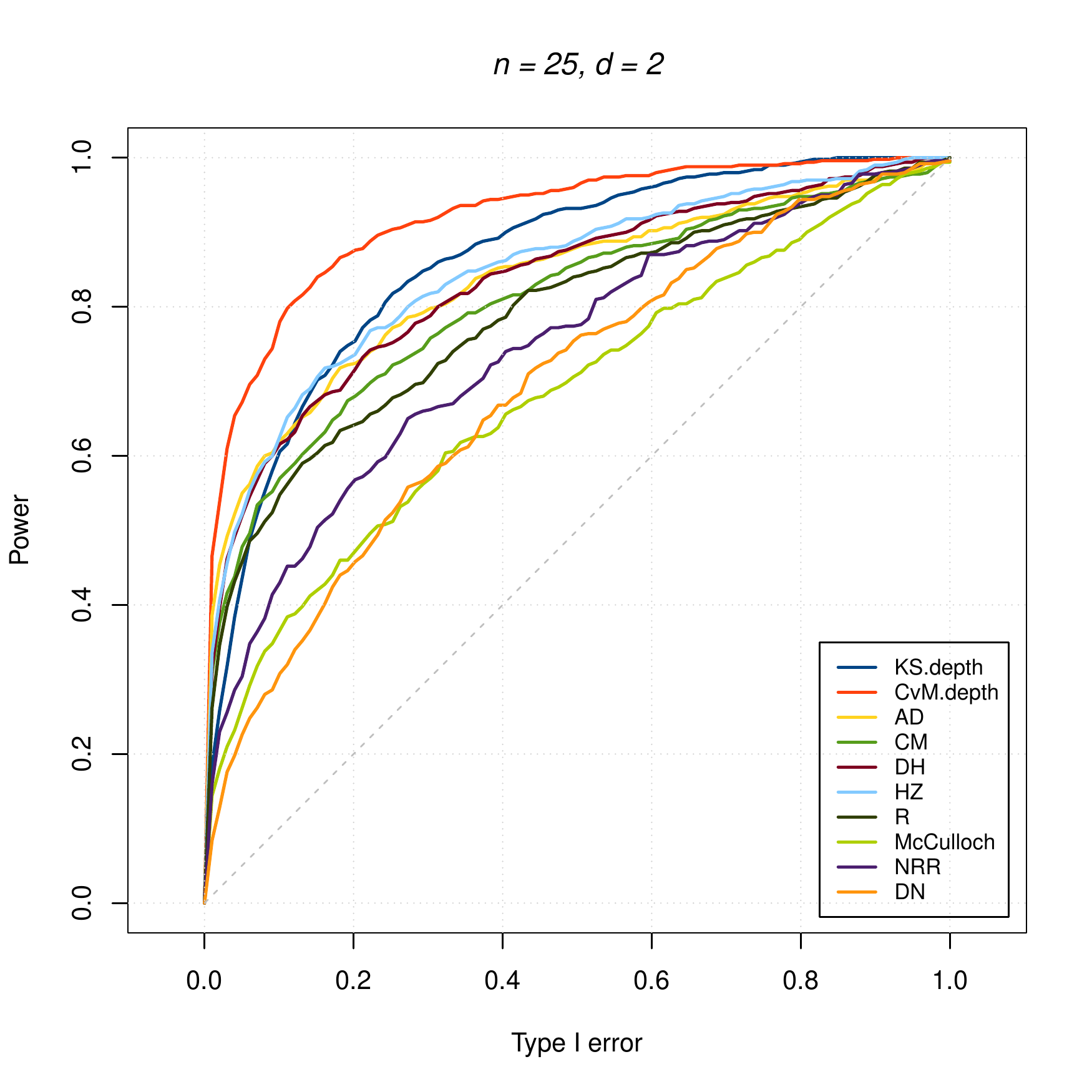}
         \includegraphics[width=0.45\textwidth]{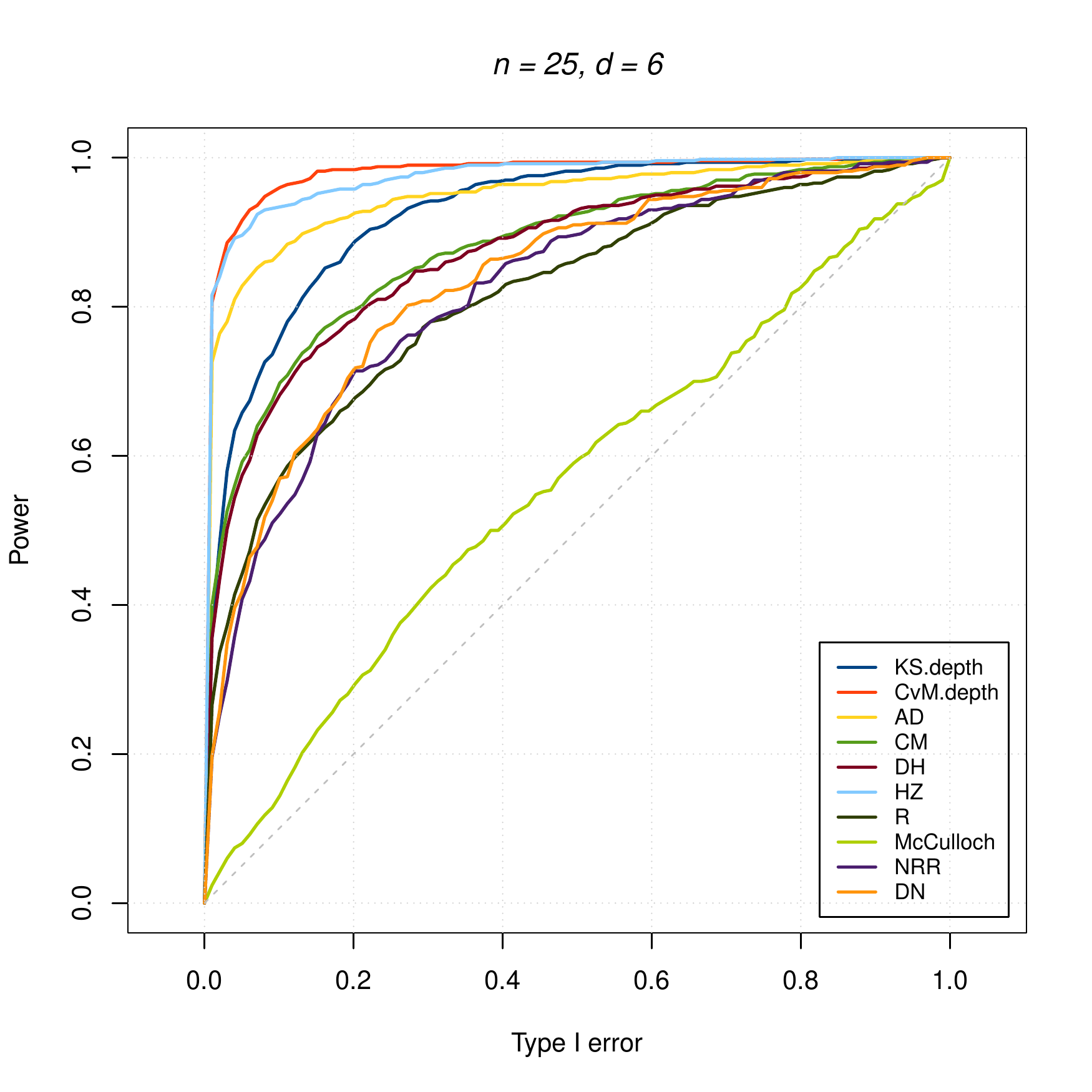}
         \caption{$H_{1}: \bX \sim L(\zero_{d}, \bI_{d})$}
         \label{fig:auc-laplace}
     \end{subfigure}
     \caption{ROC plot for the goodness-of-fit test when the alternative distributions are standard Cauchy (a) and standard Laplace distribution (b), respectively, where the null distribution is standard normal distribution.}
    \label{fig:auc1}
\end{figure}

\section{Real data analysis}
\label{sec:real-data}
To understand the practicability of the proposed tests, two well-known data sets are analyzed here. 

\paragraph{Fisher's \textit{Iris data}} 
\mbox{}\par
This data set, available in the base R, consists of 
three multivariate samples corresponding to three different species of Iris, namely \textit{Iris setosa, Iris virginica} and \textit{Iris versicolor} with sample size 50 for each species. 
In each species, the length and width of the sepals are measured in centimeters. 
We would like to determine how close the distribution of each sample associated with sepal length and sepal width is to a bivariate normal sample. This can be formulated as a goodness-of-fit testing problem, with $F_{0}$ being a bivariate normal distribution with unknown mean $\bmu$ and unknown dispersion matrix $\bSigma$. 
For each species, we estimate these unknown parameters using the corresponding mean vector and covariance matrix. The data-depth discrepancy is graphically represented for three different species in Figure \ref{fig:one-sample-iris}, where we observe that most points are tightly clustered around the straight line, indicating that the standardized data are from a standard bivariate normal distribution. 
Moreover, our goodness-of-fit test for testing $H_{0}: F = F_{0}$ against $H_{1}: F\neq F_{0}$ led to very high empirical p-values for all types of spices (see Table \ref{tab:pval-iris}). 
This indicates that $H_{0}$ is to be accepted, and therefore, bivariate normal distributions show signs of good fit for the data for the sepal length and sepal width of three Iris species. 

\begin{figure}[t!]
\centering
    \begin{subfigure}{0.49\textwidth}
         \centering
         \includegraphics[width=\textwidth]{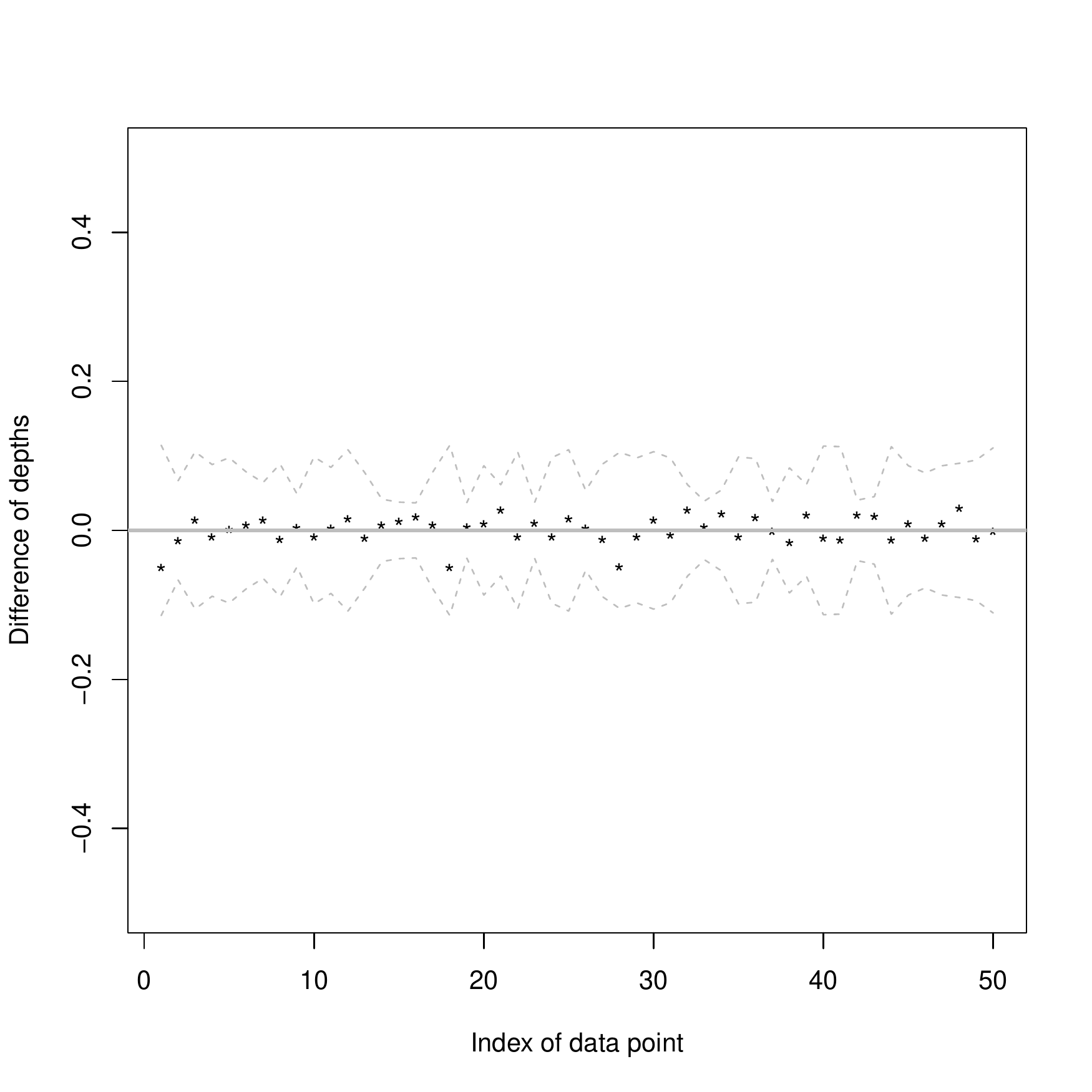}
         \caption{\textit{Iris setosa}}
    \end{subfigure}
    \begin{subfigure}{0.49\textwidth}     
         \includegraphics[width=\textwidth]{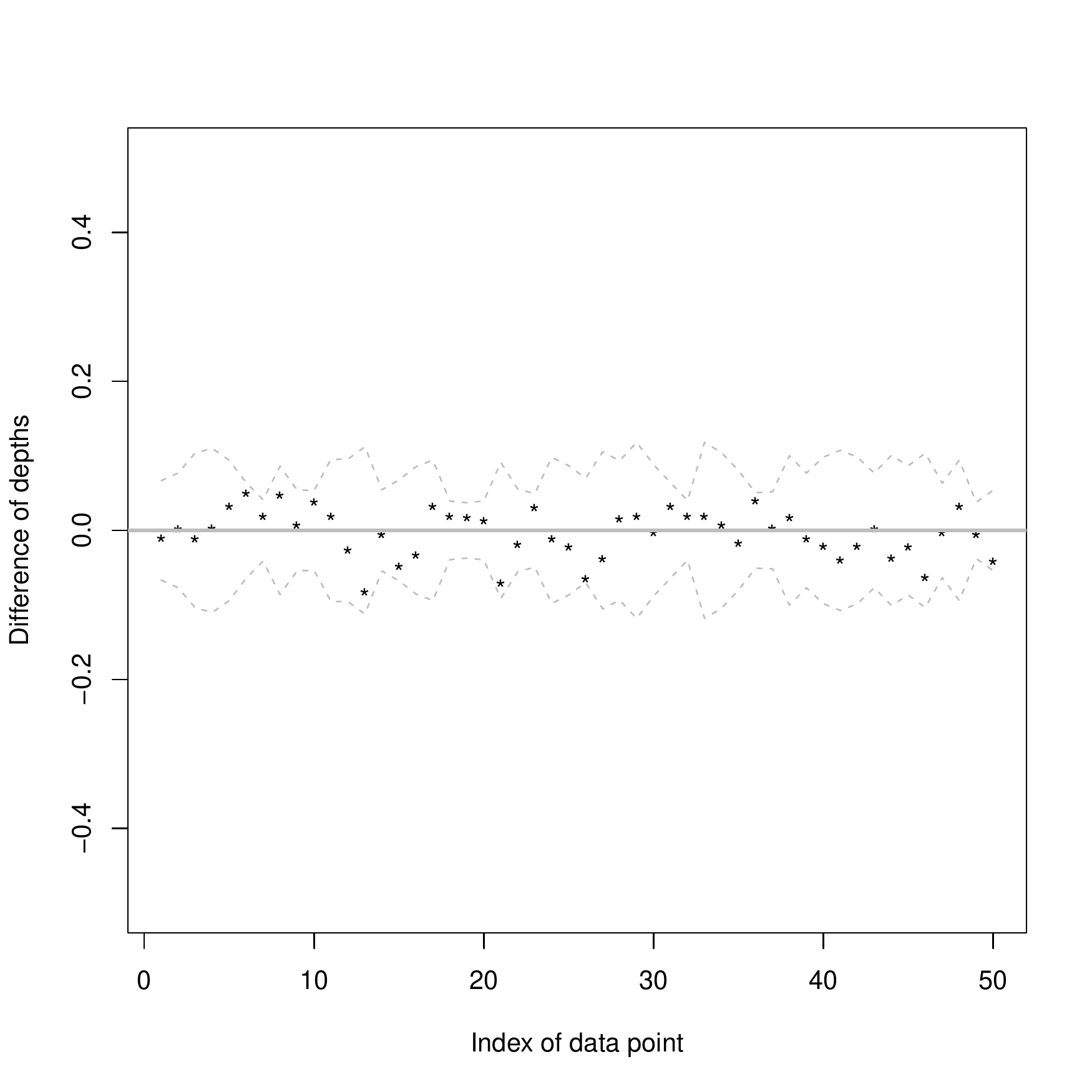}
         \caption{\textit{Iris virginica}}
    \end{subfigure}
    \begin{subfigure}{0.49\textwidth}     
         \includegraphics[width=\textwidth]{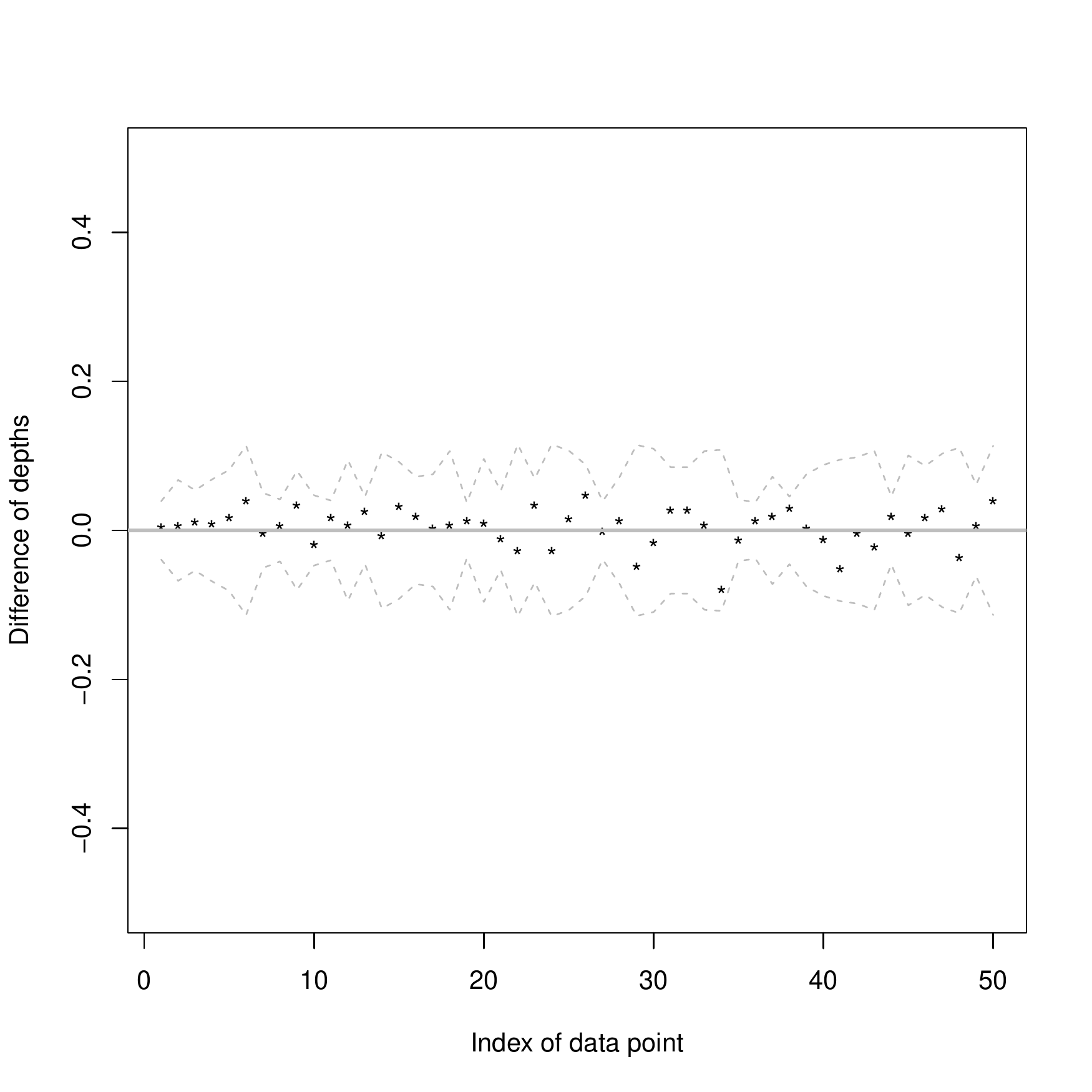}
         \caption{\textit{Iris versicolor}}
     \end{subfigure}
	\caption{\textit{Goodness-of-fit test for Iris data:} The data-depth discrepancy plot for \textit{Iris setosa, Iris virginica} and \textit{Iris versicolor} respectively for testing normality. The dotted gray curves indicate the two-sigma limits of data-depth discrepancy}
	\label{fig:one-sample-iris}
\end{figure}

\begin{table}[H]
    \footnotesize
    \centering
    \caption{Empirical p-values of the proposed goodness-of-fit tests obtained over 1000 replications.}
    \label{tab:pval-iris}
    \begin{tabular}{rrr}
    \toprule
    & \multicolumn{2}{c}{p-value of depth based tests}\\
    \cmidrule(lr){2-3}
    Species & KS.depth & CvM.depth\\
    \midrule
    \textit{Iris setosa} & 0.22 & 0.338\\
    \textit{Iris virginica} & 0.35 & 0.358\\
    \textit{Iris versicolor} & 0.192 & 0.118\\
    \bottomrule
    \end{tabular}
\end{table}

\paragraph{\textit{gilgais} data} 
\mbox{}\par
This data set is available in the \texttt{MASS} package in R.
It is collected on a line transect survey in the gilgai territory in New South Wales, Australia \citep{webster1977spectral}.  On a 4-meter spaced linear grid, 365 sampling locations are selected. Electrical conductivity (in mS / cm), pH, and chloride contained (in ppm) are collected at three different depths below the surface, 0-10 cm, 30-40 cm, and 80-90 cm. 
We would like to investigate how close the joint distribution of three variables at each level is to a trivariate normal distribution. As the previous example, this can be formulated as a goodness-of-fit problem, where $F_{0}$ is specified as a trivariate normal distribution with unknown mean $\mu$ and unknown dispersion matrix $\Sigma$. We estimate these unknown parameters using their respective maximum likelihood estimates. We then standardize the data in each sample using the corresponding mean vector and covariance matrix. The proposed data-depth dispersion plot is shown in Figure \ref{Fig:gilgais-nromality} where we observe a majority of the data cloud is far from the straight line. This phenomenon indicates that the data are not from a trivariate normal distribution. Moreover, the test of $H_{0}: F = F_{0}$ against $H_{1}: F \neq F_{0}$ led to zero empirical p-values for all height levels below the surface. Moreover, we are interested to study whether distribution of any height is same as the other. Therefore, this can be viewed as two-sample test where $F$ is the distribution of chemical parameters art height at $0-10$ cm and G is that of either at height $30–40$ cm or $80–90$ cm, respectively. The Figure \ref{Fig:gilgais-nromality-two} shows that the data cloud is far from the straight line, this phenomenon indicates that the distributions are different. Moreover, the empirical p-value for each of situations are close to zero based on the proposed test which determines that the distributions are significantly different. 
\begin{figure}[t!]
	\begin{subfigure}{0.49\textwidth}
         \centering
         \includegraphics[width=\textwidth]{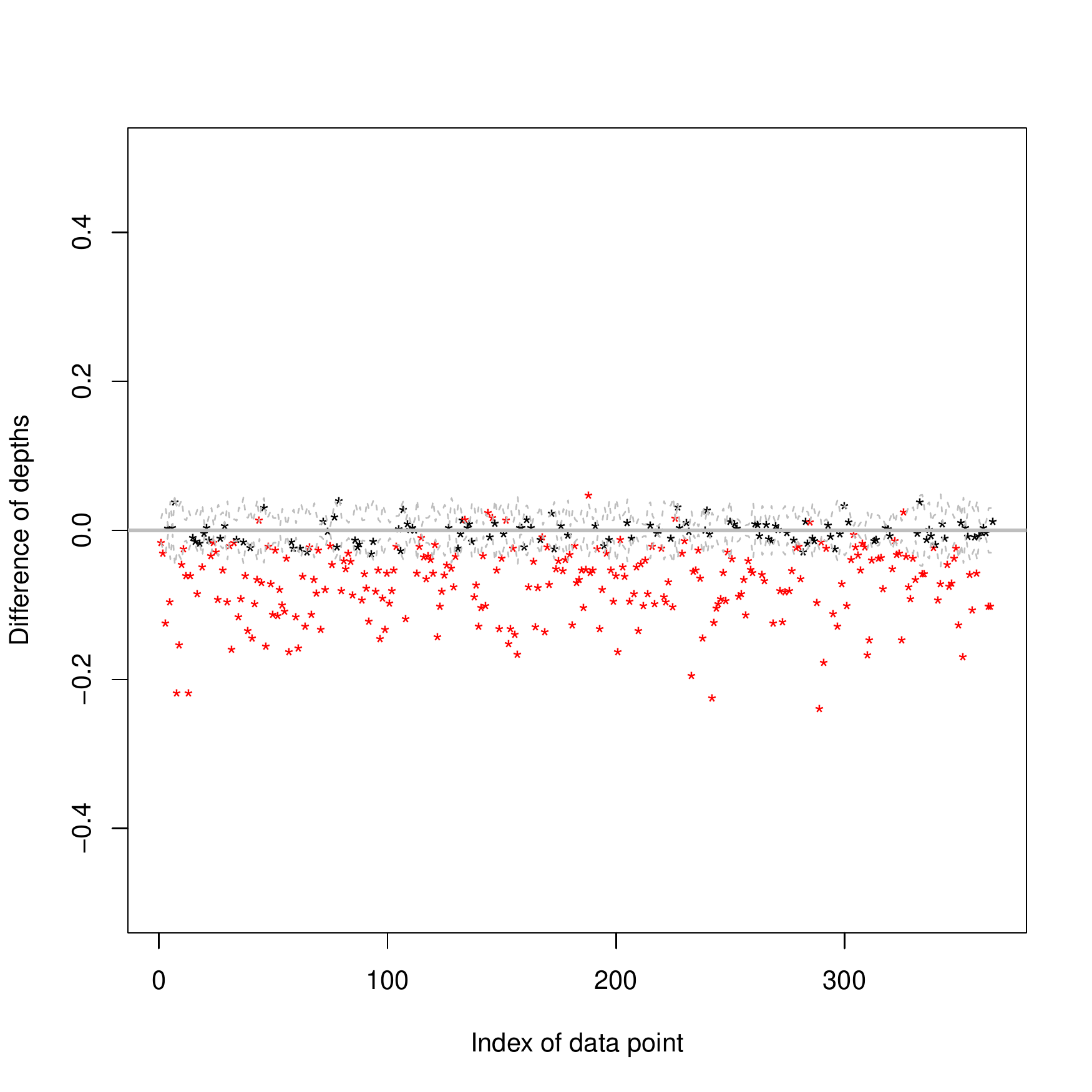}
         \caption{height 0–10 cm}
         \label{fig:depth1}
    \end{subfigure}
	\begin{subfigure}{0.49\textwidth}
         \centering
         \includegraphics[width=\textwidth]{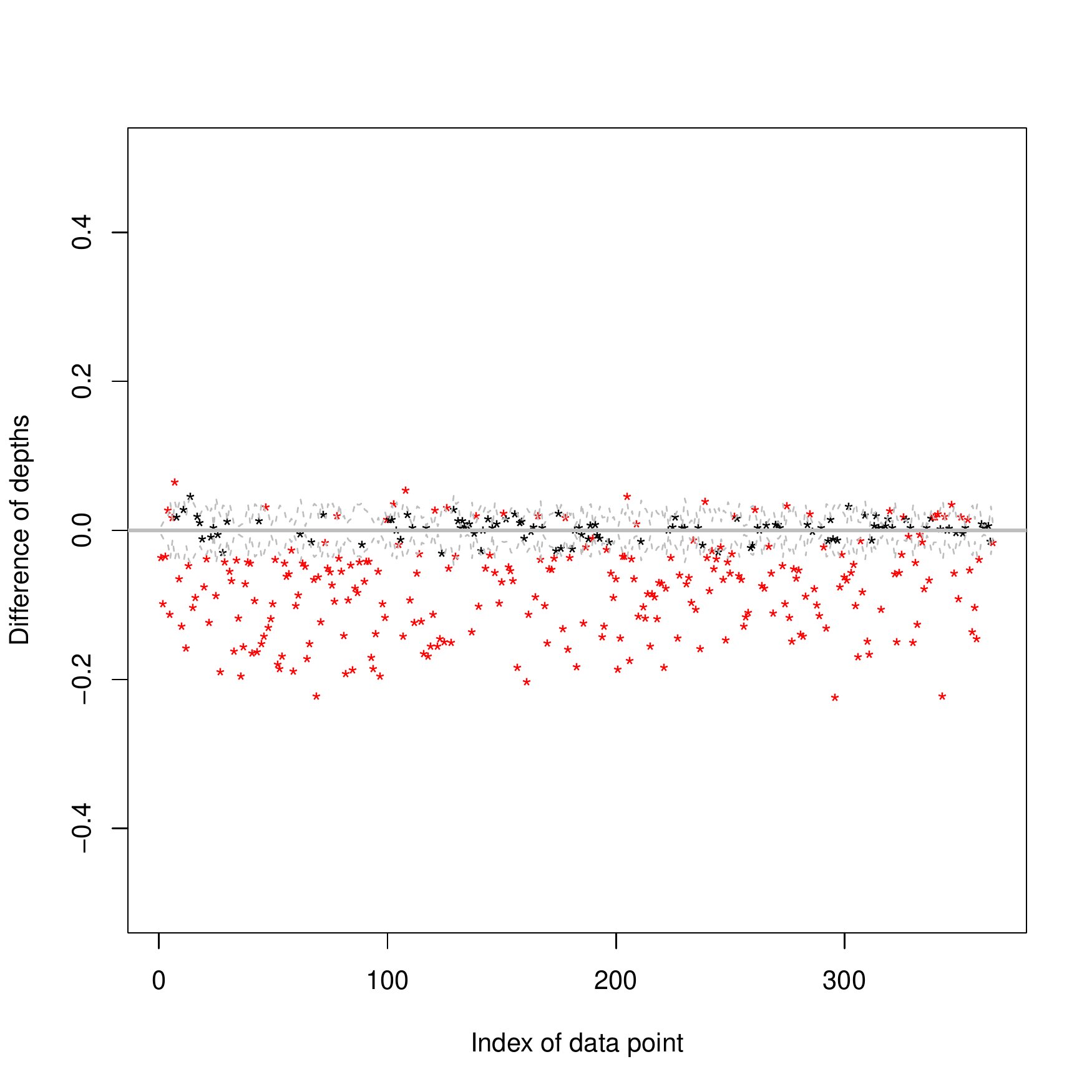}
         \caption{height 30–40 cm}
         \label{fig:depth2}
    \end{subfigure}
	\begin{subfigure}{0.49\textwidth}
         \centering
         \includegraphics[width=\textwidth]{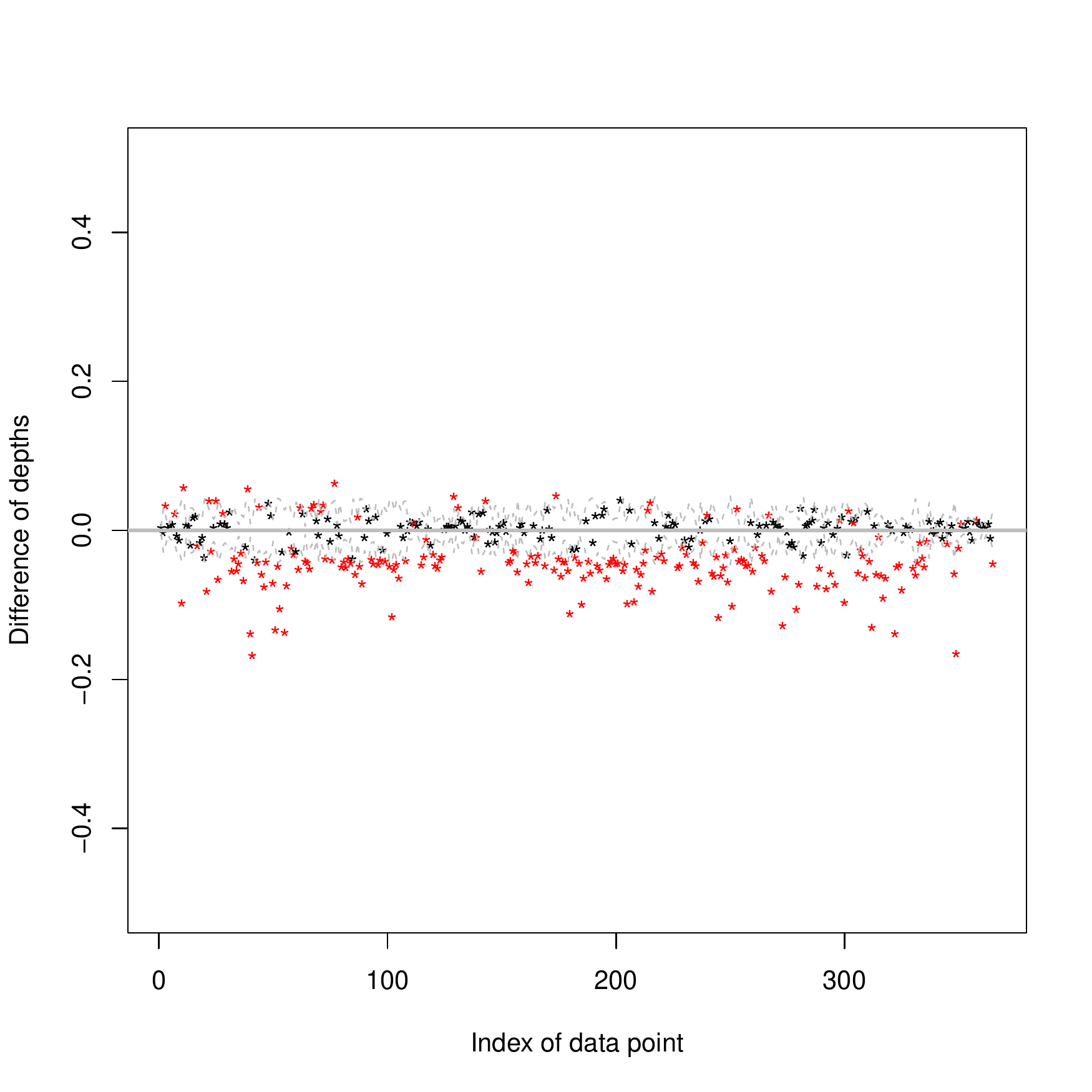}
         \caption{height 80–90 cm}
         \label{fig:depth3}
    \end{subfigure}
	\centering
	\caption{Goodness-of-fit test for gilgais data: The data-depth discrepancy plot for tuple at depths below the surface of height 0-10 cm, 30-40 cm and 80-90 cm respectively for testing normality. The dotted gray curves indicate the two-sigma limits of data-depth discrepancy}
	\label{Fig:gilgais-nromality}
\end{figure}
\begin{figure}[t!]
    \begin{subfigure}{0.49\textwidth}
         \centering
         \includegraphics[width=\textwidth]{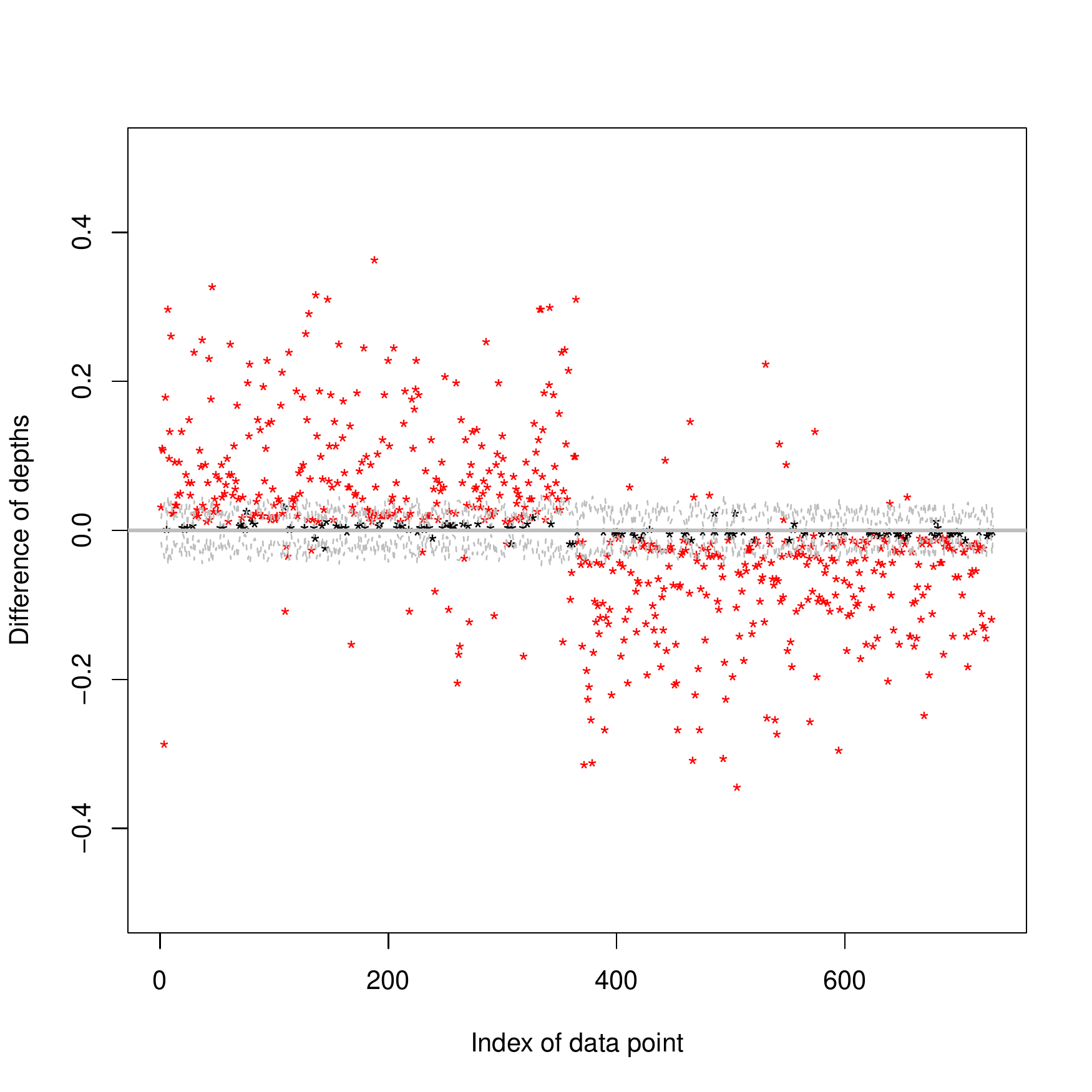}
         \caption{heights 0–10 cm against 30-40 cm}
    \end{subfigure}
    \begin{subfigure}{0.49\textwidth}
         \centering
         \includegraphics[width=\textwidth]{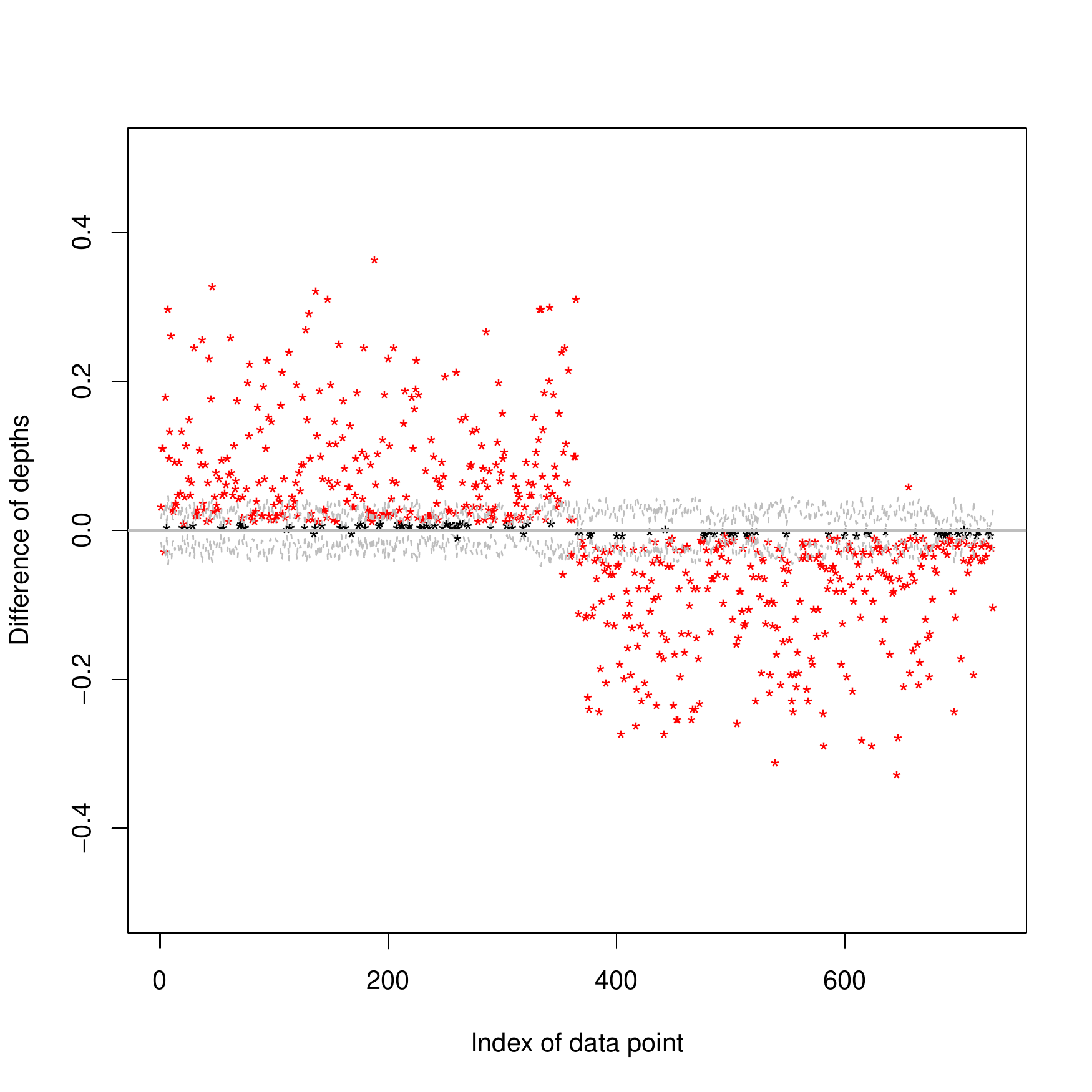}
         \caption{height 0–10 cm against 30-40 cm}
    \label{Fig:gilgais-nromality-two} 
    \end{subfigure}
    \caption{Two-sample test for gilgias data: The data-depth discrepancy plot for heights 0–10 cm against 30-40 cm and height 0–10 cm against 30-40 cm are shown. The dotted gray curves indicate the two-sigma limits of data-depth discrepancy}
\end{figure}

\newpage
\section{Technical details}
\label{sec:proofs}

\begin{proof}[Proof of Theorem \ref{thm:graph1}]$ $
\par
Note that, for every $\epsilon > 0$,
\begin{align*}
    \lim_{n\rightarrow\infty}\bbP\left\{\widehat{\sC}(\sX, F_{0}) \subset \sC_{\epsilon}(F, F_{0})\right\} 
    \geq \bbP\left\{\sup_{\bx}|D_{\sX}(\bx) - D_{F_{0}}(\bx) | < \epsilon \right\} \rightarrow 1
    \numberthis
\end{align*}
due to the Proposition 4.4 of \cite{masse2004asymptotics}.
\end{proof}

\begin{proof}[Proof of Theorem \ref{thm:graph2}]$ $\par
    By the similar argument of Theorem \ref{thm:graph2}, using the fact that $\lim\limits_{\min(n,m) \rightarrow \infty} n/(n+m) = \lambda \in (0, 1)$,
    observe that, 
    \begin{align*}
        \lim_{\min(n, m) \rightarrow \infty}\bbP\{ \widehat{\sC}(\sX, \sY) \subset \sC_{\epsilon}(F, G)
        \} \geq \bbP\left\{\sup_{\bx}|D_{\sX}(\bx) - D_{\sY}(\bx)| < \epsilon \right\} \rightarrow 1
    \end{align*}
\end{proof}
\par

\begin{proof}[Proof of Theorem \ref{thm:const1-KM}]$ $
\par
Let us define the Kolmogorv-Smirnov distance between empirical and theoretical distribution functions $\widehat{F}_{n}$ and $F_{0}$ based on Tukey's half-space depth $D$ as
$d_{K}(\sX, F_{0}) = \sup\limits_{\bx\in \bbR^{d}}|D_{\sX}(\bx) - D_{F_{0}}(\bx)| = \sup\limits_{\bx, \bu} |\widehat{F}_{n}(H[\bx, \bu]) - F_{0}(H[\bx, \bu])|$ 
where $\sX = \{\bX_{1}, \cdots, \bX_{n}\}$ and $H$ is the half-space with $H[\bx, \bu] = \{\by\in \bbR^{d}: \bu^{\tp}\bx \geq \bu^{\tp}\by\}$. 
It is important to note that $d_{K}(\cdot, \cdot)$ has the Glivenko-Cantelli property \citep{pollard2012convergence,donoho1992breakdown}, i.e., if the data samples are i.i.d. from $F_{0}$, $d_{K}(\sX, F_{0}) \rightarrow 0$ as $n \rightarrow \infty$.
Therefore, 
under an alternative $F$, we have 
\begin{equation}\label{eq:Glivenko-Cantalli}
    \sup_{\bx}|D_{\sX}(\bx) - D_{F_{0}}(\bx)| \rightarrow d_{K}(D_{F}, D_{F_{0}}) > 0
\end{equation}
almost surely. Now fix any $F$ with $d_{K}(D_{F}, D_{F_{0}}) > 0$, then there exists some $\bx$ with $D_{\sX}(\bx) \neq D_{F}(\bx)$. Without loss of generality, assume that $D_{\sX}(\bx) > D_{F_{0}}(\bx)$. Then
\begin{align*}
    \bbP_{F}\left\{T_{\sX, F_{0}}^{\KS} > s_{1-\alpha}^{(1)}\right\} &\geq \bbP_{F}\left\{\sqrt{n}|D_{\sX}(\bx) - D_{F_{0}}(\bx)| > s_{1-\alpha}^{(1)}\right\}\\
    &= \bbP_{F}\left\{\sqrt{n}[D_{\sX}(\bx) -D_{F}(\bx) + D_{F}(\bx) - D_{F_{0}}(\bx)] > s_{1-\alpha}^{(1)} \right\}\\
    & \geq 
    \bbP_{F}\left\{\sqrt{n}[D_{\sX}(\bx) - D_{F}(\bx)] > s_{1-\alpha}^{(1)} - \sqrt{n}[D_{F}(\bx) - D_{F_{0}}(\bx)]\right\}\\
    &\rightarrow 1 \text{ as } n \rightarrow \infty.
    \numberthis
\end{align*}
The last implication follows from the fact that $\left\{\sqrt{n}(D_{\sX}(\bx) - D_{F}(\bx))\right\}$ is uniformly bounded in probability in view of Proposition \ref{proposition:depth} and an application of Prokhorov's theorem, 
and for finite $s_{1-\alpha}^{(1)}$, we have $s_{1-\alpha}^{(1)} - \sqrt{n}|D_{F}(\bx) - D_{F_{0}}(\bx)| \rightarrow -\infty$ as $n \rightarrow \infty$. Hence the limiting power is one when $D_{F}(\bx) > D_{F_{0}}(\bx)$. Similar argument can show that the limiting power is one when $D_{F}(\bx) < D_{F_{0}}(\bx)$ for some $\bx$. 
Hence, the proof is complete.
\end{proof}

\begin{proof}[Proof of Theorem \ref{thm:const1-CVM}]
$ $\par
Due to the Glivenko-Cantelli property of $d_{K}(\cdot, \cdot)$, under an alternative $F$, \ref{eq:Glivenko-Cantalli} holds almost surely. Thus fix any $F$with $d_{K}(D_{F}, D_{F_{0}}) > 0$, then there exists some $\bx$ with $D_{\sX}(\bx) \neq D_{F}(\bx)$. Therefore, observe that 
\begin{align*}
\bbP_{F}\left\{T_{\sX, F_{0}}^{\CvM} > s_{1-\alpha}^{(2)}\right\}
& = \bbP_{F}\left\{n\int(D_{\sX}(\bx) - D_{F_{0}}(\bx))^{2}dF_{0}(\bx) > s_{1-\alpha}^{(2)}\right\}\\
&\geq \bbP_{F}\left\{
n\int(D_{\sX}(\bx) - D_{F}(\bx))^{2}dF_{0}(\bx) >
s_{1-\alpha}^{(2)} - a_{n} + 2b_{n}
\right\}\\
&\rightarrow 1 \text{ as } n \rightarrow \infty,
\numberthis
\end{align*}
where $a_{n} = n\int (D_{F}(\bx) - D_{F_{0}}(\bx))^{2}dF_{0}(\bx) $ and $b_{n} = n\int(D_{\sX}(\bx) - D_{F}(\bx))(D_{F}(\bx) - D_{F_{0}}(\bx))dF_{0}(\bx)$. The last implication follows from the following facts: (i) $\left\{\sqrt{n}(D_{\sX}(\bx) - D_{F}(\bx))\right\}$ is uniformly bounded in probability in view of Proposition \ref{proposition:depth}, (ii) $s_{1-\alpha}^{(2)}$ is positive finite, and (iii) $\sup\limits_{\bx}|D_{\sX}(\bx) - D_{F}(\bx)| \rightarrow 0$ as $n \rightarrow \infty$ under $F$ almost surely (see \citet{donoho1992breakdown} and Proposition 4.4 of \citet{masse2004asymptotics}). An application of Prokhorov’s theorem shows that the limiting power is one when $D_{F}(\bx) > D_{F_{0}}(\bx)$. 
Similar argument can show
that the limiting power is one when $D_{F}(\bx) < D_{F_{0}}(\bx)$ for some $\bx$. Hence, the proof is
complete.
\end{proof}

\begin{proof}[Proof of Theorem \ref{thm:unif-const1-KM}]
Let $F_{n}$ be any distribution function that satisfies $\sqrt{n}d_{K}(D_{F_{n}}, D_{F_{0}}) \geq \Delta_{n}$. By the triangle inequality, we have 
\begin{equation}
    \label{eq:triangle}
    d_{K}(D_{F_{n}}, D_{F_{0}}) \leq d_{K}(D_{F_{n}}, D_{\sX}) + d_{K}(D_{\sX}, D_{F_{0}}),
\end{equation}
which implies $T_{\sX, F_{0}}^{\KS} \geq \Delta_{n} - \sqrt{n}d_{K}(D_{F_{n}}, D_{\sX})$. 
Therefore, 
\begin{align*}
\bbP_{F_{n}}\left\{T_{\sX, F_{0}}^{\KS} > s_{1-\alpha}^{(1)}\right\} 
&\geq \bbP_{F_{n}}\left\{\Delta_{n} - \sqrt{n}d_{K}(D_{F_{n}}, D_{\sX}) > s_{1-\alpha}^{(1)}\right\}\\
&\geq \bbP_{F_{n}}\left\{\sqrt{n}d_{K}(D_{F_{n}}, D_{\sX}) \leq \Delta_{n}-s_{1-\alpha}^{(1)}\right\}\\
& \rightarrow 1 \text{ as } n \rightarrow \infty.
\numberthis
\end{align*}
The last implication follows from
the following facts: (i) $\left\{\sqrt{n}(D_{\sX}(\bx) - D_{F}(\bx))\right\}$ is uniformly bounded under $F_{n}$ in probability in view of Proposition \ref{proposition:depth}
(ii) $\Delta_{n} \rightarrow\infty$, and (iii) $s_{1-\alpha}^{(1)}$ is positive and finite. 
An application of Prokhorov's theorem, 
we get $\bbP_{F_{n}}\left\{T_{\sX, F_{0}}^{\KS} > s_{1-\alpha}^{(1)} \right\} \rightarrow 1$. 
\end{proof}

\begin{proof}[Proof of Theorem \ref{thm:unif-const1-CVM}]
$ $\par
Let $F_{n}$ be any distribution function satisfies $\sqrt{n}d_{K}(D_{F_{n}}, D_{F_{0}}) \geq \Delta_{n}$. 
Therefore, using the triangle inequality \eqref{eq:triangle}, we have 
\begin{align*}
\bbP_{F_{n}}\left\{T_{\sX, F_{0}}^{\CvM} > s_{1-\alpha}^{(2)}\right\} 
& \geq \bbP_{F_{n}}\left\{ \int \Delta_{n}^{2}dF_{0}(\bx) - n\int d_{K}^{2}(D_{F_{n}}, D_{\sX})dF_{0}(\bx) > s_{1-\alpha}^{(2)}\right\}\\
& \geq \bbP_{F_{n}}\left\{ 
n\int d_{K}^{2}(D_{F_{n}}, D_{\sX})dF_{0}(\bx) \leq \int \Delta_{n}^{2}dF_{0}(\bx) - s_{1-\alpha}^{(2)}
\right\}\\
&\rightarrow 1\text{ as } n \rightarrow \infty.
    \numberthis
\end{align*}
The last implication follows from the
the following fact: (i) $\left\{\sqrt{n}(D_{\sX}(\bx) - D_{F}(\bx))\right\}$ is uniformly bounded under $F_{n}$ in probability in view of Proposition \ref{proposition:depth}
(ii) $\Delta_{n} \rightarrow\infty$, and (iii) $s_{1-\alpha}^{(2)}$ is positive and finite. 
An application of Prokhorov's theorem, 
we get $\bbP_{F_{n}}\left\{T_{\sX, F_{0}}^{\CvM} > s_{1-\alpha}^{(2)} \right\} \rightarrow 1$. 
\end{proof}

\begin{proof}[Proof of Theorem \ref{thm:two-sample}]
$ $\par
Define $a_{n, m} = \sqrt{n+m}\times d_{K}(D_{F}, D_{G}) = \sqrt{n+m}|D_{F}(\bx) - D_{G}(\bx)|$.
\begin{align*}
    &\bbP_{H_{1}}\left\{T_{\sX, \sY}^{\KS} > t_{1-\alpha}^{(1)}\right\}\\ &\geq 
    \bbP_{H_{1}}\left\{\sqrt{n+m}|D_{\sX}(\bx) - D_{\sY}(\bx)| > t_{1-\alpha}^{(1)}\right\}\\
    & \geq \bbP_{H_{1}}\left\{ 
    \sqrt{n+m}|D_{\sX}(\bx) - D_{F}(\bx)| + \sqrt{n+m}|D_{\sY}(\bx) - D_{G}(\bx)| > t_{1-\alpha}^{(1)} - a_{n, m}
    \right\}\\
    &\geq \bbP_{H_{1}}\left\{ 
    \lambda^{-1/2}\sqrt{n}|D_{\sX}(\bx) - D_{F}(\bx)| + (1-\lambda)^{-1/2}\sqrt{m}|D_{\sY}(\bx) - D_{G}(\bx)| > t_{1-\alpha}^{(1)} - a_{n,m}
    \right\}\\
    &\rightarrow 1\text{ as } \min(n,m) \rightarrow \infty.
    \numberthis
\end{align*}
The last implication follows from the fact that $\{\sqrt{n}(D_{\sX}(\bx) - D_{F}(\bx)\}$ and  $\{\sqrt{m}(D_{\sY}(\bx) - D_{G}(\bx)\}$ are uniformly bounded in probability in view of Proposition \ref{proposition:depth}, for finite $t_{\alpha}^{(1)}$, we have $\sqrt{n+m}|D_{F}(\bx) - D_{G}(\bx)| \rightarrow \infty$ as $\min(n, m) \rightarrow \infty$. By an application of Prokhorov's theorem, we get $\bbP_{H_{1}}\left\{T_{\sX, \sY}^{\KS} > t_{1-\alpha}^{(1)}\right\} \rightarrow 1$ as $\min(n, m) \rightarrow \infty$. 
\par
Furthermore, 
\begingroup
\allowdisplaybreaks
\begin{align*}
&\bbP_{H_{1}}\left\{T_{\sX, \sY}^{\CvM} > t_{1-\alpha}^{(2)}\right\}\\
&\geq \bbP\left\{
\lambda^{-1}n\int(D_{\sX}(\bx) - D_{F}(\bx))^{2}dH_{n,m}(\bx) \right.\\
& \qquad +
\left.(1-\lambda)^{-1}m\int(D_{\sY}(\bx) - D_{G}(\bx))^{2}dH_{n,m}(\bx) > t_{1-\alpha}^{(2)} - a^{*}_{n,m}
\right\}\\
& \rightarrow 1 \text{ as } \min(n, m) \rightarrow \infty
    \numberthis
\end{align*}
\endgroup
where  $a^{*}_{n,m} = a_{n,m} + 2b_{n,m} + 2c_{n,m} + 2d_{n,m}$ with 
\begin{align*}
 a_{n,m} & = (n+m)\int (D_{F}(\bx) - D_{G}(\bx))^{2}dH_{n,,m}(\bx)\\
 b_{n,m} &= (n+m)\int(D_{\sX}(\bx) - D_{F}(\bx))(D_{\sY}(\bx) - D_{G}(\bx))dH_{n,m}(\bx)\\   
 c_{n,m} &= (n+m)\int(D_{\sX}(\bx) - D_{F}(\bx))(D_{F}(\bx) - D_{G}(\bx))dH_{n,m}(\bx)\\
 d_{n,m} &= (n+m)\int(D_{\sY}(\bx) - D_{G}(\bx))(D_{F}(\bx) - D_{G}(\bx))dH_{n,m}(\bx)
\end{align*}
The last implication follows from the following facts: (i) $\{ \sqrt{n}(D_{\sX}(\bx)-D_{F}(\bx))\}$ and $\sqrt{m}(D_{\sY}(\bx)-D_{G}(\bx))\}$ are uniformly bounded in probability in view of Proposition \ref{proposition:depth}, (ii) $t_{1-\alpha}^{(1)}$ is positive finite, (iii) $\lambda = \lim\limits_{\min(n, m)\rightarrow\infty}\frac{n}{n+m} = \lambda\in (0,1)$, 
(iv) $H_{n,m} \rightarrow H$ almost surely, due to Glivenko-Cantelli's theorem, 
(iv) $a_{n, m} \rightarrow \infty$ $b_{n, m}, c_{n, m}$ and $d_{n,m}$ are finite (due to (i)-(iv)). Thus, with an application of Prokhorov's theorem, we get 
$\bbP_{H_{1}}\left\{T_{\sX, \sY}^{\CvM} > t_{1-\alpha}^{(2)}\right\} \rightarrow 1$ as $\min(n, m) \rightarrow \infty$. 
\par
Let $F_{n}$ and $G_{m}$ be any distribution functions that satisfies $\sqrt{n+m}d_{K}(D_{F_{n}}, D_{G_{m}}) \geq \Delta_{n,m}$.
By triangle inequality, we have 
\begin{equation}
\label{triangle2}
    d_{K}(D_{F_{n}}, D_{G_{m}}) \leq d_{K}(D_{F_{n}}, D_{\sX}) + d_{K}(D_{\sX}, D_{\sY}) + d_{K}(D_{\sY}, D_{G_{m}})
\end{equation}
which implies $T_{\sX, \sY}^{\KS} \geq \Delta_{n,m} - \sqrt{n+m}d_{K}(D_{\sX}, D_{F_{n}}) - \sqrt{n+m}d_{K}(D_{\sY}, D_{G_{m}})$. Therefore, 
\begin{align*}
\label{eq:twosample-unif}
    &\bbP_{F_{n}, G_{m}}\left\{T_{\sX, \sY}^{\KS} > t_{1-\alpha}^{(1)}\right\}\\
    &\geq 
    \bbP_{F_{n}, G_{m}}\left\{
        \Delta_{n,m} - \sqrt{n+m}d_{K}(D_{\sX}, D_{F_{n}}) - \sqrt{n+m}d_{K}(D_{\sY}, D_{G_{m}}) > t_{1-\alpha}^{(1)}
    \right\}\\
    &\geq 
    \bbP_{F_{n}, G_{m}}\left\{
        \sqrt{n+m}d_{K}(D_{\sX}, D_{F_{n}}) + 
        \sqrt{n+m}d_{K}(D_{\sY}, D_{G_{m}}) \leq \Delta_{n,m} - t_{1-\alpha}^{(1)}
    \right\}\\
    & \rightarrow 1 \text{ as } \min(n,m) \rightarrow \infty.
    \numberthis
\end{align*}
The last implication follows from the following facts: (i) $\{\sqrt{n}(D_{\sX}(\bx) - D_{F}(\bx))\}$ and $\{\sqrt{m}(D_{\sY}(\bx) - D_{G}(\bx))\}$ are uniformly bounded in probability under $F_{n}$ and $G_{m}$ respectively in the view of Proposition \ref{proposition:depth}, (ii) $\lambda \in (0,1)$ and $t_{1-\alpha}^{(1)}$ is positive finite. (iii) $\Delta_{n,m} \rightarrow \infty$. An application of Prokhorov's theorem, we get $\bbP_{F_{n}, G_{m}}\left\{T_{\sX, \sY}^{\KS} > t_{1-\alpha}^{(1)}\right\} \rightarrow 1$ as $\min(n, m) \rightarrow \infty$. 
\par
Moreover, by the inequality \eqref{triangle2},
\begin{align*}
&\bbP_{F_{n}, G_{m}}\left\{T_{\sX, \sY}^{\CvM} > s_{1-\alpha}^{(2)}\right\}\\
&\geq \bbP_{F_{n}, G_{m}}\left\{
\int\Delta_{n, m}^{2}dH_{n,m}(\bx) - (n+m)\int d_{K}^{2}(D_{\sX}, D_{F_{n}})dH_{n,m}(\bx)\right. \\
& \qquad \left. - (n+m)\int d_{K}^{2}(D_{\sY}, D_{G_{m}})dH_{n,m}(\bx) > t_{1-\alpha}^{(2)}
\right\}\\
&\geq  \bbP_{F_{n}, G_{m}}\left\{
(n+m)\int d_{K}^{2}(D_{\sX}, D_{F_{n}})dH_{n,m}(\bx) \right.\\
& \qquad \left. + (n+m)\int d_{K}^{2}(D_{\sY}, D_{G_{m}})dH_{n,m}(\bx) \leq \int\Delta_{n, m}^{2}d\bx  - t_{1-\alpha}^{(2)}
\right\}\\
& \rightarrow 1 \text{ as } \min(n,m) \rightarrow \infty.
    \numberthis
\end{align*}
The last implication follows from the same facts that are used in Equation \eqref{eq:twosample-unif}. 
\end{proof}

\begin{proof}[Proof of Theorem \ref{thm:contiguous-one}]
$ $
\par
Observe that the log-likelihood ratio for testing $H_{0}: F = F_{0}$ against $H_{n}$ described in \eqref{eq:alternate}, 
\begin{align*}
\label{eq:loglikelihood}
\sL_{n} &= \sum_{i=1}^{n}\log\frac{(1-\gamma/\sqrt{n})f_{0}(\bx_{i}) + (\gamma/\sqrt{n})h(\bx_{i}) }{f_{0}(\bx_{i})}\\
&=\sum_{i=1}^{n}\log\left\{ 1+ (\gamma/\sqrt{n})\left[\frac{h(\bx_{i})}{f_{0}(\bx_{i})} - 1\right]\right\}\\
&=\frac{\gamma}{\sqrt{n}}\sum_{i=1}^{n}\left\{\frac{h(\bx_{i})}{f_{0}(\bx_{i}} - 1\right\} - \frac{\gamma^{2}}{2n}\sum_{i=1}^{n}\left\{\frac{h(\bx_{i})}{f_{0}(\bx_{i}} - 1\right\}^{2} + \sR_{n}\\
&= \frac{\gamma}{\sqrt{n}}\sum_{i = 1}^{n}\sK_{i} - \frac{\gamma^{2}}{2n}\sum_{i = 1}^{n}\sK_{i}^{2} + \sR_{n} 
    \numberthis
\end{align*}
where $\sK_{i} = \frac{h( {\bx}_{i})}{f_{0}( {\bx}_{i})} - 1$. Since $\sigma^{2} = \E_{F_{0}}\left\{\frac{h(\bx)}{f_{0}(\bx)} - 1\right\}^{2}$ is finite, $\sR_{n} \xrightarrow{\bbP} 0$ as $n \rightarrow \infty$. Contiguity of the sequence $H_{n}$ directly follows from \citet{dhar2014comparison} (see Theorem 6.1) since the first term of Equation \eqref{eq:loglikelihood} is asymptotically normal with mean zero and variance $\gamma^{2}\sigma^{2}$ due to central limit theorem and second term converges in probability to $\gamma^{2}\sigma^{2}/2$ due to weak law of large numbers. Therefore, by Slutsky's theorem, $\sL_{n}$ is asymptotically normal with mean $-\gamma^{2}\sigma^{2}/2$ and variance $\gamma^{2}\sigma^{2}$. An application of Le Cam's first lemma, we can deduce the first part of the theorem. 
\par
Now, to apply Le Cam's third lemma, we need to calculate the covariance between $D_{\sX} - D_{F_{0}}$ and $\sL_{n}$ under null. 
\par
Let $\sT(F)$ be a functional defined for all distributions in a stable class, then the influence function of $\sT$ at $F$ is defined as $\IF(\bz; \sT(F)) = \lim\limits_{\epsilon\rightarrow 0^{+}}\frac{\sT((1-\epsilon)F + \epsilon \delta_{\bz}) - \sT(F)}{\epsilon}$ where $\delta_{\bz}$ is the point mass probability measure at $\bz \in \bbR^{d}$ and $\epsilon \in [0,1]$. 
For any $\bz \in \bbR^{d}$, partition the set closed half-space $\sH$ as $\sH_{\bz} = \{H \in \sH: \bz \in H\}$ and  $\sH_{\overline{\bz}} = \{H \in \sH: \bz \notin H\}$. Thus, corresponding depths are defined as $D_{F}^{\bz}(\bx) = \inf_{\sH_{\bz}}F(H)$ and $D_{F}^{\overline{\bz}}(\bx) = \inf_{\sH_{\overline{\bz}}}F(H)$ respectively. 
Then the influence function of half-space depth becomes $\IF(\bz; D_{F}(\bx)) = -D_{F}(\bx)\textbf{1}\{D_{F}^{(\overline{\bz})}(\bx) < D_{F}^{(\bz)}(\bx)\} + (1-D_{F}(\bx))\textbf{1}\{D_{F}^{(\overline{\bz})}(\bx) \geq D_{F}^{(\bz)}(\bx)\} $ \citep{romanazzi2001influence}, 
where $\textbf{1}\{a \in A\}$ takes value $1$ if $a \in A$ and zero otherwise. If $\bz = \bx$, then, $\sH_{\bz} = \sH$ and $\sH_{\overline{\bz}}$ becomes null-set. Hence $\IF(\bx; D_{F}(\bx)) = 1 - D_{F}(\bx)$. 
Note that, $\IF$ is bounded, and is a step function. For all $\bx$ belonging to optimal half-space $\IF(\bz; D_{F_{0}}(\bx))$ is constant and equal to $1-D_{F}(\bx)$ and for all $\bz$ belonging to non-optimal half-spaces, $\IF(\bz; D_{F_{0}}(\bx))$ is constant and equal to $D_{F}(\bx)$. 
Since by the assumption optimal half-space depth associated to $\bx$ is unique, then for suitable regularity conditions on $F_{0}$ \citep{serfling2009approximation} we can write the von-Mises expansion \citep{romanazzi2001influence} as 
\begin{equation}
\label{eq:depth-expansion}
D_{\sX}(\bx) - D_{F_{0}}(\bx)  = \frac{1}{n}\sum_{i = 1}^{n}\IF( \bX_{i}; D_{F_{0}}( \bx)) + \sE_{n}
\end{equation} 
where $\sE_{n}$ is the remaining term. Due to central limit theorem, it can be shown that $\sqrt{n}\sE_{n} \xrightarrow[]{\bbP} 0$ (see Appendix 2 of \citet{romanazzi2001influence}). 
Therefore, the asymptotic covariance function between $\sqrt{n}(D_{\sX}(\bx) - D_{F_{0}}(\bx))$ and $\sL_{n}$ is 
\begin{align*}
\label{eq:cov-one}
    &\cov_{H_{0}}\left\{ 
        \sqrt{n}(D_{\sX}( {\bx}) - D_{F_{0}}( {\bx})) , \sL_{n}
    \right\}\\
	& = \frac{\gamma}{n} 
	\cov_{H_{0}}\left\{  
	    \sum_{i = 1}^{n} \IF( \bX_{i}; D_{F_{0}}( {\bx})), \sum_{i = 1}^{n}\sK_{i}
	\right\}\\ 
	& \qquad - \frac{\gamma^{2}}{2n^{3/2}} 
	\cov_{H_{0}}\left\{  
	    \sum_{i = 1}^{n} \IF( {\bX_{i}}; D_{F_{0}}( {\bx})), \sum_{i = 1}^{n}\sK_{i}^{2}
	\right\}\\
	& = - \gamma \cov_{H_{0}}\left\{ D_{F_{0}}( {\bx}), \frac{h( {\bx})}{f_{0}( {\bx})}\right\}
	+ \frac{\gamma^{2}}{2\sqrt{n}}\cov_{H_{0}}\left\{D_{F_{0}}(\bx), \left(\frac{h(\bx)}{f_{0}(\bx)} -1 \right)^{2}\right\}\\
	& = -\gamma \int D_{F_{0}}(\bx)h(\bx)d\bx + o(1)
	\numberthis
\end{align*}
The last implication follows from the fact that
$\cov_{H_{0}}\left\{D_{F_{0}}(\bx), \left(\frac{h(\bx)}{f_{0}(\bx)} -1 \right)^{2} \right\}$ is finite that can be shown by applying Cauchy-Schwarz inequality with the fact that 
$\Var_{H_{0}}\{D_{F_{0}(\bx)}\}$ and $\E_{H_{0}}\left\{\frac{h(\bx)}{f_{0}(\bx)} - 1\right\}^{4}$ are finite. 
Thus, by Le Cam's third lemma, under contiguous alternatives the empirical Tukey's half-space depth process i.e. $\sqrt{n}(D_{\sX}( {\bx}) - D_{F_{0}}( {\bx}))$ converges to $\sG_{1}'$ which is Gaussian process with mean $-\gamma\E_{\bx \sim h}\left\{D_{F_{0}}(\bx)\right\}$ and the covariance kernel $F_{0}(H[\bx_{1}]\cap H[\bx_{2}]) - F_{0}(H[\bx_{1}])F_{0}(H[\bx_{2}])$. 
Moreover, $\sqrt{n}(D_{\sX}( {\bx}) - D_{F_{0}}( {\bx}))$ satisfies the tightness condition under contiguous alternatives since it is tight under null. The tightness under null follows from the weak convergence of the empirical Tukey's depth process (see Proposition \ref{proposition:depth}). 
Therefore, by Le Cam's third lemma, 
under the contiguous alternatives alternatives $H_{n}$, 
the asymptotic power of the test based on $T_{\sX, F_{0}}^{\KS}$ and  $T_{\sX, F_{0}}^{\CvM}$ are 
$\bbP_{\gamma}\left\{\sup\limits_{\bx}|\sG_{1}'(\bx)|> s_{1-\alpha}^{(1)}\right\}$ and
$\bbP_{\gamma}\left\{\int\limits|\sG_{1}'(\bx)|^{2}dF_{0}(\bx) > s_{1-\alpha}^{(2)}\right\}$, respectively.
\end{proof}

\begin{proof}[Proof of Theorem \ref{thm:contiguous-one}]
Observe that the likelihood ratio for testing $H_{0}: F = G$ against $H_{n,m}$ described in \eqref{eq:alternate-two}, 
\begin{align*}
\label{eq:likelihood-two}
\sL_{n,m} &= \sum_{i=1}^{n}\sum_{j=1}^{m}\log \frac{f(\bx_{i})\left\{ (1-\gamma/\sqrt{n+m})f(y_{j}) +\gamma h(\by_{j})/\sqrt{n+m}  \right\}}{f(\bx_{i})f(\by_{i})}\\
&= \sum_{j=1}^{m}\log\left\{ 
    1 + \frac{\gamma}{\sqrt{n+m}}\left(\frac{h(y_{j})}{f(y_{j})} -1 \right)
\right\}\\
&= \frac{\gamma}{\sqrt{n+m}}\sum_{j=1}^{m}\left\{
\frac{h(\by_{j})}{f(\by_{j})} - 1
\right\}
- \frac{\gamma^{2}}{2(n+m)}\sum_{j=1}^{m}\left\{
\frac{h(\by_{j})}{f(\by_{j})} - 1
\right\}^{2} + \sR_{n,m}\\
&=\frac{\gamma}{\sqrt{n+m}}\sum_{j=1}^{m}\sK_{j}' - \frac{\gamma^{2}}{2(n+m)}\sum_{j=1}^{m}\sK_{j}'^{2} + \sR_{n,m}\numberthis
\end{align*}
where $\sK_{i}' = \frac{h(\by_{j})}{f(\by_{j})} - 1$. Since $\sigma^{2} = \E\left\{\frac{h(\by)}{f(\by)}-1 \right\}^{2}$ is finite, $\sR_{n,m}\xrightarrow[]{\bbP} 0$ as $\min(n,m) \rightarrow \infty$. Contiguity of the sequence $H_{n,m}$ directly follows from \citet{dhar2014comparison} (see Theorem 6.2) since the first term of Equation \eqref{eq:likelihood-two} is asymptotically normal with mean zero and variance $\gamma^{2}\sigma^{2}(1-\lambda)$ due to central limit theorem and second term converges in probability to 
$-\gamma^{2}\sigma^{2}(1-\lambda)/2$. 
Therefore, by Slutsky's theorem, $\sL_{n,m}$ is asymptotically normal with mean  $-\gamma^{2}\sigma^{2}(1-\lambda)/2$ and variance $\gamma^{2}\sigma^{2}(1-\lambda)$. An application of Le Cam's first lemma, we can deduce the first part of the theorem. 
\par
Now, to apply Le Cam's third lemma, we need to calculate the covariance between $D_{\sX} - D_{\sY}$ and $\sL_{n,m}$ under null. As in the proof of Theorem \ref{thm:contiguous-one}, using von-Mises expansion described in Equation \eqref{eq:depth-expansion} and the fact that $\sX$ and $\sY$ are independent, 
we have 
\begin{align*}
&\cov_{H_{0}}\left\{ \sqrt{n+m}(D_{\sX}(\bu) - D_{\sY}(\bu)), \sL_{n,m} 
\right\}\\
&=\gamma\cov_{H_{0}}\left\{ \sqrt{n+m}(D_{\sX}(\bu) - D_{\sY}(\bu)), \frac{1}{\sqrt{n+m}}\sum_{j=1}^{m}\sK_{j}'
\right\} \\
& \qquad - \frac{\gamma^{2}}{2}\cov_{H_{0}}\left\{ \sqrt{n+m} (D_{\sX}(\bu) - D_{\sY}(\bu)), \frac{1}{(n+m)}\sum_{j=1}^{m}\sK_{j}'^{2}
\right\}
:= \gamma\sA_{1} - \frac{\gamma^{2}}{2}\sA_{2}
    \numberthis
\end{align*}
where 
\begin{align*}
\sA_{1} &=
\cov_{H_{0}}\left\{  \sqrt{n+m}(D_{\sX}(\bu) - D_{\sY}(\bu)), \frac{1}{\sqrt{n+m}}\sum_{j=1}^{m}\sK_{j}'
\right\}\\
&=-\cov_{H_{0}}\left\{ 
    (1-\lambda)^{-1/2}\sqrt{m}(D_{\sY}(\bu) - D_{F}(\bu)), \frac{1}{\sqrt{n+m}}\sum_{j=1}^{m}\sK_{j}'
\right\}\\
&= -(1-\lambda)^{-1/2}\cov_{H_{0}}\left\{
    \frac{1}{\sqrt{m}}\sum_{j=1}^{m}\IF(\by_{j}; D_{F}(\bu)), \frac{1}{\sqrt{n+m}}\sum_{j=1}^{m}\sK_{j}'
\right\}\\
&=\sqrt{\frac{\lambda}{1-\lambda}}\cov_{H_{0}}\left\{
  D_{F}(\bu), \frac{h(\bu)}{f(\bu)}  
\right\} = \sqrt{\frac{\lambda}{1-\lambda}}\int D_{F}(\bu)h(\bu)d\bu
    \numberthis
\end{align*}
and similar to the second term of Equation \eqref{eq:cov-one}, under the condition $\E\left\{\frac{h(\bx)}{f(\bx)}-1 \right\}^{4} < \infty$, by applying Cauchy-Schwarz inequality, $\sA_{2} = o(1)$ as $\min(n, m) \rightarrow \infty$. 
Thus, the covariance between $\sqrt{n+m}(D_{\sX}(\bu) - D_{\sY}(\bu))$ and $\sL_{n,m}$ is $\gamma\sqrt{\lambda/(1-\lambda)}\E_{\bu \sim h}\left\{ D_{F}(\bu) \right\}$. 
\par
Thus, by Le Cam's third lemma, under contiguous alternatives the empirical Tukey's half-space depth process i.e. $\sqrt{n+m}(D_{\sX}( {\bu}) - D_{\sY}( {\bu}))$ converges to $\sG_{2}'$ which is Gaussian process with mean $\gamma\sqrt{\lambda/(1-\lambda)}\E_{\bu \sim h}\left\{D_{F}(\bu)\right\}$ and the covariance kernel $\{F(H[\bu_{1}]\cap H[\bu_{2}]) - F(H[\bu_{1}])F_{0}(H[\bu_{2}])\}/\lambda(1-\lambda)$. 
Moreover, $\sqrt{n+m}(D_{\sX}( {\bu}) - D_{\sY}( {\bu}))$ satisfies the tightness condition under contiguous alternatives since it is tight under null. The tightness under null follows from the weak convergence of the empirical Tukey's depth process (under the independence of $\sX$ and $\sY$, for $\lambda\in (0,1)$, see Proposition \ref{proposition:depth}). 
Therefore, by Le Cam's third lemma, 
under the contiguous alternatives alternatives $H_{n}$, 
the asymptotic power of the test based on $T_{\sX, \sY}^{\KS}$ and  $T_{\sX, \sY}^{\CvM}$ are 
$\bbP_{\gamma}\left\{\sup\limits_{\bu}|\sG_{2}'(\bu)|> t_{1-\alpha}^{(1)}\right\}$ and
$\bbP_{\gamma}\left\{\int\limits|\sG_{2}'(\bu)|^{2}dH_{n,m}(\bu) > t_{1-\alpha}^{(2)}\right\}$, respectively.
\end{proof}
\end{supplementary}
\end{document}